\documentclass[11pt]{article}

	\usepackage{a4,geometry}
\usepackage{graphicx}
\usepackage{amsmath,amssymb,amsthm,mathtools}
\usepackage{paralist}
\usepackage{bm}
\usepackage{bbm}
\usepackage{xspace}
\usepackage{url}
\usepackage{fullpage, prettyref}
\usepackage{boxedminipage}
\usepackage{wrapfig}
\usepackage{ifthen}
\usepackage{color}
\usepackage{xcolor}
\usepackage{framed}
\usepackage{algorithm}
\usepackage[linktocpage=true,pagebackref,letterpaper=true,colorlinks=true,pdfpagemode=none,urlcolor=blue,linkcolor=blue,citecolor=violet,pdfstartview=FitH]{hyperref}
\usepackage[nameinlink]{cleveref}
\usepackage{fullpage}
\usepackage{esvect}
\usepackage{mdframed}
\usepackage{times}
\usepackage{thmtools}
\usepackage{thm-restate}
\usepackage{algpseudocode}

\usepackage{enumitem}
\newtheorem{theorem}{Theorem}[section]

\newtheorem{lemma}[theorem]{Lemma}
\newtheorem{claim}[theorem]{Claim}
\newtheorem{corollary}[theorem]{Corollary}
\newtheorem{definition}[theorem]{Definition}

\newtheorem{fact}[theorem]{Fact}

\newtheorem{assumption}[theorem]{Assumption}
\newtheorem{observation}[theorem]{Observation}
\newtheorem{remark}[theorem]{Remark}

\newcommand{\ignore}[1]{}

\newcommand{\cD}{\mathcal{D}}
\newcommand{\cE}{{\cal E}}

\newcommand{\cM}{{\cal M}}

\newcommand{\cS}{\mathcal{S}}

\newcommand{\cU}{{\cal U}}

\newcommand{\HH}{\mathbb H}

\newcommand{\Z}{{\mathbb Z}}

\newcommand{\eps}{\varepsilon}

\newcommand{\poly}{\mathrm{poly}}
\newcommand{\polylog}{\mathrm{polylog}}

\newcommand{\calE}{{\cal E}}

\newcommand{\be}{\mathbf{e}}

\newcommand{\bs}{\mathbf{s}}
\newcommand{\bt}{\mathbf{t}}

\newcommand{\bw}{\mathbf{w}}
\newcommand{\bx}{\mathbf{x}}
\newcommand{\by}{\mathbf{y}}
\newcommand{\bz}{\mathbf{z}}

\newcommand{\bA}{\boldsymbol{A}}
\newcommand{\bB}{\boldsymbol{B}}

\newcommand{\bH}{\boldsymbol{H}}
\newcommand{\bT}{\boldsymbol{T}}

\newcommand{\bL}{\boldsymbol{L}}
\newcommand{\bR}{\boldsymbol{R}}

\newcommand{\bX}{\boldsymbol{X}}
\newcommand{\bY}{\boldsymbol{Y}}

\newcommand{\RR}{\mathbb{R}}

\newcommand{\norm}[1]{\left\lVert #1 \right\rVert}
\newcommand{\ceil}[1]{\lceil#1\rceil}
\newcommand{\Exp}{\EX}
\newcommand{\floor}[1]{\lfloor#1\rfloor}

\newcommand{\EX}{\mathbb{E}}

\newcommand{\supp}{\mathrm{supp}}
\newcommand{\eqdef}{:=}

\newcommand{\otilde}{\widetilde{O}}

\newcommand{\hyp}[1]{\{0,1\}^{#1}}
\newcommand{\dtal}[1]{\mathrm{Tal}(#1)}

\newcommand{\updist}[2]{\cU_{#2}(#1)}
\newcommand{\downdist}[2]{\cD_{#2}(#1)}

\newcommand{\upshift}[2]{\cU\cS_{#2}(#1)}
\newcommand{\downshift}[2]{\cD\cS_{#2}(#1)}

\newcommand{\Sec}[1]{\hyperref[sec:#1]{\S\ref*{sec:#1}}} 
\newcommand{\Eqn}[1]{\hyperref[eq:#1]{(\ref*{eq:#1})}} 
\newcommand{\Fig}[1]{\hyperref[fig:#1]{Fig.\,\ref*{fig:#1}}} 
\newcommand{\Tab}[1]{\hyperref[tab:#1]{Tab.\,\ref*{tab:#1}}} 
\newcommand{\Thm}[1]{\hyperref[thm:#1]{Theorem\,\ref*{thm:#1}}} 
\newcommand{\Fact}[1]{\hyperref[fact:#1]{Fact\,\ref*{fact:#1}}} 
\newcommand{\Lem}[1]{\hyperref[lem:#1]{Lemma\,\ref*{lem:#1}}} 
\newcommand{\Prop}[1]{\hyperref[prop:#1]{Prop.~\ref*{prop:#1}}} 
\newcommand{\Cor}[1]{\hyperref[cor:#1]{Corollary~\ref*{cor:#1}}} 
\newcommand{\Conj}[1]{\hyperref[conj:#1]{Conjecture~\ref*{conj:#1}}} 
\newcommand{\Def}[1]{\hyperref[def:#1]{Definition~\ref*{def:#1}}} 
\newcommand{\Alg}[1]{\hyperref[alg:#1]{Alg.~\ref*{alg:#1}}} 
\newcommand{\Ex}[1]{\hyperref[ex:#1]{Ex.~\ref*{ex:#1}}} 
\newcommand{\Clm}[1]{\hyperref[clm:#1]{Claim~\ref*{clm:#1}}} 
\newcommand{\Step}[1]{\hyperref[step:#1]{Step~\ref*{step:#1}}} 
\newcommand{\Obs}[1]{\hyperref[observation:#1]{Observation~\ref*{observation:#1}}} 

\renewcommand{\deg}{\mathsf{deg}}

\newcommand{\cper}{C_{per}}
\def\mzb{\mathsf{mzb}}
\newcommand{\pathtester}{\texttt{path-tester}}
\newcommand{\linetester}{\texttt{line-tester}}
\newcommand{\flow}{\mathsf{flow}}

\begin{document}

\title{A $d^{1/2+o(1)}$ Monotonicity Tester for Boolean Functions on $d$-Dimensional Hypergrids\footnote{An initial version of this manuscript appeared at FOCS 2023.}}

\author{Hadley Black\\
University of California, San Diego\\
{\tt hablack@ucsd.edu}
\and
Deeparnab  Chakrabarty\thanks{Supported by NSF-CAREER award CCF-2041920 and CCF-2402571} \\
Dartmouth\\
{\tt deeparnab@dartmouth.edu}
\and
C. Seshadhri\thanks{Supported by NSF DMS-2023495, CCF-1740850, 1839317, 1813165, 1908384, 1909790, 2402572} \\
University of California, Santa Cruz\\
{\tt sesh@ucsc.edu}
}

\maketitle
\date{}

\begin{abstract}
Monotonicity testing of Boolean functions on the hypergrid, $f:[n]^d \to \{0,1\}$, is a classic topic in property testing. Determining the non-adaptive complexity of this problem is an important open question. For arbitrary $n$, [Black-Chakrabarty-Seshadhri, SODA 2020] describe a tester with query complexity $\widetilde{O}(\varepsilon^{-4/3}d^{5/6})$. This complexity is independent of $n$, but has a suboptimal dependence on $d$. Recently, [Braverman-Khot-Kindler-Minzer, ITCS 2023] and [Black-Chakrabarty-Seshadhri, STOC 2023] describe $\widetilde{O}(\varepsilon^{-2} n^3\sqrt{d})$ and $\widetilde{O}(\varepsilon^{-2} n\sqrt{d})$-query testers, respectively. These testers have an almost optimal dependence on $d$, but a suboptimal polynomial dependence on $n$. \smallskip

In this paper, we describe a non-adaptive, one-sided monotonicity tester with query complexity $O(\varepsilon^{-2} d^{1/2 + o(1)})$, \emph{independent} of $n$. Up to the $d^{o(1)}$-factors, our result resolves the non-adaptive complexity of monotonicity testing for Boolean functions on hypergrids. The independence of $n$ yields a non-adaptive, one-sided $O(\varepsilon^{-2} d^{1/2 + o(1)})$-query monotonicity tester for Boolean functions $f:\mathbb{R}^d \to \{0,1\}$ associated with an arbitrary product measure.
\end{abstract}
\thispagestyle{empty}

\newpage
\tableofcontents
\thispagestyle{empty}
\setcounter{page}{0}
\newpage
\section{Introduction} 

Since its introduction more than two decades ago, the problem of monotonicity testing has attracted an immense amount of attention.
In this paper, we focus on the question of monotonicity testing of Boolean functions $f:[n]^d\to \{0,1\}$ over the $d$-dimensional hypergrid.
The problem was introduced in the seminal paper of Goldreich, Goldwasser, Lehman, Ron, and Samorodnitsky~\cite{GGLRS00} and early results were achieved by Raskhodnikova~\cite{Ras99} and Dodis, Goldreich, Lehman, Raskhodnikova, Ron, and Samorodnitsky ~\cite{DGLRRS99}. (See~\Sec{related} for more details.)
%

Each element
$\bx \in [n]^d$ is represented as a $d$-dimensional vector with $\bx_i \in [n]$ denoting the $i$th coordinate.
The partial order of the hypergrid is defined as: $\bx \preceq \by$ iff $\bx_i \leq \by_i$ for all $i\in [d]$. When $n=2$, the hypergrid $[n]^d$ is isomorphic to the
hypercube $\{0,1\}^d$. 
A Boolean hypergrid function $f:[n]^d \to \{0,1\}$ is monotone if $f(\bx) \leq f(\by)$ whenever $\bx \preceq \by$.
The distance between two functions $f$ and $g$, denoted $\Delta(f,g)$,
is the fraction of points where they differ. 
A function $f:[n]^d \to \{0,1\}$ is called $\eps$-far
from monotone if $\Delta(f,g) \geq \eps$ for all monotone functions $g:[n]^d\to \{0,1\}$.
Given a proximity parameter $\eps$ and query access to a function, a \emph{monotonicity tester} is a randomized algorithm which accepts
a monotone function and rejects a function that is $\eps$-far
from monotone, each with probability $\geq 2/3$.
If the tester accepts monotone functions with probability $1$, it is said to have one-sided error or simply called \emph{one-sided}.
If the tester decides its queries without seeing any responses, it is called \emph{non-adaptive}.

An outstanding open question in property testing is to determine the optimal non-adaptive query
complexity of monotonicity testing for Boolean hypergrid functions. Here we mention the current best bounds and refer the reader to \Sec{related} for a more extensive background. 
Black, Chakrabarty, and Seshadhri~\cite{BlackCS18,BlackCS20} give a $\otilde(\eps^{-4/3}d^{5/6})$-query tester.
Note that the query complexity is independent of $n$. 
Building on seminal work of Khot, Minzer, and Safra~\cite{KMS15}, Braverman, Khot, Kindler, and Minzer~\cite{BrKh+23} 
and Black, Chakrabarty, and Seshadhri~\cite{BlChSe23} recently give $\otilde(\eps^{-2}n^3\sqrt{d})$ and $\otilde(\eps^{-2}n\sqrt{d})$ testers, respectively.
Chen, Waingarten, and Xie~\cite{Chen17} give an $\widetilde{\Omega}(\sqrt{d})$ lower bound for non-adaptive Boolean monotonicity testing on hypercubes ($n=2$). Hence,
these last bounds are nearly optimal in $d$, but are sub-optimal in $n$. Can one achieve the optimal $\sqrt{d}$ dependence while being independent of $n$?
We answer in the affirmative, giving a non-adaptive, one-sided monotonicity tester for Boolean functions over hypergrids with almost optimal query complexity.

\begin{mdframed}[backgroundcolor=gray!20,topline=false,bottomline=false,leftline=false,rightline=false,innertopmargin=-5pt] 
\begin{theorem}\label{thm:mono-testing}~
    Consider Boolean functions over the hypergrid, $f:[n]^d \to \{0,1\}$.
    There is a one-sided, non-adaptive tester for monotonicity that makes $\eps^{-2} d^{1/2 + O(1/\log\log d)}$ queries. 
\end{theorem}
\end{mdframed}

Query complexities independent of $n$ allow for monotonicity testing
over continuous spaces. Let $\mu = \prod_{i=1}^d \mu_i$
be an associated product Lebesgue measure over $\RR^d$. A function $f:\RR^d\to \{0,1\}$
is measurable if the set $f^{-1}(1)$ is Lebesgue-measurable with respect to $\mu$.
%
The $\mu$-distance of $f$ to monotonicity is defined
as $\inf_{g \in \cM} \mu(\Delta(f,g))$, where $\cM$ is the family of measurable monotone functions and $\Delta$ is the symmetric difference operator.
(Refer to Sec. 6 of~\cite{BlackCS20} for more details.)
Domain reduction results~\cite{BlackCS20,HY22}
show that monotonicity testing over general hypergrids and continuous (measurable) spaces can be reduced to
the case where $n = \poly(\eps^{-1}d)$ via sampling. 
A direct consequence of \Thm{mono-testing}
is the following theorem for continuous monotonicity testing.

\begin{mdframed}[backgroundcolor=gray!20,topline=false,bottomline=false,leftline=false,rightline=false,innertopmargin=-5pt] 
\begin{theorem}\label{thm:cont-testing}~
    Consider measurable Boolean functions $f:\RR^d \to \{0,1\}$, with a product measure $\mu$.
    There is a one-sided, non-adaptive tester for monotonicity that makes $\eps^{-2} d^{1/2 + O(1/\log\log d)}$ queries. 
\end{theorem}
\end{mdframed}
All $o(d)$ non-adaptive, one-sided monotonicity testers are \emph{path testers} (also called pair testers) that check for violations among comparable points at a distance from each other, rejecting if they form a violation. 
Consider the {\em fully augmented} directed hypergrid graph defined as follows. Its vertices are $[n]^d$
and its edges connect all pairs $\bx \prec \by$ that differ in exactly one coordinate.
A path tester picks a random point $\bx$ in $[n]^d$, performs a random walk in this directed
graph to get another point $\by \succ \bx$, and rejects if $f(\bx) > f(\by)$. 
The whole game is to lower bound the probability that $f(\bx) > f(\by)$ when $f$ is $\eps$-far from being monotone.
Unlike random walks on undirected graphs, these directed random walks are ill-behaved. In particular, one cannot walk for ``too long''
and the length of the walk has to be carefully chosen. The approach to analyzing such path testers has two distinct parts.

\begin{asparaitem}
	\item {\em Directed Isoperimetry.} 
	A Boolean isoperimetric theorem relates the volume of a subset of the hypergrid, in our case the preimage $f^{-1}(1)$, 
	to the edge and vertex expansion properties of this set in the graph. A directed analogue
	replaces the volume with the distance to monotonicity, and
	deals with directed expansion properties. 
	
	\item {\em Random walk analysis.} The second part is to use the directed isoperimetric theorem to lower bound the success probability
	of the path tester. The analogy is: if the (directed) expansion of a set is large, then the probability of a directed random walk starting from a $1$ and ending at a $0$ is also large. This analysis proceeds via special combinatorial substructures in the graph of violations.
\end{asparaitem}
\smallskip

\noindent
The seminal result of Khot, Minzer, and Safra~\cite{KMS15} (henceforth KMS) gave near optimal analyses for both parts,
for the {\em hypercube} domain. For the first part, they prove a directed, robust
version of the Talagrand isoperimetric theorem. KMS use this directed isoperimetric theorem to construct ``good subgraphs'' of the hypercube comprised of violated edges. 
For the second part mentioned above, KMS relate the success probability of the directed random walk
to properties of this subgraph. 
Coming to hypergrids, one needs to generalize both parts of the analysis, and this offers many challenges. 
For the first part, Black, Chakrabarty, and Seshadhri~\cite{BlChSe23} generalize the directed Talagrand inequality to
the hypergrid domain. 
%
Unfortunately, even with this stronger directed Talagrand isoperimetric bound for hypergrids, the generalization
of the KMS random walk analysis only yields a $1/(n\sqrt{d})$ lower bound on the success probability. 

\begin{quote}
	\emph{The main technical contribution of this paper is a new random walk algorithm and analysis whose success 
		probability is at least $\eps^2 d^{-(1/2 + o(1))}$.}
\end{quote}

\subsection{Algorithm Description} \label{sec:path}

Our algorithm performs directed random walks as all previous monotonicity testers do, 
but augments these with coordinated walks. It starts at a random $\bx$ and performs an ``up-walk'' of a certain (random) length on the fully augmented directed hypergrid to reach a point $\by$.
The algorithm then ``walks down'' in a {\em coordinated} fashion from both $\bx$ and $\by$ to get to points $\bw = \bx - \bs$ and $\bz = \by - \bs$ where $\bs$ is some random vector.
The main contribution of this paper is to show that the combinatorial properties implied by directed isoperimetry theorems of~\cite{BlChSe23} can be used to analyze these coordinated tests. 
We give a formal description of the algorithm.

%
%

Without loss of generality\footnote{See Theorem A.1 of \cite{BlackCS18}. Note this assumption is not crucial, but we choose to use it for the sake of a cleaner presentation.}, we assume that $n$ is a power of $2$. We use $x \in_R S$ to denote choosing a uniform random element $x$ from the set $S$. We use the notation $[n] := \{1,2,\ldots,n\}$. Abusing notation, we define intervals in $\Z_n$ by wrapping around. So, if $1 \leq i \leq n < j$, then the interval $[i,j]$ in $\Z_n$ is the set $[i,n] \cup [1,j \pmod n]$. 
The directed (lazy) random walk distribution in $[n]^d$ that we consider is defined as follows.
The distribution induced by this directed walk has multiple equivalent formulations, which are discussed in \Sec{equivalent_dists}.

\begin{definition} [Hypergrid Walk Distribution] \label{def:walkdist} For a point $\bx \in [n]^d$ and walk length $\tau$, the distribution $\mathcal{U}_{\tau}(\bx)$ over $\by \in [n]^d$ reached by an {\em upward lazy random walk from $\bx$ of $\tau$-steps} is defined as follows. 
	\begin{enumerate}[noitemsep]
		\item Pick a uniform random subset $R\subseteq [d]$ of $\tau$ coordinates.
		\item For each $r \in R$:
		\begin{asparaenum}
			\item Choose $q_r \in_R \{1,2,\ldots,\log n\}$ uniformly at random.
			\item Choose a uniform random interval $I_r$ in $\mathbb{Z}_n$ of size $2^{q_r}$ such that $\bx_r \in I_r$.
			\item Choose a uniform random $c_r \in_R I_r \setminus \{\bx_r\}$.
		\end{asparaenum}
		\item Generate $\by$ as follows. For every $r \in [d]$, if $r \in R$ \emph{and} $c_r > \bx_r$, set $\by_r = c_r$. Else, set $\by_r = \bx_r$.
	\end{enumerate}
	Analogously, let $\downdist{\bx}{\tau}$ be the distribution defined precisely as above, but the $>$-sign is replaced by the $<$-sign in
	step 3. This is the distribution of the endpoint of a {\em downward lazy random walk from $\bx$ of $\tau$-steps}.
\end{definition}
\noindent
As mentioned earlier, a crucial step of our algorithm involves performing the exact same random walk, but originating from two different points. 
We use the notion of {\em shifts}.

\begin{definition} [Shift Distributions] \label{def:shiftdist} The up-shift distribution from $\bx$, denoted $\upshift{\bx}{\tau}$ is the distribution of $\bx'-\bx$, where $\bx' \sim \updist{\bx}{\tau}$. The down-shift distribution from $\bx$, denoted $\downshift{\bx}{\tau}$ is the distribution of $\bx-\bx'$, where $\bx' \sim \downdist{\bx}{\tau}$. \end{definition}
\noindent
Using \Def{walkdist} and \Def{shiftdist}, our tester is defined in \Alg{alg}.


\begin{algorithm}[ht!]
	\caption{Monotonicity tester for Boolean functions on $[n]^d$}\label{alg:alg}
	\textbf{Input:} A Boolean function $f \colon [n]^d \to \{0,1\}$
	\begin{asparaenum}
		\item Choose $p \in_R \{0,1,2,\ldots,\floor{\log d}\}$ uniformly at random and set $\tau := 2^{p}$.
		\item Run the \textit{upward path test} with walk length $\ell = \tau-1$ and $\ell = \tau$:
		\begin{enumerate}[noitemsep]
			\item Choose $\bx \in_R [n]^d$ and sample $\by$ from $\cU_{\ell}(\bx)$. 
			\item If $f(\bx) > f(\by)$, then reject. 
		\end{enumerate}
		\item Run the \textit{downward path test} with walk length $\ell = \tau-1$ and $\ell = \tau$:
		\begin{enumerate}[noitemsep]
			\item Choose $\by \in_R [n]^d$ and sample $\bx$ from $\downdist{\by}{\ell}$. 
			\item If $f(\bx) > f(\by)$, then reject. 
		\end{enumerate}
		\item Run the \textit{upward path + downward shift test} with walk length $\ell = \tau-1$ and $\ell = \tau$:
		\begin{enumerate}[noitemsep]
			\item Choose $\bx \in_R [n]^d$, sample $\by$ from $\updist{\bx}{\ell}$, and sample $\bs$ from $\downshift{\bx}{\tau-1}$.
			\item If $f(\bx - \bs) > f(\by - \bs)$, then reject.
		\end{enumerate}
		\item Run the \textit{downward path + upward shift test} with walk length $\ell = \tau-1$ and $\ell = \tau$:
		\begin{enumerate}[noitemsep]
			\item Choose $\by \in_R [n]^d$, sample $\bx$ from $\downdist{\by}{\ell}$, and sample $\bs$ from $\upshift{\by}{\tau-1}$.
			\item If $f(\bx + \bs) > f(\by + \bs)$, then reject.
		\end{enumerate}
	\end{asparaenum} 
\end{algorithm}

\noindent

\begin{remark}\label{rem:wlog}
	Given a function $f:[n]^d \to \{0,1\}$, consider the doubly-flipped function $g:[n]^d \to \{0,1\}$ defined as 
	$g(\bx) := 1 - f(\bar{\bx})$ where $\bar{\bx}_i := n - \bx_i + 1$. That is, we swap all the zeros and ones in $f$, 
	and then reverse the hypergrid (the all $1$'s point becomes the all $n$'s point and vice-versa).
	The distance to monotonicity of both $f$ and $g$ are the same: a pair $(\bx,\by)$ is violating in $f$
	if and only if $(\bar{\bx}, \bar{\by})$ is violating in $g$.
	In \Alg{alg}, Step 2 on $f$ is the same as Step 3 on $g$, and Step 4 on $f$ is the same as Step 5 on $g$.
	In our analysis, we will construct a violation subgraph between vertex sets $\bX$ and $\bY$. 
	Points in $\bX$ are $1$-valued and points in $\bY$ are $0$-valued.
	If $|\bX| \leq |\bY|$, then the steps 2, 3, and 4 suffice for the analysis. If $|\bY| \leq |\bX|$, then (by the same analysis)
	we run steps 2,3, and 4 on the function $g$. This is equivalent to running steps 2, 3, and 5 on the function $f$.
	So, the tester covers both situations, and we can assume without loss of generality that $|\bX| \leq |\bY|$. This discussion happens
	in \Cref{sec:choice-tau}.
\end{remark}
	
	

\noindent
Our main result is the following lower bound on the rejection probability of \Alg{alg}.


\begin{mdframed}
	[backgroundcolor=gray!20,topline=false,bottomline=false,leftline=false,rightline=false,innertopmargin=-5pt] 
	\begin{restatable}[Main Theorem]{theorem}{maintheorem}
		\label{thm:main} Let $n \leq \poly(d)$ and $\eps \geq d^{-1/2}$. If $f \colon [n]^d \to \{0,1\}$ is $\eps$-far from being monotone, then \Alg{alg} rejects $f$ with probability at least $\eps^2 \cdot d^{-(1/2 + O((\log \log nd)^{-1}) ) }$.
	\end{restatable}
\end{mdframed}
\noindent
\Thm{main} is proved in \Sec{anal}. Using \Thm{main}, our main testing results \Thm{mono-testing} and \Thm{cont-testing} follow easily from prior techniques, and so we defer their proofs to \Cref{sec:proofs_of_main_theorems}.

\subsection{Analysis Overview}\label{sec:ourideas}
We give a leisurely  overview of the main ideas that go behind proving~\Cref{thm:main}. 
We begin with a high-level overview of the KMS analysis for the hypercube case, sketch certain challenges that hypergrids pose, and then 
discuss our ideas that led us to the ``shifted walks'' view. 

\paragraph{The KMS random walk analysis on $\{0,1\}^d$ in a nutshell.} 
For simplicity, let's assume $\eps$ is a small constant so that we ignore the dependence on $\eps$.
As mentioned earlier, KMS prove a directed, robust
version of the Talagrand isoperimetric theorem. Using this, they extract a large ``good subgraph" of violations. 
A violation subgraph $G = (\bX, \bY,E)$ is a bipartite graph where $\forall \bx \in \bX, f(\bx) = 1$, $\forall \by \in \bY$, $f(\by) = 0$, and
all edges in $E$ are hypercube edges. A good subgraph is a violation subgraph that satisfies
certain lower bounds on the total number of edges and has an approximate regularity
property. The specifics are a bit involved (Definition 6.1-6.3, in~\cite{KMS15}), but it is instructive to consider the 
simplest good subgraph: a {\em matching} between $\bX$ and $\bY$ where $|\bX| = |\bY| = \Omega(2^d)$. 
%
%

When the good subgraph is a matching, KMS show that a random walk of length
$\tau = \widetilde{\Theta}(\sqrt{d})$ succeeds in finding a violation with $\widetilde{\Omega}(d^{-1/2})$ probability.
A key insight in the analysis is the notion of \emph{$\tau$-persistence}: a vertex $\bx$ is $\tau$-persistent
if a $\tau$-length directed random walk leads to a point $\bz$ where $f(\bx) = f(\bz)$ with constant probability. 
A simple argument shows that there are $o(2^d)$ non-persistent vertices (for walk length $< d^{1/2}$).
We can remove all such vertices from $\bX$ and $\bY$.
%
We get a matching between subsets $\bX'$ and $\bY'$ where all points in $\bX'$ and $\bY'$
are $(\tau-1)$-persistent, and $|\bX'| = |\bY'|$ is still $\Omega(2^d)$.

With $\Omega(1)$ probability, the tester starts from $\bx \in \bX'$. Note that $f(\bx) = 1$.
Let the matched partner of $\bx$ be $\by$. 
Let $i$ be the dimension of the violated edge $(\bx,\by)$. 
With probability roughly $\tau/d = \widetilde{\Omega}(d^{-1/2})$, the directed
walk will cross the $i$th dimension. Let us condition on this event. 
We can interpret the random walk as traversing the edge $(\bx,\by)$,
and then taking a $(\tau-1)$-length directed walk from $\by$ to reach
the destination $\by'$. (Note that we do not care about the specific order of edges traversed by the random walk.
We only care about the value at the destination.)
Since $\by$ is $(\tau-1)$-persistent 
with $\Omega(1)$ probability the final destination $\by'$ will satisfy $f(\by') = f(\by) = 0$. 
Putting it all together, the tester succeeds with probability $\widetilde{\Omega}(d^{-1/2})$.

\paragraph{The challenge in hypergrids.} \label{sec:challenge}
As mentioned earlier,~\cite{BlChSe23} proves an isoperimetric theorem for hypergrids generalizing the one in~\cite{KMS15}.
Using similar techniques to the hypercube case, one can construct ``good subgraphs'' of the fully augmented hypergrid.
The definition is involved (Theorem 7.8 in~\cite{BlChSe23}), but the simplest case
is again a violation matching of $(\bX,\bY,E)$ of size
$|\bX| = |\bY| = \Omega(n^d)$. Note that the matched pairs $(\bx,\by)$ are axis-aligned, that is,
differ in exactly one coordinate $i$. But $\by_i - \bx_i$ is an integer in $\{1,2\ldots, n-1\}$.

In the hypergrid, the directed random
walk must necessarily perform ``jumps". At each step, the walk changes a chosen 
coordinate to a random larger value. 
One can generalize the hypercube persistence arguments to show that with constant probability,
a $\tau = \widetilde{\Theta}(\sqrt{d})$-step random walk will result in both endpoints having the same value.
And so, like before, we can remove all ``non-persistent'' points to end up with an $\Omega(n^d)$ violation matching $(\bX',\bY')$ where all vertices are $\tau$-persistent.

The tester picks $\bx \in \bX'$ with $\Omega(1)$ probability. Let $\by$ be its matched partner, which differs in the $i$th coordinate.
If the number of steps is $\tau$, then with $\tau/d \approx \widetilde{\Omega}(d^{-1/2})$ probability, the walk will choose to move along the $i$th coordinate. Conditioned on this event, we would like to relate the random walk to a persistent walk from $\by$.
However, there is only a $1/n$ chance that the {\em length} jumped along that coordinate will be the jump $\by_i - \bx_i$.
One loses an extra $n$ factor in the success probability, and indeed, this is the high-level
analysis of the $\otilde_\eps(n\sqrt{d})$-tester from~\cite{BlChSe23} (at least for the case of the matching).

How does one get rid of this dependence on $n$? There is no simple way around this impasse.
If $\by_i - \bx_i$ is, say $\Theta(n)$, we cannot relate the walk from $\bx$ to a (persistent) walk from
$\by$ without losing this $n$ factor.
If one desires to be free of the parameter $n$, then one needs to consider the {\em internal points} in the segment $(\bx,\by)$. 
But all internal points could be non-persistent. Even though most internal points $\bz$ in the segment $(\bx,\by)$ may be $0$-valued, a
$(\tau-1)$-step walk from $\bz$ could lead to $1$-valued points.
So the final pair will not be a violation. One may think that since the matching size was large ($\approx n^d$), the ``interior'' (the union of the interiors of the matching segments) would be large,
and most of the internal nodes would be persistent. Unfortunately, that may not be the case, and the following is an illustrative example.
%
%
%
%
%
%
%
%
%
We define a Boolean hypergrid function $f$ and
an associated violation matching iteratively. Let $n \leq d/\ln d$. Start with all function values undefined.
If $\bx_1 = 1$, set $f(\bx) = 1$. If $\bx_1 = n$, set $f(\bx) = 0$. Take the natural violation
matching between these points. For every undefined point $\bx$: if $\bx_2 = 1$, set $f(\bx) = 1$
and if $\bx_2 = n$, set $f(\bx) = 0$. Iterating over all coordinates, we define the function
at all points ``on the surface". In the ``interior", where $\forall i$, $\bx_i \notin \{1,n\}$,
we set $f$ arbitrarily. The interior has size $n^d \cdot (1-2/n)^d \approx n^d \exp(-2\ln d) \leq n^d/d^2$.
This is a tiny fraction
of the domain, while the matching has size $\Omega(n^d)$.
Hence, it is possible to have a large violation matching such that the union of (strict) interiors is vanishingly
small. 

 \paragraph{Mostly-zero-below Points and Red Edges.}
As mentioned above, we begin with the basic case of a violation matching $G = (\bX, \bY, E)$ of size $\Omega(n^d)$
in the fully augmented hypergrid. The general case will be discussed at the end of this section.
We set $\tau = \widetilde{\Theta}(\sqrt{d})$. As argued above, using the persistence and Markov inequality arguments, we can
assume that all points in $\bX \cup \bY$ are $(\tau-1)$-persistent. Recall that our algorithm 
performs both up-walks and down-shifted walks. In particular, it compares a pair $(\bx, \bw)$ where $\bw$ is $\tau$-steps ``above'' $\bx$, 
and then also compares $(\bx-\bs, \bw-\bs)$ where $\bs$ has $(\tau-1)$ non-zero coordinates.
The following is a key definition: we call a point $\bw$ {\em mostly-zero-below} for length $\ell$, or simply $\ell$-$\mzb$, if an $\ell$-length down-walk from $\bw$ leads to a zero with $\ge 0.9$ probability (\Cref{def:predict}). 
Suppose an up-walk of length $\tau$ from a point $\bx\in \bX$ reaches a $(\tau-1)$-$\mzb$ point $\bw$. 
Then, a random shift $(\bx-\bs, \bw-\bs)$ has a constant probability of being a violation.
The reason is (i) $\Pr[f(\bx-\bs) = f(\bx) = 1] \geq 0.9$ because $\bx$ is $(\tau-1)$-persistent, and (ii) $\Pr[f(\bw-\bs) = 0] \geq 0.9$ because $\bw$ is $(\tau-1)$-mostly-zero-below. By a union bound, the tester will find a violation with constant probability
(conditioned on discovering the pair $(\bx,\bw)$). 

To formalize this analysis, we define an edge $(\bx,\by)$ of our large matching to be {\em red} if it satisfies the following condition: 
for a constant fraction of the points $\bz$ in the segment $(\bx,\by)$, a $(\tau-1)$-length up-walk ends at a $(\tau-1)$-$\mzb$ point with constant probability (\Cref{def:red}). If there are $\Omega(n^d)$ red matching edges, we can argue that the tester
succeeds with the desired probability. Firstly, with probability $\Omega(1)$, the tester
starts the walk at an endpoint $\bx$ of a red edge. Let the matched edge be $(\bx,\by)$. 
With probability $\tau/d \approx d^{-1/2}$, the walk will cross the dimension corresponding
to $(\bx,\by)$. Conditioned on this event, we can interpret the walk as first moving\footnote{There is an annoying edge-case: the point $\bx$ may itself contribute to the ``redness'' of the 
edge $(\bx, \by)$. In this case, we do not move along that dimension but rather take a $(\tau-1)$ length walk from $\bx$. This is why \Alg{alg} takes walks of both lengths, $\tau$ and $\tau-1$. \label{fn:fn1}} to a random interior point $\bz$ in the segment $(\bx, \by)$ and then taking a $(\tau-1)$-length up-walk
from $\bz$ to get to the point $\bz'$. (Refer to \Fig{shift}.)
Since the edge was red, with constant probability, $\bz'$ is $\tau$-$\mzb$. Consider a random shift of $(\bx,\bz')$, shown as $(\bx-\bt, \bz'-\bt)$ in \Fig{shift}. As discussed in the previous paragraph, this shifted
pair is a violation with constant probability.
All in all, the tester succeeds with $\Omega(d^{-1/2})$ probability. 


\begin{figure}[ht!]
    \begin{center}
        \includegraphics[scale=0.33]{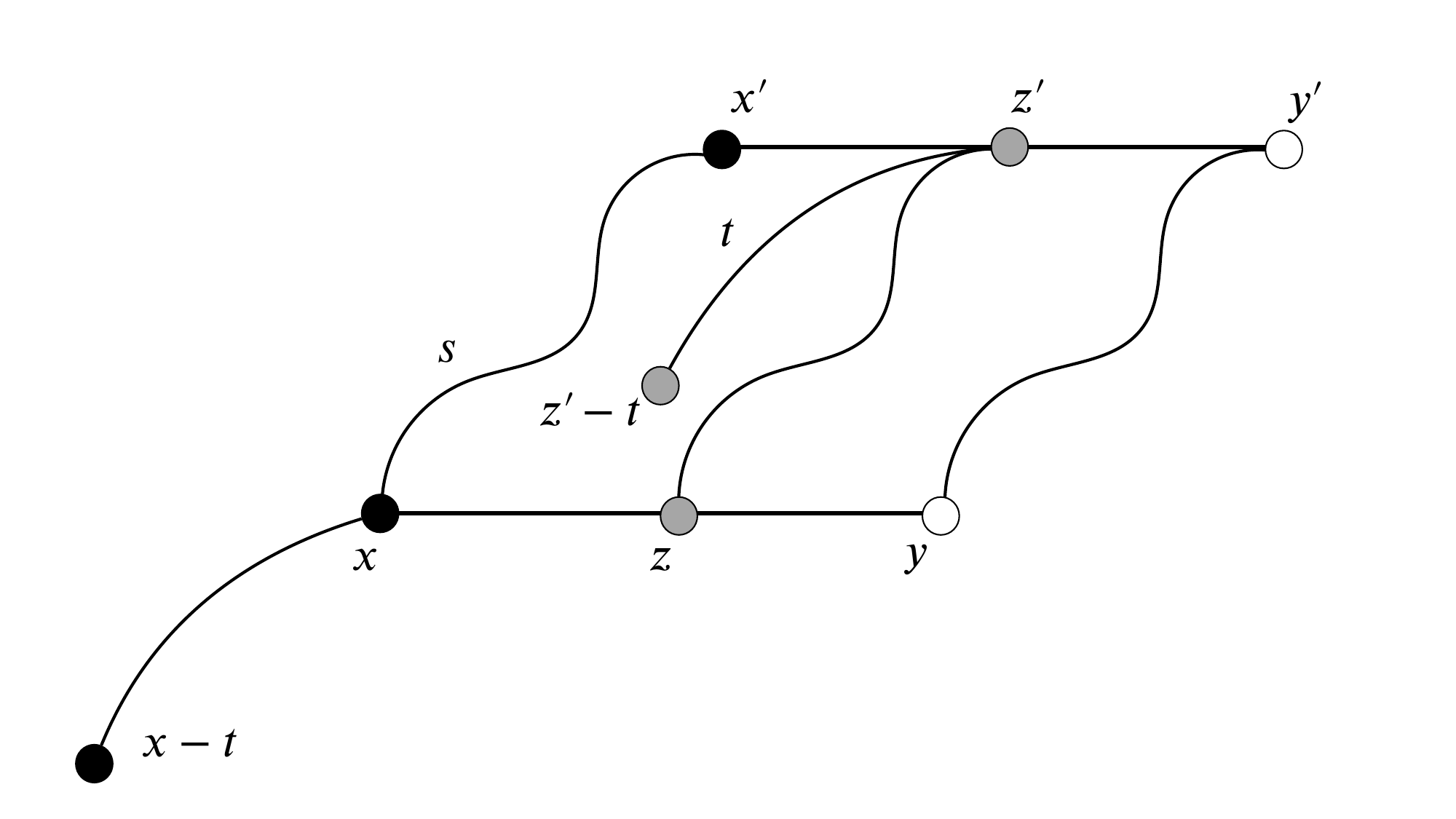}
    \end{center}
    \caption{\em This figure shows the key argument that either up-walks + downshifts, or down-walks find violations. 
    The edge $(\bx,\by)$ is in the initial violation matching. Parallel curves of the same shape denote the same shift.
    So $\bx' = \bx+\bs$, $\by' = \by + \bs$, and $\bz' = \bz + \bs$. Similarly, we see both $\bx$ and $\bz'$ shifted below
    by $\bt$.
    The $1$-valued points are colored black and the $0$-valued points are colored white. Gray points do not have an
    a priori guarantee on function value. If $\bz'$ is $\mzb$, then $f(\bz'-\bt) = 0$ with high probability. In this case,
    $(\bx-\bt, \bz'-\bt)$ is a likely violation. If not, then $(\bz'-\bt,\by')$ is a likely violation.}\label{fig:shift}
\end{figure}
\noindent
But what if there are no red edges? This takes us to the next key idea of our paper: {\em translations} of violation subgraphs.

\paragraph{Translations of violation subgraphs, and blue edges.} Suppose most of the matching edges edges $(\bx,\by)$ are not red.
So, for most points $\bz$ in the segment $(\bx, \by)$, a $(\tau-1)$-length walk  does not reach a $(\tau-1)$-$\mzb$ point. 
Fix one such walk, which can be described by an ``up-shift" $\bs$. So the walk from $\bz$ reaches
$\bz' \eqdef \bz + \bs$.

Consider the corresponding shift of the full edge $(\bx, \by)$ to $(\bx',\by')$, where $\bx' = \bx + \bs$ and $\by' = \by + \bs$. 
Refer to \Fig{shift}. What can we say about this edge?
Since both $\bx$ and $\by$ are up-persistent, with good probability both $f(\bx') = f(\bx) = 1$ and $f(\by') = f(\by) = 0$. 
Observe that most internal points $\bz'$ in $(\bx',\by')$ are not
mostly-zero-below. Consider a $(\tau-1)$-length downward walk from such a point $\bz'$, whose destination
can be represented as $\bz'-\bt$ (for a downshift $\bt$).
With probability $\geq 0.1$, $f(\bz'-\bt) = 1$.

Recall, the tester performs a downward random walk (\Cref{alg:alg}, Step 3) as well. 
Suppose this walk starts at $\by'$. With probability $\approx \tau/d \approx d^{-1/2}$, the walk moves (downward) in the $i$th coordinate. Conditioned on this, the walk ends up at a point $\bz'-\bt$. As discussed above,
$\bz'$ is likely to be not mostly-zero-below. Hence $f(\bz'-\bt) = 1$ with constant probability,
and the tester discovers the violating pair $(\bz'-\bt, \by')$.

\Fig{shift} summarizes the above observations. If $(\bx,\by)$ is red, then the pair $(\bx-\bt,\bz'-\bt)$ is likely to a violation.
If $(\bx,\by)$ is not red, then the pair $(\bz'-\bt,\by')$ is a likely violation.
This motivates the definition of our {\em blue} edges. We call a violating edge blue, if for a constant fraction of points in the interior, a downward walk of length $(\tau - 1)$ leads to a $1$-point with constant probability (\Cref{def:blue}). We argued above that if the edge $(\bx,\by)$ in the violation matching was not red, then a random shift or translation up to $(\bx',\by')$ leads to a blue edge. 
If most edges in our original violation matching were not red, then we could translate ``all these edges together'' to get a (potentially) new large violation subgraph. If most of these new edges are blue, then the downward walk would catch a violation with $\approx d^{-1/2}$ probability. 

What does it mean to translate ``all edges together''? In particular, how do we pin down this new violation matching? 
We use ideas from network flows. Through the random translation, every non-red edge $(\bx,\by)$ in the original violation matching leads to a {\em distribution} over blue edges $(\bx',\by')$.
We treat this as a fractional flow on these blue edges. If the original matching had few red edges, we can construct a large collection of blue edges sustaining a large flow. Integrality of flow implies there must be another large violation matching in the support of this distribution whose edges are blue. 
This is the essence of the ``red/blue'' lemma (\Cref{lem:redblue}).

Putting it together, suppose $G = (\bX,\bY,E)$ is a large violation matching. Either the up-walk with a shift (Step 4) or the down-walk (Step 3) succeeds with probability $\approx d^{-1/2}$.

\paragraph{Lopsided violation subgraphs and translation again.}
We have discussed the situation of a large violation matching $G = (\bX, \bY, E)$ with $|\bX| = |\bY| = \Omega(n^d)$. 
However, such a large matching may not exist. Instead, the directed isoperimetric theorems
imply the existence of a ``good subgraph'' with bounded maximum degree and many edges. 
These graphs $G = (\bX, \bY, E)$ may be lopsided with $|\bX| \ll |\bY|$.
This causes a significant headache for our algorithm, and once again, the issue is persistence. 
The good subgraph could have $|\bX| \approx n^d/\sqrt{d}$, $|\bY| \approx n^d$,
and edges that are structured as follows. All edges incident to an individual $\by \in \bY$
are aligned along the same dimension. For the path tester to find a violation starting from any $\by \in \bY$,
it must take a walk of length $\tau = \widetilde{\Omega}(\sqrt{d})$.

Unlike in~\cite{KMS15} or in~\cite{BlChSe23}, the tester must run both the up-walk and down-walk. 
In the situation of \Fig{shift}, it is critical that both up-walks and down-walks have the same length.
In the lopsided good subgraph indicated above, the walk length is $\widetilde{\Omega}(\sqrt{d})$.
For this length, the fraction of non-persistent points could be $\widetilde{\Omega}(1)$.
In particular, all the vertices in $\bX$ could be non-persistent with respect to this length. 
Thus, the upward walk + downward shift is no longer guaranteed to work. (In \Fig{shift},
we are no longer guaranteed that $f(\bx-\bt) = 1$. To ensure that, the walk must be much shorter.
But in that case, the walk from $\by'$ is unlikely to cross the $i$th dimension.) 

%
%
%
%

To cross this hurdle, we use the translation idea again. Suppose we had a lopsided violation subgraph $G = (\bX, \bY, E)$ with $|\bX| \ll |\bY|$. 
For the walk length $\tau$ determined by $\bY$, most vertices in $\bX$ are not down persistent. However, the vertices in $\bX$ must be up persistent for otherwise the upward walk
would succeed (\Cref{clm:G0_persistence_X}). Therefore, we can take upward translations of $G$ and again using network flow arguments alluded to in the previous paragraph, 
we are able to construct another violation subgraph $G' = (\bX',\bY', E')$
that satisfies the following properties. Firstly, $G'$ is ``structurally" similar to $G$, in terms of degree bounds and the number
of edges. Either vertices in $\bX'$ are $\tau$-down persistent or $|\bX'| \geq 2|\bX|$. We refer to this as the `persist-or-blow-up' lemma (\Cref{lem:blowup}).
The argument is somewhat intricate and requires a delicate balance of parameters. An interesting aspect
is that we can either beat the usual Markov upper bound for persistent vertices, or improve the parameters of the violation
graph. By iterations of the lemma, we can argue the existence of a violation subgraph with all the desired persistent properties.
Then, the analysis akin to the matching case generalizes to give the desired result. 

%
\paragraph{Thresholded degrees, peeling, and the $d^{o(1)}$ loss.} Another gnarly issue with hypergrids is the distinction between degree and ``thresholded degree''. 
The relevant ``degree''  of a vertex $\bx$ (for the path tester analysis) in a violation subgraph is not the {\em number} of edges incident on it, but rather the number of different
dimensions $i$ so that there is an $i$-edge incident on it. We refer to this quantity as the ``thresholded degree",
and it is between $0$ and $d$. Note that the standard degree could be as large as $(n-1)d$. It is critical one uses thresholded degree 
for the path tester analysis, to avoid the linear dependence on $n$ in our calculations. 
Observe that for the matching case, these degrees are identical, making the analysis easier. 

While the path tester analysis works with thresholded degree, the flow-based translation arguments alluded to above need to use normal degrees. 
In particular, we can use flow-arguments to relate the bound the standard degree of the new violation subgraphs.
But we cannot a priori do so for the thresholded degree.

To argue about the thresholded degree, we begin with a stronger notion of a good subgraph called the seed regular violation subgraph (\Cref{lem:peeling}).
This subgraph satisfies specific conditions for both thresholded and standard degrees of the vertices.
It is in the construction of the seed graph where we lose the $d^{o(1)}$ factor.


\paragraph{Road map.} \Cref{sec:prelim} contains technical preliminaries about the hypergrid and defining random walk
distributions. \Cref{sec:redbluedef} has the main definitions of red/blue edges, and provides the terminology
for the main ideas. The analysis begins in \Cref{sec:anal}, and is broken up into the remaining sections.
We use the isoperimetric theorem in~\cite{BlChSe23} to prove the existence of the seed regular graph $G$ (\Cref{lem:peeling}).
This graph may not have the desired persistence properties, so we apply the persist-or-blow-up lemma,~\Cref{lem:blowup}, to obtain a more robust graph $G'$. 
In particular, if Step 2 of \Cref{alg:alg} does not succeed with good probability, then $G$ has good up-persistence properties, allowing us to apply~\Cref{lem:blowup}, obtaining $G'$ with the needed down-persistence properties. 
This graph $G'$ may have lots of red edges, in which case it is a ``nice red subgraph'' (\Cref{def:red-nice}), and then the up-walk + down-shift (Step 4 in~\Cref{alg:alg}) succeeds with good probability.
Otherwise, we apply the ``red/blue'' lemma to obtain a ``nice blue subgraph'' (\Cref{def:blue-nice}), and then the down-walk (Step 3 in~\Cref{alg:alg}) succeeds with good probability. 
Of course, the lopsidedness in the seed graph can be $|\bX| \gg |\bY|$ in which case the argument is analogous, but with the roles of Steps 2 and 3 being exchanged, and the roles of Steps 4 and 5 being exchanged. In particular, in this case one of Step 2 or Step 5 in~\Cref{alg:alg} succeed with good probability.

%
%
%
%
%

\subsection{Related Work} \label{sec:related}

Monotonicity testing, and in particular that of Boolean functions on the hypergrid, has been studied extensively in the past 25 years~~\cite{Ras99,EKK+00,GGLRS00,DGLRRS99,LR01,FLNRRS02,HK03,AC04,HK04,ACCL04,E04,SS08,Bha08,BCG+10,FR,BBM11,RRSW11,BGJ+12,ChSe13,ChSe13-j,ChenST14,BeRaYa14,BlRY14,ChenDST15,ChDi+15,KMS15,BeBl16,Chen17,BlackCS18,BlackCS20,BKR20,HY22,BrKh+23,BlChSe23}. The problem was first considered by \cite{GGLRS00} for the special case of $n=2$ (the \emph{hypercube}) and for general $n \geq 2$ (the \emph{hypergrid}) by \cite{DGLRRS99}. Most of the early works focused on the hypercube domain, $\{0,1\}^d$. 
Early works defined the problem and described a $O(d)$ tester~\cite{Ras99,GGLRS00}. This was improved by~\cite{ChSe13-j} to give an $\otilde_\eps(d^{7/8})$ tester and this paper introduced the connection to directed isoperimetry.
Subsequently, \cite{KMS15} described their $\otilde_\eps(\sqrt{d})$ non-adaptive, one-sided tester via the directed robust version of Talagrand's isoperimetric theorem, and this dependence on $d$ is tight even for two-sided testers~\cite{FLNRRS02,ChenDST15,Chen17}. The best lower bound for adaptive testers is $\Omega(d^{1/3})$~\cite{Chen17,BeBl16}.


Dodis, Goldreich, Lehman, Raskhodnikova, Ron, and Samorodnitsky~\cite{DGLRRS99} were the first to define the problem of monotonicity testing on general hypergrids, 
and they gave a non-adaptive, one-sided $O((d/\eps)\log^2(d/\eps))$-query tester for the Boolean range. Thus, it was known from the beginning
that independence of $n$ is achievable for Boolean monotonicity testing. 
Berman, Raskhodnikova, and Yaroslavtsev improved the upper bound to $O((d/\eps)\log(d/\eps))$ \cite{BeRaYa14}.
They also show a non-adaptive lower bound of $\Omega(\log(1/\eps)/\eps)$ and prove an adaptivity gap
by giving an adaptive $O(1/\eps)$-query tester for constant $d$.

The first $o(d)$ tester for hypergrids was given by Black, Chakrabarty, Seshadhri~\cite{BlackCS18}.
Using a directed Margulis inequality, they achieve a $\otilde_\eps(d^{5/6}\log n)$ upper bound.
In a subsequent result, they introduce the concept of domain reduction and show that
$n$ can be reduced to $\poly(d\eps^{-1})$ by sub-sampling the hypergrid~\cite{BlackCS20}. Harms and Yoshida
gave a substantially simpler proof of the domain reduction theorem, though their result is not ``black-box"~\cite{HY22}.

Most relevant to our work are the independent, recent results of Black, Chakrabarty, Seshadhri,
and Braverman, Kindler, Khot, Minzer~\cite{BlChSe23,BrKh+23}. These results give $\otilde(\poly(n)\sqrt{d})$
query testers, but with different approaches. The former follows the KMS path, and proves 
a new directed Talagrand inequality over the hypergrid. This theorem is a key tool in our
result. The result of~\cite{BrKh+23} follows a different approach, via reductions to hypercube
monotonicity testing. This is a tricky and intricate construction; naive sub-sampling
approaches to reduce to the hypercube are known to fail (see Sec. 8 of \cite{BlackCS20}). 
Instead, their result uses a notion of ``monotone" embeddings that embed functions over
arbitrary product domains to hypercube functions, while preserving the distance to monotonicity.
However, these embeddings increase the dimension by $\poly(n)$, which appears to be inherent.

\subsection{Discussion}

It is an interesting question to see if the $d^{o(1)}$ dependence can be reduced to polylogarithmic in $d$.
As mentioned above, the loss arises due to our need for a stronger notion of a ``good subgraph''.
Nevertheless, we feel one could obtain an $\otilde(\eps^{-2}\sqrt{d})$-tester. In Section 8 of their paper,~\cite{BlChSe23} conjecture a stronger ``weighted'' isoperimetric theorem which would imply a $\otilde(\eps^{-2}\sqrt{d})$-tester. Our work currently has no bearing on that conjecture, and that is still open. 

At a qualitative level, our work and the result in~\cite{BlChSe23} indicates the Boolean monotonicity testing question on the hypergrid seems more challenging than on the hypercube.
Is there a quantitative separation possible? It is likely that non-adaptive monotonicity testing for general hypergrids is  
harder than hypercubes by ``only" a $\log d$ factor. The gap between
the non-adaptive upper and lower bounds even for hypercubes is $\poly(\log d)$. So, achieving this separation
between hypergrids and hypercubes seems quite challenging, as it would require upper and lower bounds 
of far higher precision.

\section{Technical Preliminaries} \label{sec:prelim}
In this section, we list out preliminary definitions and notations.
Throughout the section, we fix a function $f:[n]^d\to \{0,1\}$ that is $\eps$-far from monotone.
For ease of readability, most proofs of this section are in the appendix. 


\subsection{Violation Subgraphs and Isoperimetry}

\noindent
The {\em fully augmented hypergrid} is a graph whose vertex set is $[n]^d$ where edges connect all
pairs that differ in exactly one coordinate. We direct all edges from lower to higher endpoint.
The edge $(\bx, \by)$ is called an $i$-{\em edge} for $i\in [d]$ if $\bx$ and $\by$ differ in the $i$th coordinate.
We use $I(\bx,\by) = \{\bz \colon \bx \preceq \bz \preceq \by\}$ to denote the points $\bz$ in the segment $[\bx, \by]$, that is, they are the points which differ from $\bx$ and $\by$ only in the $i$th coordinate, and $\bx_i \leq \bz_i \leq \by_i$.
Given a function $f:[n]^d \to \{0,1\}$ the edge $(\bx,\by)$ of the fully augmented hypergrid is a {\em violating/violated edge} if $f(\bx)= 1$ and $f(\by) = 0$.

\begin{definition} \label{def:bip} A \emph{violation subgraph} is a subgraph of the fully augmented hypergrid all of whose edges are violations. 
\end{definition}

Note that any violation subgraph is a bipartite subgraph, where
the bipartition is given by the $1$-valued and $0$-valued points. We henceforth always express a violation subgraph as $G = (\bX, \bY, E)$ such that $\forall \bx \in \bX$, $f(\bx) = 1$ and $\forall \by \in \bY$, $f(\by) = 0$. There are a number of relevant parameters of violation subgraphs
that play a role in our analysis.

\begin{definition} \label{def:par} 
	Fix a violation subgraph $G = (\bX,\bY,E)$ and a point $\bx \in \bX$.
	\begin{itemize}[noitemsep]
		\item 	The {\em degree} of $\bx$ in $G$ is the number of edges in $E$ incident to $\bx$ and is denoted as $D_G(\bx)$.
		\item 	For any coordinate $i \in [d]$, the \emph{$i$-degree} of $\bx$
		in $G$ is the total number of $i$-edges in $E$ incident
		to $\bx$ and is denoted as $\Gamma_{G,i}(\bx)$. Note $D_G(\bx) = \sum_{i=1}^d \Gamma_{G,i}(\bx)$.
		\item The {\em thresholded degree} of $\bx$ in $G$ is the number of coordinates $i\in [d]$ with $\Gamma_{G,i}(\bx) > 0$ and is denoted as $\Phi_G(\bx)$.
	\end{itemize}
Whenever $G$ is clear from context, for brevity, we remove it from the subscript.
\end{definition}
\noindent
Note that $\Phi(\bx)$ is an integer between $0$ and $d$, $\Gamma_i(\bx)$ is an integer between $0$ and $(n-1)$, and $D(\bx)$ is an integer between $0$ and $(n-1)d$.
We next define the following parameters of a violation subgraph $G$.
\begin{definition} \label{def:weak-reg} Consider a violation subgraph $G = (\bX, \bY, E)$.
	\begin{asparaitem}
		\item $D(\bX)$ is the maximum degree of a vertex in $\bX$, that is, $D(\bX) = \max_{\bx\in \bX} D(\bx)$.
		\item For $i\in [d]$, $\Gamma_i(\bX)$ is the maximum $i$-degree in $\bX$, that is, $\Gamma_i(\bX) = \max_{\bx\in \bX} \Gamma_i(\bx)$.
		\item $\Gamma(\bX)$ is the maximum value of $\Gamma_i(\bX)$, that is, $\Gamma(\bX) = \max_{i=1}^d \Gamma_i(\bX)$.
		\item $\Phi(\bX)$ is the maximum \emph{thresholded} degree in $\bX$, that is, $\Phi(\bX) = \max_{\bx\in \bX} \Phi(\bx)$.
		\item $m(G)$ is the number of edges in $G$.
	\end{asparaitem}
	(We analogously define these parameters for $\bY$.)
\end{definition}

We recall the notion of {\em thresholded influence} of a function $f:[n]^d\to \{0,1\}$ as defined in~\cite{BlChSe23,BrKh+23}. 
For any $\bx \in [n]^d$ and $i\in [d]$, $\Phi_f(\bx;i)$ is the indicator for the existence of a violating $i$-edge incident to $\bx$.
The thresholded influence of $f$ at $\bx$ is $\Phi_f(\bx) = \sum_{i=1}^d \Phi_f(\bx;i)$. 
We use the same Greek letter $\Phi$ both for thresholded influence and thresholded degree.
In the graph $G_0 = (\bX_0,\bY_0,E)$ consisting of {\em all} violating edges of the fully augmented hypergrid, 
$\Phi_f(\bx)$ is indeed $\Phi_{G_0}(\bx)$.

For applications to monotonicity testing, we require {\em colored/robust} versions of the thresholded influence. For hypercubes this was suggested by~\cite{KMS15}, and 
for hypergrids this was generalized by~\cite{BlChSe23}. Let $\chi:E\to \{0,1\}$ be an {\em arbitrary} coloring of all the edges of the fully augmented hypergrid to $0$ or $1$.
Given a point $\bx$ and $i\in [d]$, $\Phi_{f,\chi}(\bx;i)$ is the indicator of a violating $i$-edge $e$ incident to $\bx$ with $\chi(e) = f(\bx)$.
The colored thresholded influence of $\bx$ with respect to $\chi$ is simply $\Phi_{f,\chi}(\bx) = \sum_{i=1}^d \Phi_{f,\chi}(\bx;i)$. The Talagrand objective of $f$ 
is defined as
\[
	\dtal{f} := \min_{\chi:E\to \{0,1\}}  \sum_{\bx\in [n]^d} \sqrt{\Phi_{f,\chi}(\bx) } \text{.}
\]
The main result of~\cite{BlChSe23} is the following.
\begin{theorem}[Theorem 1.4, \cite{BlChSe23}]\label{thm:bcs23}
	If $f\colon[n]^d \to \{0,1\}$ is $\eps$-far from monotone, then $\dtal{f} = \Omega(\frac{\eps n^d}{\log n})$.
\end{theorem}

We stress that the RHS above only loses a $\log n$ factor, which allows for domain reduction
(setting $n = \poly(d)$). This is what yields the nearly optimal $\sqrt{d}$ dependence
and independence on $n$ in the tester query complexity. 

We extend the definition of $\dtal{f}$ to arbitrary violation subgraphs as follows. Given a violation subgraph $G = (\bX,\bY,E)$ and a bicoloring $\chi:E\to \{0,1\}$ of its edges,
for $\bz\in \bX\cup \bY$ and $i\in [d]$ let $\Phi_{G,\chi}(\bz;i) = 1$ if there is a violating $i$-edge $e \in E(G)$ incident to $\bz$ such that $\chi(e) = f(\bz)$, and $\Phi_{G,\chi}(\bz;i) = 0$ otherwise.
Define $\Phi_{G,\chi}(\bx) = \sum_{i=1}^d \Phi_{G,\chi}(\bx;i)$.
Note, if $\chi \equiv 1$, that is every edge is colored $1$, then $\Phi_{G,\chi}(\bx) = \Phi_G(\bx)$ for $\bx\in \bX$ and $\Phi_{G,\chi}(\by) = 0$ for all $\by\in \bY$. Similarly, if $\chi \equiv 0$, then $\Phi_{G,\chi}(\by) = \Phi_G(\by)$ for $\by \in \bY$ and $\Phi_{G,\chi}(\bx) = 0$ for $\bx\in \bX$. 

\begin{definition} \label{def:tal-graph} Given a  violation subgraph $G = (\bX, \bY, E)$, we define $\dtal{G} \eqdef \min_\chi \sum_{\bz \in \bX \cup \bY}[\sqrt{\Phi_{G,\chi}(\bz)}]$, where the min is taken over all edge bicolorings $\chi:E(G) \to \{0,1\}$. \end{definition}

If $G_0$ is the subgraph of all violations in the fully augmented hypergrid, then~\Cref{thm:bcs23} states $\dtal{G_0} = \Omega(\eps n^d/\log n)$.
We make a couple of observations.

\begin{observation}\label{obs:oo}
	For any violation subgraph $G = (\bX, \bY, E)$, 
	\begin{itemize}
		\item $D(\bX) \leq \Gamma(\bX)\Phi(\bX)$ and $D(\bY) \leq \Gamma(\bY)\Phi(\bY)$.
		\item $m(G) \geq \dtal{G}$.
	\end{itemize}
\end{observation}
\begin{proof}
	For any $\bx \in \bX$, we have $D(\bx) = \sum_{i=1}^d \Gamma_i(\bx) = \sum_{i: \Gamma_i(\bx) > 0} \Gamma_i(\bx) \leq \left(\max_i \Gamma_i(\bx)\right) \cdot \Phi(\bx) \leq \Gamma(\bX)\Phi(\bX)$. The proof is analogous for $\bY$. For the second bullet, observe that
	$m(G) = \sum_{\bx\in \bX} D(\bx) \geq \sum_{\bx \in \bX} \Phi(\bx) \geq \sum_{\bx \in \bX} \sqrt{\Phi(\bx)} = \sum_{\bz \in \bX\cup \bY} \sqrt{\Phi_{G,\chi\equiv 1}(\bz)} \geq \dtal{G}$.
\end{proof}

\begin{remark} \label{rem:delta_d_remark}
	Since we assume that $n$ is at most a polynomial in $d$, we fix a constant $c$ such that $nd \leq d^c$. 
	Throughout the remainder of the paper, we consider $d$ to be at least a large constant and fix $\delta = \frac{1}{\floor{\log\log nd}}$. 
	As a result, we use bounds such as ``$d^\delta \geq \poly\log d$'' or ``$d - C\sqrt{d} \geq d/3$'' without explicitly reminding 
	the reader that $d$ is large. We use $\Theta(\delta)$ to denote $C \cdot \delta$ for some unspecified, but fixed constant $C$. 
\end{remark}

\subsection{Equivalent Formulations of the Random Walk Distribution} \label{sec:equivalent_dists}

Recall the random walk distribution described in \Cref{def:walkdist}. It is useful to think of this walk as first sampling a random hypercube and then taking a random walk on that hypercube. The following definition describes the appropriate distribution over sub-hypercubes in $[n]^d$. 

\begin{definition} [Hypercube Distribution] \label{def:hypercube-dist} We define the following distribution $\HH_{n,d}$ over sub-hypercubes in $[n]^d$. For each coordinate $i \in [d]$:
	\begin{enumerate}
		\item Choose $q_i \in_R \{1,2,\ldots,\log n\}$ uniformly at random.
		\item Choose a uniform random interval $I_i$ of size $2^{q_i}$ in $\Z_n$.
		\item Choose a uniform random pair $a_i < b_i$ from $I_i$.
	\end{enumerate}
	Output $\bH = \prod_{i=1}^d \{a_i,b_i\}$. When $n$ and $d$ are clear from context, we abbreviate $\HH = \HH_{n,d}$. \end{definition}

It will also be useful for us to think of our random walk distribution as first sampling $\bx \in_R [n]^d$, then sampling a random hypercube which contains $\bx$, and then taking a random walk from $\bx$ in that hypercube. The appropriate distribution over hypercubes containing a point $\bx$ is defined as follows.

\begin{definition} [Conditioned Hypercube Distribution] \label{def:con-hypercube-dist} Given $\bx \in [n]^d$, we define the conditioned sub-hypercube distribution $\HH_{n,d}(\bx)$ as follows. For each $i \in [d]$:
	\begin{enumerate}
		\item Choose $q_i \in_R \{1,2,\ldots,\log n\}$ uniformly at random.
		\item Choose a uniform random interval $I_i$ in $\mathbb{Z}_n$ of size $2^{q_i}$ such that $\bx_i \in I_i$.
		\item Choose a uniform random $c_i \in_R I_i \setminus \{\bx_i\}$.
		\item Set $a_i = \min(\bx_i,c_i)$ and $b_i = \max(\bx_i,c_i)$.
	\end{enumerate}
	Output $\bH = \prod_{i=1}^d \{a_i,b_i\}$. When $n$ and $d$ are clear from context we will abbreviate $\HH(\bx) = \HH_{n,d}(\bx)$. \end{definition}

The random walk distribution in a hypercube $\bH$ is defined as follows.

\begin{definition} [Hypercube Walk Distribution] \label{def:hypercube-walk} For a hypercube $\bH = \prod_{i=1}^d \{a_i,b_i\}$, a point $\bx \in \bH$, and a walk length $\tau$, we define the upward random walk distribution $\mathcal{U}_{\bH,\tau}(\bx)$ over points $\by \in \bH$ as follows.
	\begin{enumerate}
		\item Pick a uniform random subset $R \subseteq [d]$ of $\tau$ coordinates.
		\item Generate $\by$ as follows. For every $r \in [d]$, if $r \in R$ and $\bx_r = a_r$, set $\by_r = b_r$. Else, set $\by_r = \bx_r$.
	\end{enumerate}  
	Analogously, the downward random walk distribution $\mathcal{D}_{\bH,\tau}(\bx)$ is defined precisely as above, but instead in step 2 if $r \in R$ and $\bx_r = b_r$, we set $\by_r = a_r$, and otherwise $\by_r = \bx_r$. \end{definition}

We observe that the following walk distributions are equivalent and defer the proof to the appendix \Sec{dists}. 

\begin{fact} \label{fact:dists} The following three distributions over pairs $(\bx,\by) \in [n]^d \times [n]^d$ are all equivalent.
	\begin{enumerate}
		\item $\bx \in_R [n]^d$, $\by \sim \cU_{\tau}(\bx)$.
		\item $\bH \sim \HH$, $\bx \in_R \bH$, $\by \sim \cU_{\bH,\tau}(\bx)$.
		\item $\bx \in_R [n]^d$, $\bH \sim \HH(\bx)$, $\by \sim \cU_{\bH,\tau}(\bx)$.
	\end{enumerate}
	The analogous three distributions defined using downward random walks are also equivalent. \end{fact}

It is also convenient to define the shift distribution for hypercubes.

\begin{definition} [Shift Distributions for Hypercube Walks] Given a hypercube $\bH$, the up-shift distribution from $\bx \in \bH$, denoted $\cU\cS_{\bH,\tau}(\bx)$ is the distribution of $\bx'-\bx$, where $\bx' \sim \cU_{\bH,\tau}(\bx)$. The down-shift distribution from $\by \in \bH$, denoted $\cD\cS_{\bH,\tau}(\by)$ is the distribution of $\by-\by'$, where $\by' \sim \cD\cS_{\bH,\tau}(\by)$. \end{definition}

\subsection{Influence and Persistence}

We define the following notion of influence for our random walk distribution \Def{walkdist}.

\begin{definition} \label{def:hypergrid-influence} The total and negative influences of $f \colon [n]^d \to \{0,1\}$ are defined as follows.
	\begin{itemize}
		\item $\widetilde{I}_f = \EX_{\bx \in [n]^d}\left[d \cdot \Pr_{\by \sim \cU_{1}(\bx)}[f(\bx) \neq f(\by)]\right]$
		\item $\widetilde{I}_f^- = \EX_{\bx \in [n]^d}\left[d \cdot \Pr_{\by \sim \cU_{1}(\bx)}[f(\bx) > f(\by)]\right]$
	\end{itemize}
\end{definition}


The probability of the tester (\Alg{alg}) finding a violation in step (2b) when $\tau = 1$ is precisely $\widetilde{I}_f^-/d$. Recall the definition of the distribution $\HH$ in \Def{hypercube-dist}. For brevity, for a hypercube $\bH = \prod_{i=1}^d \{a_i,b_i\}$ sampled from $\HH$, we abbreviate $I_{\bH} := I_{f|_{\bH}}$ and $I_{\bH}^- := I_{f|_{\bH}}^-$. That is, $I_{\bH}(\bx)$ is the number of coordinates $i$ for which $\bx_i = a_i$, and $f(\bx_1,\ldots,\bx_{i-1},b_i,\bx_{i+1},\ldots,\bx_d) \neq f(\bx)$. If $f(\bx) = 1$, then $I_{\bH}^-(\bx)$ is the number of coordinates $i$ for which $\bx_i = a_i$, and $f(\bx_1,\ldots,\bx_{i-1},b_i,\bx_{i+1},\ldots,\bx_d) = 0$, and if $f(\bx) = 0$, then $I_{\bH}^-(\bx) = 0$. Note that these definitions are such that influential edges are always charged to the endpoint with $a_i$ in the $i$'th coordinate so that we do not double count. Then, $I_{\bH} = \EX_{\bx \in \bH}[I_{\bH}(\bx)]$ and $I_{\bH}^- = \EX_{\bx \in \bH}[I_{\bH}^-(\bx)]$.


\begin{claim} \label{clm:inf-hypercube} $\widetilde{I}_f = \EX_{\bH\sim \HH}\left[I_{\bH}\right]$ and $\widetilde{I}_f^- = \EX_{\bH\sim \HH}\left[I_{\bH}^-\right]$.
\end{claim}

\begin{proof} By \Fact{dists}, the distribution $(\bx \in_R [n]^d,\by \sim \cU_1(\bx))$ is equivalent to first sampling $\bH \sim \HH$, then sampling $(\bx \in_R \bH, \by \sim \cU_{\bH,1}(\bx))$. Recalling \Cref{def:hypercube-walk}, observe that $\Pr_{\by \sim \cU_{\bH,1}(\bx)}[f(\bx) \neq f(\by)] = I_{\bH}(\bx)/d$. Putting these observations together yields
	\[
	\widetilde{I}_f = \EX_{\bx \in [n]^d}\left[d \cdot \Pr_{\by \sim \cU_{1}(\bx)}[f(\bx) \neq f(\by)]\right] = \EX_{\bH\sim \HH} \EX_{\bx \in \bH}\left[I_{\bH}(\bx)\right] = \EX_{\bH \sim \HH}[I_{\bH}]
	\]
	An analogous argument proves the statement for negative influence. \end{proof}
\noindent

The following claim states that if the normal influence is (very) large, then so is the negative influence. This is a simple generalization of, and indeed easily follows from, Theorem 9.1 in~\cite{KMS15}. The proof can be found in~\Sec{app:inf-per}.


\begin{restatable}{claim}{totaltonegativeinf}\label{clm:total-to-negative-inf} If $\widetilde{I}_f > 9\sqrt{d}$, then $\widetilde{I}_f^- > \sqrt{d}$. \end{restatable}

Next, we define the notion of persistent points. This is similar to that in~\cite{KMS15} with a parameterization that we need for our purpose.

\begin{definition} \label{def:pers} Given a point $\bx \in [n]^d$, a walk length $\tau$, and a parameter $\beta \in (0,1)$, we say that $\bx$ is $(\tau,\beta)$-up-persistent if
	\[
	\Pr_{\by \sim \cU_{\tau}(\bx)}[f(\by) \neq f(\bx)] \leq \beta\text{.}
	\]
	Similarly, $\bx$ is called $(\tau,\beta)$-down-persistent if the above bound holds when $\by$ is drawn from the downward walk distribution, $\cD_{\tau}(\bx)$. If both bounds hold, then we call $\bx$ $(\tau,\beta)$-persistent. \end{definition}

The following claim upper bounds the fraction of non-persistent points. This is a generalization of Lemma 9.3 in~\cite{KMS15}.
The proof is deferred to~\Sec{app:inf-per}.

\begin{restatable}{claim}{persistence} \label{clm:persistence} If $\widetilde{I}_f \leq 9\sqrt{d}$, then the fraction of vertices that are not $(\tau,\beta)$-persistent is at most $C_{per} \frac{\tau}{\beta\sqrt{d}}$ where $C_{per}$ is a universal constant. \end{restatable}

We make another simple but technical assumption used later on. If there are a ``lot'' of points $\bx$ with $f(\bx) = 1$ that are {\em not}
up-persistent, it is easy to show that the algorithm succeeds with high probability. The interesting case is when this does not occur, which is codified in~\Cref{ass:tech-ass}.
Recall the definition of $\delta, c$ from~\Cref{rem:delta_d_remark}.

\begin{claim}\label{clm:tech-ass-just}
	If there exists a $\tau$ which is a power of $2$ such that the number of points $\bx$ with $f(\bx) = 1$ {\em and} 
	which are {\em not} $(\tau - 1, \log^{-5} d)$-up-persistent is $ \geq \frac{\eps n^d}{d^{1/2 + 7 c\delta}}$, then \Cref{alg:alg} succeeds with 
	probability at least $\eps^2 d^{-1/2 + \Theta(\delta)}$.
\end{claim}

\begin{proof} By the definition of persistence and the tester definition, \Alg{alg} rejects with the desired probability when it runs the upward path tester with walk length $\tau-1$ (step (2) of \Alg{alg}). \end{proof}

\begin{assumption}\label{ass:tech-ass}
	For any power of two, $\tau =2^p$, the number of points  $\bx$ with $f(\bx) = 1$ {\em and} 
	which are {\em not} $(\tau - 1, \log^{-5} d)$-up-persistent is $< \frac{\eps n^d}{d^{1/2 + 7c\delta}}$.	
\end{assumption}

\subsection{The Middle Layers, Typical Points, and Walk Reversibility} \label{sec:middle}

All proofs in this section are deferred to \Sec{app:reversible}.

\begin{definition} In a hypercube $\hyp{d}$, the \emph{$c$-middle layers} consist of all points with Hamming weight in the range $[d/2 \pm \sqrt{4cd\log(d/\eps)}]$. Given a $d$-dimensional hypercube $\bH$, we let $\bH_c \subseteq \bH$ denote the $c$-middle layers of $\bH$. \end{definition}

We state a bound on the number of points in the hypercube which lie in the middle layers. This follows from a standard Chernoff bound argument.

\begin{restatable}{claim}{hypmiddle} \label{clm:hyp-middle} For a $d$-dimensional hypercube $\bH$ and $c \geq 1$, we have $|\bH_c| \geq \left(1-(\eps/d)^c\right) \cdot 2^d$. \end{restatable}


We now define the notion of typical points in $[n]^d$. Recall the distribution $\HH_{n,d}$ (\Def{hypercube-dist}) over random sub-hypercubes in $[n]^d$ and the distribution $\HH_{n,d}(\bx)$ (\Def{hypercube-dist}) over random sub-hypercubes in $[n]^d$ that contain $\bx$. 
A point $\bx$ is $c$-\emph{typical} if for most sub-hypercubes containing $\bx$, the point $\bx$ is present in their $c$-middle layers.

\begin{definition} [Typical Points] \label{def:typical} Given $c \geq 1$, a point $\bx \in [n]^d$ is called $c$-\emph{typical}
	if
	\[
	\Pr_{\bH \sim \HH(\bx)}\left[\bx \in \bH_c \right] \geq 1-(\eps/d)^5 \text{.}
	\]
\end{definition}

\begin{restatable}{claim}{typical} \label{clm:typical} For any $\eps \in (0,1)$ and $c \geq 6$, 
	\[
	\Pr_{\bx \in_R [n]^d}\left[\bx \text{ is } c\text{-typical}\right] \geq 1-(\eps/d)^{c-5} \text{.}
	\]
\end{restatable}


Intuitively, a short random walk from a typical point will always lead to point that is almost as typical. This is formalized as follows.

\begin{restatable}[Translations of Typical Points]{claim}{translatetypical} \label{clm:translate-typical} Suppose $\bx \in [n]^d$ is $c$-typical. Then for a walk length $\tau \leq \sqrt{d}$, every point $\bx' \in \text{supp}(\cU_{\tau}(\bx)) \cup \text{supp}(\cD_{\tau}(\bx))$ is $(c+\frac{\tau}{\sqrt{d}})$-typical. \end{restatable}


Recall the three equivalent ways of expressing the walk distribution in \Fact{dists}.
We define the random walk probabilities only on points in the middle layers. This setup
allows for the approximate reversibility of \Lem{prob-rev}.

\begin{definition} \label{def:walkpdf} Consider two vertices $\bx \prec \bx' \in [n]^d$ and a walk length $\tau$. We define
	\begin{align} \label{eq:walkpdf}
		p_{\bx,\tau}(\bx') = \EX_{\bH \sim \HH(\bx)}\left[\mathbf{1}\left(\bx \in \bH_{100} ~\wedge~ \bx'\in \bH_{100}\right) \cdot \Pr_{\bz \sim \cU_{\bH,\tau}(\bx)}[\bz = \bx']\right]
	\end{align}
	to be the probability of reaching $\bx'$ by a random walk from $\bx$, only counting the contribution when the random walk is taken on a hypercube that contains $\bx$ and $\bx'$ in the $100$-middle layers. We analogously define $p_{\bx',\tau}(\bx)$ using the downward random walk distribution in $\bH$. \end{definition}

Consider $\bx \prec \bx'$ are two points in the middle layers.
The following lemma asserts that the probability of reaching from $\bx$ to $\bx'$ via an upward walk of length $\ll \sqrt{d}$ is similar to the probability of reaching from $\bx'$ to $\bx$ via downward walk of the same length. 

\begin{restatable}[Reversibility Lemma]{lemma}{probrev} \label{lem:prob-rev} For any points $\bx \prec \bx' \in [n]^d$ and walk length $\ell \leq \sqrt{d}/\log^5(d/\eps)$, we have 
\[
p_{\bx,\ell}(\bx') = (1 \pm \log^{-3} d) p_{\bx',\ell}(\bx)\text{.}
\]
\end{restatable}


\section{Red Edges, Blue Edges, and Nice Subgraphs} \label{sec:redbluedef}

We now set the stage to prove~\Cref{thm:main}.
The first definition is that of {\em mostly-zero-below} points. These are points from which a downward random walk (\Cref{def:walkdist}) 
leads to a point where the function evaluates to $0$ with high probability.

\begin{definition} \label{def:predict} A point $\bz$ is called \emph{$\ell$-mostly-zero-below}, or $\ell$-$\mzb$, if $\Pr_{\bz' \sim \downdist{\bz}{\ell}}[f(\bz') = 0] \geq 0.9$. \end{definition}

To appreciate the utility of $\ell$-$\mzb$ points, consider the following scenario. Suppose $\bx$ is a point with $f(\bx) = 1$ and is $(\ell, \beta)$-down-persistent (\Cref{def:pers}) for some small $\beta$. Next suppose an upward random walk from $\bx$ reaches an $\ell$-$\mzb$ point $\bz$. Then, we claim that Step 4 of \Alg{alg} would succeed with constant probability in finding a violated edge. 
An $\ell$-length downward walk from $\bx$, due to down-persistence, would lead to a $\bx'$ with $f(\bx') =1$ with probability at least $1-\beta$.
The same $\ell$-length downward walk from $\bz$ would lead to a $\bz'$ with $f(\bz') = 0$ with $\geq 0.9$ probability, since $\bz$ is mostly-zero-below. 
Since $(\bx,\bz)$ are comparable, so would be $(\bx',\bz')$. By a union bound, $(\bx',\bz')$ is a violation with probability at least $0.9-\beta$.

The next definition describes edges $(\bx,\by)$ of the violation subgraph most of whose internal vertices lead to $\mzb$-points via an upward random walk. 
Rather un-creatively, we call such edges {\em red}. 
Recall that $I(\bx,\by) = \{\bz \colon \bx \preceq \bz \preceq \by\}$ denotes the closed interval of points from $\bx$ to $\by$. 

\begin{definition} \label{def:red} A violated edge $(\bx,\by)$ is called \emph{red} for walk length $\ell$ if
	$$ \Pr_{\bz \in I(\bx,\by)\text{, }\bz' \sim \updist{\bz}{\ell}} [\bz' \ \textrm{is $\ell$-$\mzb$}] \geq 0.01 \text{.}$$
	When $\ell$ is clear by context, we call the edge red.
\end{definition}


There may be no $\ell$-$\mzb$ points for the lengths we choose, that is, a downward walk from any point leads to a point where the function evaluates to $1$. 
In that case, Step 3 of \Alg{alg} is poised to succeed; for any violating edge $(\bx, \by)$, if we start from $\by$ then the downward walk should give a violation.
This motivates the next definition which recognizes violated edges $(\bx,\by)$ most of whose internal vertices lead to points where the function evaluates to $1$ via 
a downward random walk. We call such edges {\em blue}.

\begin{definition} \label{def:blue} A violated edge $(\bx,\by)$ is called \emph{blue}
	for walk length $\ell$ if
	$$ \Pr_{\bz \in I(\bx,\by)\text{, }\bz' \sim \downdist{\bz}{\ell}} [f(\bz') = 1] \geq 0.01\text{.}$$
		When $\ell$ is clear by context, we simply call the edge blue.
\end{definition}


\noindent
We note that a violating edge $(\bx,\by)$ may be {\em both} red and blue, or perhaps more problematically, {\em neither} red nor blue. 
The next two definitions capture certain ``nice'' violation subgraphs consisting of either red or blue edges. 
In~\Sec{anal}, we show that if either of these subgraphs exist then we can prove the tester works with the desired probability. In~\Sec{obt} we show that one of these subgraphs must exist. Recall, $\Phi_H(\bx)$ is the {\em thresholded degree} of $\bx$ in the subgraph $H$ and $\delta > (\log \log nd)^{-1}$ is fixed (\Cref{rem:delta_d_remark}).

\begin{definition}[$(\sigma, \tau)$-nice red violation subgraph]\label{def:red-nice} Given a parameter $\sigma \in (0,1)$ and a walk length $\tau$, a violation subgraph  $H(\bA,\bB,E)$ is called a $(\sigma, \tau)$-nice red violation subgraph if the following hold. 
\begin{enumerate}[label=(\alph*)]
	\item All edges in $H$ are red for walk length $\tau-1$.
	\item All vertices in $\bA$ are $(\tau-1,0.6)$-down-persistent.
	\item $\sigma \Phi_H(\bx) \leq d^{1/2}$ for all $\bx \in \bA$.
	\item $\sigma \sum_{\bx \in \bA} \Phi_H(\bx) \geq \eps^2 \cdot n^d \cdot d^{-\Theta(\delta)}$.
	\item $d^{1/2 - \Theta(\delta)} \geq  \tau \geq \sigma \cdot d^{1/2 - \Theta(\delta)}$.
\end{enumerate}
\end{definition}
The first two conditions dictate that the subgraph is nice with respect to the length of the walk. In particular, the edges are red with respect to this length
and furthermore the $1$-vertices are down-persistent. As explained before the definition of red edges, this property is crucial for the success of Step 4 of \Alg{alg}.
The fourth condition says that the total thresholded degree of the $1$-vertices in $H$ is large. I.e. for an average vertex $\bx \in \bA$, there will be many coordinates $i$ for which there is an $i$-edge in $H$ incident to $\bx$. The third condition says that the max thresholded degree of vertices in $\bA$ is not too large and so the total thresholded degree from the fourth condition must be somewhat spread among the vertices in $\bA$.
The final condition shows that the length of the walk is large compared to $\sigma$. 
Note, if $\sigma = \Theta(1)$ and the third bullet point's right hand side was $1$ instead of $\sqrt{d}$, 
we would be in the case of a large matching of violated edges, which was the ``simple case'' discussed in~\Sec{ourideas}.

The next definition is the analogous case of blue edges. When this type of subgraph exists we argue that Step 3 of \Alg{alg} succeeds. Note that Step 3 is the downward path test (without a shift) and so we do not need a persistence property like condition (b) in the previous definition. This definition has the same conditions on the thresholded degree as the previous definition, but with respect to the $0$-vertices of the subgraph.

\begin{definition}[$(\sigma, \tau)$-nice blue violation subgraph]\label{def:blue-nice}
	Given a parameter $\sigma \in (0,1)$ and a walk length $\tau$, a violation subgraph  $H(\bA,\bB,E)$ is called a $(\sigma, \tau)$-nice blue violation subgraph if the following hold. 
	\begin{enumerate}[label=(\alph*)]
		\item All edges in $H$ are blue for walk length $\tau-1$.
		\item $\sigma \Phi_H(\by) \leq d^{1/2}$ for all $\by \in \bB$.
		\item $\sigma \sum_{\by\in \bB} \Phi_H(\by) \geq  \eps^2 \cdot n^d \cdot d^{-\Theta(\delta)}$.
		\item $d^{1/2 - \Theta(\delta)} \geq \tau \geq \sigma \cdot d^{1/2 - \Theta(\delta)}$.
	\end{enumerate}
\end{definition}

\noindent The following lemma captures the utility of the above definitions. Its proof can be found in~\Sec{anal}.

\begin{lemma}[Nice Subgraphs and Random Walks]\label{lem:tester-analysis}
	Suppose for a power of two $\tau\geq 2$, there exists a $(\sigma,\tau)$-nice red subgraph or a $(\sigma,\tau)$-nice blue subgraph. Then \Alg{alg} finds a violating pair, and thus rejects $f$, with probability at least $\eps^2 \cdot d^{-(1/2 + \Theta(\delta))}$.
\end{lemma}

\noindent The following lemma shows that one of the two nice subgraphs always exists. Its proof can be found in \Sec{obt}.

\begin{restatable}[Existence of nice subgraphs]{lemma}{mainlemma}\label{lem:good_subgraph} Let $\eps \geq d^{-1/2}$ and let $c$ be a constant such that $nd \leq d^c$. 
	Suppose $f \colon [n]^d \to \{0,1\}$ is $\eps$-far from monotone , $\widetilde{I}_f \leq 9\sqrt{d}$, and~\Cref{ass:tech-ass} holds for $\delta = \frac{1}{\floor{\log\log nd}}$.
	There exist $0 < \sigma_1 \leq \sigma_2 < 1$, 	a violation subgraph $H(\bA,\bB,E)$,
	and a power of two $\tau \geq 2$, 
 such that $H$ is either a $(\sigma_1, \tau)$-nice red subgraph or a $(\sigma_2, \tau)$-nice blue subgraph.
\end{restatable}

\section{Tester Analysis}\label{sec:anal}

In this section we prove~\Cref{thm:main}. First, in \Sec{testeranalysis1} we prove \Lem{tester-analysis} which is the main tester analysis. Then in \Sec{testeranalysis2} we combine \Lem{tester-analysis}, \Lem{good_subgraph} (which will be proven in~\Sec{obt}), and \Clm{total-to-negative-inf} to prove \Cref{thm:main}.


\subsection{Main Analysis: Proof of \Cref{lem:tester-analysis}} \label{sec:testeranalysis1}

There are two cases depending on whether we have a nice red subgraph or a nice blue subgraph.
In Case 1, Step 4 of \Alg{alg} proves the lemma while in Case 2, Step 3 of \Alg{alg} proves the lemma. The proofs are similar, but we provide both for completeness.

	\subsubsection{Case 1: there exists a $(\sigma,\tau)$-nice red subgraph}


Suppose there exists a $(\sigma,\tau)$-nice red subgraph, $H(\bA, \bB, E)$, where $\tau \geq 2$ is a power of two. In Step 1, the tester in \Alg{alg} chooses $\tau$ as the walk length with probability at least $\log^{-1} d$. Thus, in the rest of the analysis we will condition on this event. 

Given $\bx \in \bA$, let $C_{\bx} \subseteq [d]$ denote the set of coordinates for which $\bx$ has an outgoing edge in $H$. Note $|C_{\bx}|= \Phi_H(\bx)$.
Now recall Step 4 of \Alg{alg}. We choose $\bx$ uniformly at random and sample $\by$ from $\cU_{\ell}(\bx)$ for $\ell \in \{\tau-1,\tau\}$.
Let these two samples of $\by$ be called $\by'$ and $\by$, respectively. The reason for choosing the two different walk lengths (as alluded to in~\Cref{fn:fn1}) will be made clear below. 

We first lower bound the probability that the sampled $\bx$ lies in  $\bA$ and $R \cap C_{\bx} \neq \emptyset$ where $R \subseteq [d]$ is a random set of $\tau$ coordinates. Let $\cE_1$ denote this event. The main calculation is to lower bound the probability of this event as follows.
\begin{align*}
	\Pr[\cE_1] &= \frac{1}{n^d} \sum_{\bx \in \bA} \Pr[R \cap C_{\bx} \neq \emptyset] \geq \frac{1}{n^d} \sum_{\bx \in \bA} \left[1-\left(1-\frac{|C_{\bx}|}{d}\right)^{\tau}~\right] \geq \frac{1}{n^d} \sum_{\bx \in \bA} \left[1 - \exp\left(-\frac{\tau |C_{\bx}|}{d}\right)\right]  
\end{align*}
The RHS can only decrease if we replace $\tau$ with its lower bound (\Cref{def:red-nice}, (e)) of $\sigma\cdot d^{1/2-\Theta(\delta)}$. Also, observe that
$\frac{\sigma d^{1/2-\Theta(\delta)} |C_{\bx}|}{d} = \frac{\sigma \Phi_H(\bx)}{d^{1/2+\Theta(\delta)}} \leq 1$ using our upper bound, $\sigma \Phi_H(\bx) \leq d^{1/2}$ (\Cref{def:red-nice}, (c)). 
Now, using $e^{-x} \leq 1-\frac{x}{2}$ for $x \leq 1$, the exponential term in the RHS is at most $1-\frac{\sigma \Phi_{H}(\bx)}{2d^{1/2+\Theta(\delta)}}$, yielding
\begin{align}
	\Pr[\cE_1] &\geq \frac{\sigma}{2d^{1/2+\Theta(\delta)}} \cdot \frac{1}{n^d} \sum_{\bx \in \bA} \Phi_H(\bx) \underbrace{\geq}_{\text{(\Cref{def:red-nice}, (d))}} \frac{\eps^2}{ d^{1/2+\Theta(\delta)}}
\end{align}

The event $\cE_1$ asserts that the tester has chosen a point $\bx \in \bA$ and there is at least one  $r\in R$ for which there exists a red edge $(\bx,\bx + a \be_r) \in E$ for some integer $a > 0$ in the subgraph $H$. Fix the smallest such $r \in R \cap C_{\bx}$ and the corresponding edge in $H$. 

Recall the random walk process in \Def{walkdist}. We define the following good events.

\begin{itemize}
	\item $\cE_2$: Step (2a) chooses $q_r$ satisfying: if $a \leq n/4$, then $2^{q_r} \in [2a, 4a]$; if $a > n/4$, then $2^{q_r} = n$.
        \item $\cE_3$: Step (2b) chooses the interval $I_r \supseteq [\bx_r + 1,\bx_r+a]$.
        \item $\cE_4$: Step (2c) chooses $c_r$ uniformly~
        from $[\bx_r + 1,\bx_r+a]$.
	\item $\cE_5$: $\by$ or $\by'$ is $(\tau-1)$-mostly-zero-below as per~\Cref{def:predict}.
	\item $\cE_6$: $f(\by-\bs) = 0$ or $f(\by' - \bs) = 0$ for $\bs$ chosen in Step 4 of \Alg{alg} from $\downshift{\bx}{\tau-1}$.
	\item $\cE_7$: $f(\bx-\bs) = 1$ for $\bs$ chosen in Step 4 of \Alg{alg} from $\downshift{\bx}{\tau-1}$.
\end{itemize}

Firstly, note that $\Pr[\cE_2] = \log^{-1} n$ for both cases of the edge length, $a$. Now, suppose $a \leq n/4$. Then, $\Pr[\cE_3 ~|~ \cE_2] \geq 1/2$ by the condition $2^{q_r} \geq 2a$ and $\Pr[\cE_4 ~|~ \cE_2, \cE_3] \geq 1/4$ by the condition $2^{q_r} \leq 4a$. If $a > n/4$, then $\Pr[\cE_3 ~|~ \cE_2] = 1$, since in this case $I_r = [n]$ and again $\Pr[\cE_4 ~|~ \cE_2, \cE_3] \geq 1/4$ since $(\bx_r,\bx_r + a]$ is at least a fourth of the entire line, $[n]$.

Now, the edge $(\bx, \bx + a\be_r)$ is red for walk length $\tau - 1$. Recall \Cref{def:red}; we get that if we sample $\bz \in [\bx, \bx+a \be_r]$ uniformly at random and then 
sample $\bz' \sim \cU_{\tau - 1}(\bz)$, then $\bz'$ is $(\tau - 1)$-$\mzb$ with probability $\geq 0.01$. Note that the interval $\bz$ is drawn from is a {\em closed} interval containing $\bx$.
On the other hand $c_r$ above is chosen from $[\bx_r + 1, \bx_r + a]$. To account for this, we branch into two possibilities.
Either $\bz' \sim \cU_{\tau - 1}(\bx)$ is $(\tau - 1)$-$\mzb$ with probability $\geq 0.01$, and in this case $\by'$ is $(\tau-1)$-mostly-zero-below.
Or $\bz'  \sim \cU_{\tau - 1}(\bz)$  where $\bz$ itself is sampled from $\bz \sim [\bx + \be_r, \bx+a \be_r]$ is $(\tau - 1)$-$\mzb$ with probability $\geq 0.01$, and in this case $\by$ is $(\tau-1)$-mostly-zero-below. In sum, we have $\Pr[\cE_5 ~|~ \cE_4] \geq 0.01$. 

If $\by$ is $(\tau - 1)$-mostly-zero-below, then if we sample $\hat{\bs}$ from $\downshift{\by}{\tau-1}$ we get $f(\by-\hat{\bs}) = 0$ with probability $\geq 0.9$.
Now note that $\downshift{\by}{\tau-1}$ and $\downshift{\bx}{\tau-1}$ differ only when the set $R\subseteq [d]$ chosen in~\Cref{def:walkdist} contains
a coordinate in $\supp(\by-\bx)$. Since $|\supp(\by-\bx)| \leq \tau$, $|R|\leq \tau$, and $\tau = o(\sqrt{d})$, we have $\Pr_R[R \cap \supp(\by-\bx) \neq \emptyset] \leq \tau^2/d = o(1)$. Therefore, when $\bs$ is drawn from $\downshift{\bx}{\tau-1}$, 
we get $f(\by-\bs) = 0$ with probability $\geq 0.9(1-o(1)) \geq 0.8$. 
Analogously, if $\by'$ is $(\tau - 1)$-mostly-zero-below, then $f(\by' - \bs) = 0$ with probability $\geq 0.8$.
In sum, we get that $\Pr[\cE_6 ~|~ \cE_5] \geq 0.8$.
 
 Finally, all points in $\bA$ are $(\tau-1,0.6)$-down-persistent (\Cref{def:pers}) and so $\Pr[\cE_7 ~|~ \bx\in A] \geq 0.4$.  

\noindent
Now, let's put everything together. The final success probability of the tester is at least $\Pr[\cE_6 \wedge \cE_7]$, which by a union bound and the reasoning above, is at least
\begin{align*}
	&\left(1-\Pr[\neg \cE_6 ~|~ \cE_5] - \Pr[\neg \cE_7 ~|~\bx\in A]\right) \cdot \Pr\left[\bigwedge_{i=1}^5 \cE_i\right] \\
	&\geq (1-0.2-0.6) \cdot \frac{\eps^2}{d^{1/2+\Theta(\delta)}} \cdot \frac{1}{\log n} \cdot \frac{1}{2} \cdot \frac{1}{4} \cdot \frac{1}{100} \geq \frac{\eps^2}{d^{1/2+\Theta(\delta)}}
\end{align*}
where in the last inequality we used $n \leq \poly(d)$. This completes the proof when the nice subgraph is red.

\subsubsection{Case 2: there exists a $(\sigma,\tau)$-nice blue subgraph}

Suppose there exists a $(\sigma,\tau)$-nice blue subgraph, $H(\bA, \bB, E)$, where $\tau \geq 2$ is a power of two. As in Case 1, in Step 1, the tester in \Alg{alg} chooses $\tau$ as the walk length with probability at least $\log^{-1} d$. Thus, in the rest of the analysis we will condition on this event. 	
Now recall Step 3 of \Alg{alg}. 
We choose $\by$ uniformly at random and sample $\bx$ from $\cD_{\ell}(\by)$ for $\ell \in \{\tau-1,\tau\}$.
Let these two samples of $\bx$ be called $\bx'$ and $\bx$, respectively. 

	Given $\by \in \bB$, let $C_{\by} \subseteq [d]$ denote the set of coordinates for which $\by$ has an incoming edge in $H$. Note $|C_\by| = \Phi_H(\by)$.
We first lower bound the probability that $\by \in \bB$ and $R \cap C_{\by} \neq \emptyset$ where $R \subseteq [d]$ is a random set of $\tau$ coordinates. Let $\cE_1$ denote this event. The main calculation is to lower bound the probability of this event as follows.
	\begin{align*}
		\Pr[\cE_1] &= \frac{1}{n^d} \sum_{\by \in \bB} \Pr[R \cap C_{\by} \neq \emptyset] \geq \frac{1}{n^d} \sum_{\by \in \bB} \left[1-\left(1-\frac{|C_{\by}|}{d}\right)^{\tau}~\right] \geq \frac{1}{n^d} \sum_{\by \in \bB} \left[1 - \exp\left(-\frac{\tau |C_{\by}|}{d}\right)  \right]
	\end{align*}
As in Case 1, the RHS can only decrease if we replace $\tau$ with its lower bound (\Cref{def:blue-nice}, (d)) of $\sigma\cdot d^{1/2-\Theta(\delta)}$, and a similar argument 
as in Case 1 gives
\begin{align}
	\Pr[\cE_1] &\geq \frac{\sigma}{d^{1/2+\Theta(\delta)}} \cdot \frac{1}{n^d} \sum_{\by \in \bB} \Phi_H(\by) \underbrace{\geq}_{\text{(\Cref{def:blue-nice}, (c))}} \frac{\eps^2}{d^{1/2+\Theta(\delta)}}
\end{align}
%

As in Case 1,
	the event $\cE_1$ says that the tester has chosen a point $\by \in \bB$ and there exists $r \in R$ such that there exists an edge $(\by - a \be_r,\by) \in E$ in the subgraph $H$ for some integer $a > 0$. Fix the smallest $r \in R \cap C_{\by}$ and the corresponding edge in $H$. 
Now define the following good events for the remainder of the tester analysis. 
	
	\begin{itemize}
		\item $\cE_2$: Step (2a) chooses $q_r$ satisfying: if $a \leq n/4$, then $2^{q_r} \in [2a, 4a]$; if $a > n/4$, then $2^{q_r} = n$.
		\item $\cE_3$: Step (2b) chooses the interval $I_r \supseteq [\by_r - a,\by_r-1]$.
		\item $\cE_4$: Step (2c) chooses $c_r$ uniformly~
		from $[\by_r - a,\by_r-1]$.
		\item $\cE_5$: $f(\bx) = 1$ or $f(\bx') = 1$.
	\end{itemize}
	
	The final success probability of the tester is at least $\Pr[\wedge_{i=1}^5 \cE_i]$. Firstly, note that $\Pr[\cE_2] = \log^{-1} n$ for both cases of the edge length, $a$. Suppose $a \leq n/4$. Then, $\Pr[\cE_3 ~|~ \cE_2] \geq 1/2$ by the condition $2^{q_r} \geq 2a$ and $\Pr[\cE_4 ~|~ \cE_2, \cE_3] \geq 1/4$ by the condition $2^{q_r} \leq 4a$. If $a > n/4$, then $\Pr[\cE_3 ~|~ \cE_2] = 1$, since in this case $I_r = [n]$ and again $\Pr[\cE_4 ~|~ \cE_2, \cE_3] \geq 1/4$.

    Now, the edge $(\by-a \be_r, \by)$ is blue for walk length $\tau-1$. Recall \Cref{def:blue}; we get that if we sample  $\bz \in [\by-a \be_r, \by]$ uniformly at random and then 
    sample $\bz' \sim \cD_{\tau - 1}(\bz)$, then $f(\bz') = 1$ with probability $\geq 0.01$. 
    We split in cases depending on whether the contribution to the probability comes primarily from $\bz = \by$. This is equivalent
    to splitting the interval $[\by-a \be_r,\by]$ into $[\by-a\be_r, \by-\be_r]$ and the singleton $\by$.
    Either $\bz' \sim \cD_{\tau -1}(\by)$ satisfies $f(\bz')=1$ with probability $\geq 0.01$ or the following occurs.
    When $\bz \in [\by-a \be_r, \by - \be_r]$ is chosen uniformly at random, $\bz' \sim \cD_{\tau - 1}(\bz)$  satisfies $f(\bz')=1$ with probability $\geq 0.01$.
    In the former case, the distribution of $\bz'$ is the same as $\bx'$ and in the latter it is same as $\bx$.
    So, 
     we have $\Pr[\cE_5 ~|~ \cE_4] \geq 0.01$. Putting everything together, we have
	\begin{align*}
		\Pr\left[\bigwedge_{i=1}^5 \cE_i\right] \geq \frac{\eps^2}{d^{1/2+\Theta(\delta)}} \cdot \frac{1}{\log n} \cdot \frac{1}{2} \cdot \frac{1}{4} \cdot \frac{1}{100} \geq \frac{\eps^2}{d^{1/2+\Theta(\delta)}}
	\end{align*}
where in the last step we used $n \leq \poly(d)$ and this completes the proof when the nice subgraph is blue. Together, the cases complete the proof of~\Cref{lem:tester-analysis}. 

\subsection{Tying it Together: Proof of \Cref{thm:main}} \label{sec:testeranalysis2}

Suppose $f:[n]^d \to \{0,1\}$ is $\eps$-far from being monotone with $n \leq \poly(d)$ and $\eps \geq d^{-1/2}$. 
In particular, we have constant $c$ such that $nd \leq d^c$. Recall $\delta = 1/\floor{\log\log nd} = o(1)$.	
By~\Cref{clm:tech-ass-just}, we may assume~\Cref{ass:tech-ass} holds for otherwise we are done. 
Also recall the definitions of $\widetilde{I}_f, \widetilde{I}_f^-$ in \Def{hypergrid-influence}. By \Clm{total-to-negative-inf}, if $\widetilde{I}_f > 9\sqrt{d}$, then $\widetilde{I}_f^- > \sqrt{d}$ and so the tester (\Alg{alg}) finds a violation in step (2) when $\tau = 1$ with probability $\Omega(d^{-1/2})$. Thus, we will assume $\widetilde{I}_f \leq 9\sqrt{d}$. Therefore, we may invoke \Lem{good_subgraph} which gives us either a nice red subgraph or a nice blue subgraph.~\Cref{lem:tester-analysis} then proves that \Alg{alg} finds a violating pair and rejects with probability at least $\eps^2\cdot d^{-(1/2 + \Theta(\delta))}$. This proves \Cref{thm:main}. 

\section{Finding Nice Subgraphs}\label{sec:obt}	

In this section we prove~\Cref{lem:good_subgraph} which we restate below.
\mainlemma*

The proof proceeds over multiple steps and constitutes a key technical contribution of the paper.
We give a sketch of what is forthcoming. 
\begin{itemize}

	\item In~\Sec{peeling} we describe the construction of a {\em seed} regular violation subgraph $G$.
	This uses the directed isoperimetric result~\Cref{thm:bcs23} proved in~\cite{BlChSe23} and a ``peeling argument'' not unlike that present in~\cite{KMS15}.
	At the end of this section, we will fix the parameters $\sigma_1, \sigma_2$ and the walk length $\tau$. In particular, the length $\tau$ will be defined by the 
	{\em larger} side of this violating bipartite graph.
	
	\item In~\Sec{deriveH}, we obtain a regular violating graph $H$ that has persistence properties with respect to the walk length $\tau$. 
In~\cite{KMS15} and~\cite{BlChSe23}, one obtained this violating graph by simply deleting the non-persistent points from the seed violation subgraph.
In our case, since we choose the walk length depending on the larger side, we need 
to be careful. We use the idea of ``translating violation subgraphs'' on $G$ (repeatedly) to find a different violation subgraph $H$ with
the desired persistence properties. 

	\item In~\Sec{derive_redblue}, we use the graph $H$ to obtain either a nice red subgraph $H_1$ or a nice blue subgraph $H_2$. 
	If most of the edges in $H$ were red, then a simple surgery on $H$ itself gives us $H_1$. On the other hand, if $H$ has few red edges 
	(but has the persistence properties as guaranteed), then we apply the red/blue lemma (\Cref{lem:redblue}) to obtain the desired nice blue-subgraph $H_2$.
	The proof of the red/blue lemma, which is present in~\Sec{redblue}, uses the translating violation subgraphs idea as well.
\end{itemize}

Recall, we assume $f \colon [n]^d \to \{0,1\}$ is $\eps$-far from monotone, $\widetilde{I}_f \leq 9\sqrt{d}$, and~\Cref{ass:tech-ass} holds.

\subsection{Peeling Argument to Obtain Seed Regular Violation Subgraph} \label{sec:peeling}

Recall the definition of the Talagrand objective (\Cref{def:tal-graph}) $\dtal{G}$ of a violation subgraph $G = (\bX,\bY, E)$.
Let $G_0$ denote the violation subgraph formed by all violating edges in the fully augmented hypergrid.
Theorem 1.4 in~\cite{BlChSe23} (paraphrased in this paper as~\Cref{thm:bcs23}) is that $\dtal{G_0} = \Omega(\eps n^d /\log n)$. Also recall the definitions in~\Cref{def:par}.
The following lemma asserts that there exists a subgraph of $G_0$ whose Talagrand objective is not much lower, but satisfies certain regularity properties.

\begin{lemma} [Seed Regular Violation Subgraph] \label{lem:seed}\label{lem:peeling} There exists a violation subgraph $G(\bX,\bY,E)$ satisfying the following properties.\footnote{We remark that this lemma in particular \emph{does not} rely on $\widetilde{I}_f \leq 9\sqrt{d}$ or \Cref{ass:tech-ass}. That is, it holds as long as $f \colon [n]^d \to \{0,1\}$ is $\eps$-far from monotone.}
	\begin{enumerate}[label=(\alph*)]
		\item $\dtal{G} \geq \eps \cdot d^{-c\delta} \cdot n^d$.
		\item $m(G) \geq d^{-3c\delta}\max(|\bX| \cdot \Phi(\bX) \cdot \Gamma(\bX),|\bY| \cdot \Phi(\bY) \cdot \Gamma(\bY))$.
		\item All vertices in $\bX \cup \bY$ are $98$-typical.
		\item $|\bX|,|\bY| \geq \frac{\eps}{d^{1/2+c\delta}} \cdot n^d$.
	\end{enumerate}
\end{lemma}

Let us make a few comments before proving the above lemma. Condition (a) shows that the Talagrand objective degrades only by a $d^{o(1)}$ factor. Condition (b) shows
that the graph is nearly regular
since the RHS term without the $d^{-o(1)}$ term is the maximum value of $m(G)$. This is because $\Phi(\bX)\Gamma(\bX)$ is an upper bound on the maximum degree of any vertex $\bx \in \bX$.
Indeed, if one can prove a stronger lemma which replaces the $d^{o(1)}$ terms in (a) and (b) by $\polylog(d)$'s, then the remainder of our analysis could be easily modified to give a $\tilde{O}(\eps^{-2}\sqrt{d})$ tester. \smallskip

We need a few tools to prove this lemma. Our first claim is a consequence
of the subadditivity of the square root function.

\begin{claim} \label{clm:subadd} Consider a partition of (the edges of) a violation subgraph $G$ into $H_1, H_2, \ldots, H_k$. Then $\sum_{j \leq k} \dtal{H_j} \geq \dtal{G}$.
\end{claim}

\begin{proof} Let $\chi_j$ denote the coloring of the subgraph $H_j$ that obtains the minimum $\dtal{H_j}$. Since the $H_1,\ldots,H_k$ form a partition,
	we can aggregate the colors to get a coloring $\chi$ of $G$.
	
	Consider any $\bz \in \bX \cup \bY$. Let $\Phi_{H_j, \chi_j}(\bz)$ be the thresholded
	degree of $\bz$, restricted to the edges colored by $\chi_j$.
	By the subadditivity of the square root function, $\sum_{j \leq k} \sqrt{\Phi_{H_j,\chi_j}(\bz)} \geq \sqrt{\sum_{j \leq k} \Phi_{H_j,\chi_j}(\bz)}$.
	Observe that thresholded degrees are also subadditive, so $\sum_{j \leq k} \Phi_{H_j, \chi_j}(\bz) \geq \Phi_{G,\chi}(\bz)$. Hence,
	\begin{equation}
		\sum_{j \leq k} \dtal{H_j} = \sum_{j \leq k} \sum_{\bz \in \bX \cup \bY} \sqrt{\Phi_{H_j, \chi_j}(\bz)}
		= \sum_{\bz \in \bX \cup \bY} \sum_{j \leq k} \sqrt{\Phi_{H_j,\chi_j}(\bz)}
		\geq \sum_{\bz \in \bX \cup \bY} \sqrt{\Phi_{G,\chi}(\bz)} \geq \dtal{G}
	\end{equation}
\end{proof}

\begin{remark}
	The proof of~\Cref{clm:subadd} crucially uses the fact that in the definition of $\dtal{}$,  we minimize over all possible colorings $\chi$'s of the edges.
	In particular, if we had defined $\dtal{G}$ only with respect to the all ones or the all zeros coloring, then the above proof fails. 
	In the remainder of the paper, we will only be using the $\chi \equiv 1$ or $\chi\equiv 0$ colorings, and the curious reader may wonder why we need the definition of $\dtal{G}$ to 
	minimize over all colorings. This is exactly the point where we need it. We make this remark because the ``uncolored'' isoperimetric theorem 
    is much easier to prove than the ``colored'' version, but the colored/robust version is essential for the tester analysis.
\end{remark}

\noindent
Our next step is a simple bucketing argument. 

\begin{claim} \label{clm:bucket} Consider a violation subgraph $G = (\bX,\bY,E)$. Both of the following are true.
	\begin{enumerate}
		\item There exists a subgraph $G' = (\bX',\bY',E')$ of $G$ such that $\dtal{G'} \geq \delta^2 \dtal{G}$ and $m(G') \geq (nd)^{-\delta} |\bX'| \Phi_{G'}(\bX') \Gamma_{G'}(\bX')$.
		\item There exists a subgraph $G' = (\bX',\bY',E')$ of $G$ such that $\dtal{G'} \geq \delta^2 \dtal{G}$ and $m(G') \geq (nd)^{-\delta} |\bY'| \Phi_{G'}(\bY') \Gamma_{G'}(\bY')$.
	\end{enumerate}
\end{claim}

\begin{proof} We prove item (1) and the proof of item (2) is analogous.
	
	For convenience, we assume that $\delta$ is the reciprocal of a natural number.
	For each $\bx \in \bX$, we bucket the incident edges as follows. First,
	for each $a \in [1/\delta]$, let $S_a$ be the set of dimensions $i$,
	such that the $i$-degree of $\bx$ is in the range $[n^{(a-1)\delta}, n^{a\delta})$. Note that $S_1,\ldots,S_{1/\delta}$ forms a partition of the set of coordinates, $[d]$.
	Now, for each $a, b \in [1/\delta]$, let the $(a,b)$ edge bucket of $\bx$, denoted $E_{a,b,\bx}$,
	be defined as follows. If $|S_a| \in [d^{(b-1)\delta}, d^{b\delta})$,
	then $E_{a,b,\bx}$ is the set of all edges incident to $\bx$ along dimensions in $S_a$. If $|S_a| \notin [d^{(b-1)\delta}, d^{b\delta})$, then $E_{a,b,\bx} = \emptyset$. Observe that $\{E_{a,b,\bx} \colon a,b \in [1/\delta]\}$ partitions the edges incident to $\bx$.

	Now, let $G_{a,b}$ denote the subgraph formed by the edge set $\cup_{\bx \in \bX} E_{a,b,\bx}$. Let $\bX_{a,b}$ be the set
	of vertices in $\bX$ with non-zero degree in $G_{a,b}$.
	Observe that $\Phi_{G_{a,b}}(\bX_{a,b}) \leq d^{b\delta}$
	and $\Gamma_{G_{a,b}}(\bX_{a,b}) \leq n^{a\delta}$. Moreover, the degree of each $\bx \in \bX_{a,b}$
	is at least $d^{(b-1)\delta} \times n^{(a-1)\delta} \geq (nd)^{-\delta} \Phi_{G_{a,b}}(\bX_{a,b}) \Gamma_{G_{a,b}}(\bX_{a,b})$.
	Hence, $m(G_{a,b}) \geq (nd)^{-\delta} |\bX_{a,b}| \Phi_{G_{a,b}}(\bX_{a,b}) \Gamma_{G_{a,b}}(\bX_{a,b})$.
	
	Finally, by construction, the $G_{a,b}$ subgraphs partition the edges of $G$. Hence, by \Clm{subadd} we have $\sum_{a,b \in [1/\delta]} \dtal{G_{a,b}} \geq \dtal{G}$. By averaging, there exists some choice of $a,b$ such that $\dtal{G_{a,b}} \geq \delta^2 \dtal{G}$.
	This gives the desired subgraph $G'$. \end{proof}

\Clm{bucket}, part 1 above gives the regularity condition only with respect to $\bX$, and part 2 gives
the analogous guarantee with respect to $\bY$, but the trouble is in getting both simultaneously. We do an iterative construction using \Clm{bucket} to get the simultaneous guarantee.

\begin{proof} (Conditions (a) and (b) of \Lem{peeling}.) By the robust directed Talagrand theorem for hypergrids (\Cref{thm:bcs23}), there is a violation subgraph $G_0 = (\bX_0,\bY_0,E_0)$ such that $\dtal{G_0} = \Omega(\eps n^d/\log n)$. We construct a series of subgraphs $G_0 \supseteq G_1 \supseteq G_2 \supseteq \cdots \supseteq G_r$ as follows.

	Let $i \geq 1$. If $i$ is odd, we apply item (1) of \Clm{bucket} to $G_{i-1}$ to get $G_i(\bX_i,\bY_i,E_i)$ with the regularity condition on $\bX_i$. If $i$ is even,
	we apply item (2) of \Clm{bucket} to $G_{i-1}$ to get $G_i(\bX_i,\bY_i,E_i)$ with the regularity condition on $\bY_i$. If $i > 1$ and $m(G_i) \geq (nd)^{-\delta} m(G_{i-1})$, then we terminate the series. By \Clm{bucket}, the series satisfies the following three properties for all $i \geq 1$.
	\begin{itemize}
		\item $\dtal{G_i} = \Omega(\delta^{2i} \eps n^d/\log n)$.
		\item If $i$ is odd, $m(G_i) \geq (nd)^{-\delta} |\bX_i| \Phi_{G_i}(\bX_i) \Gamma_{G_i}(\bX_i)$. If $i$ is even, $m(G_i) \geq (nd)^{-\delta} |\bY_i| \Phi_{G_i}(\bY_i) \Gamma_{G_i}(\bY_i)$.
		\item If the series has not terminated by step $i$, then $m(G_i) < (nd)^{-\delta} m(G_{i-1})$.
	\end{itemize}
	
	\medskip
	
	The first two statements hold by the guarantees of \Clm{bucket} and the fact that $\dtal{G_0} = \Omega(\eps n^d/\log n)$. The third statement holds simply by the termination condition for the sequence. The trivial bound on the number of edges is $m(G_0) \leq nd \cdot n^{d}$.
    The third bullet point yields $m(G_i) < (nd)^{-i\delta} \cdot nd \cdot n^d$, if the series has not terminated by step $i$. 
	
	\begin{claim} The series terminates in at most $3/\delta$ steps. \end{claim}
	
	\begin{proof} Suppose not. Noting that $m(G_i) \geq \dtal{G_i}$ (\Cref{obs:oo}), we get the following chain of inequalities using the properties of our subgraph graph $G_{3/\delta}$.
		\begin{equation*}
			(nd)^{-(3/\delta)\cdot \delta} \cdot nd \cdot n^{d} > m(G_i) \geq \dtal{G_i} = \Omega(\delta^{6/\delta} \eps n^d/\log n)
			\ \ \ \Longrightarrow \ \ \ (nd)^{-2} = \Omega(\delta^{6/\delta} \eps/\log n)
		\end{equation*}
		Note that we may assume $\eps \geq 1/d$ and so $C\eps/\log n \geq (nd)^{-1}$ for any constant $C$. Thus we have $(nd)^{-1} \geq \delta^{6/\delta}$. Given that $\delta > 1/\log\log nd$, this inequality is a contradiction. \end{proof}
	
	By the previous claim the series terminates in some $r \leq 3/\delta$ steps, producing $G_r(\bX_r,\bY_r,E_r)$, which we claim has the desired properties to prove conditions (a) and (b) of \Lem{peeling}. Since $r \leq 3/\delta$, $\dtal{G_r} = \Omega(\delta^{6/\delta} \eps n^d/\log n)$. Note that since 
$\delta > 1/\log\log nd$, we have
	\begin{align*}
		\delta^{6/\delta} > \left(\log \log nd\right)^{-\frac{6}{\delta}} = (nd)^{-\frac{6}{\delta} \cdot \frac{\log\log\log nd}{\log nd}} > (nd)^{-\delta^2} > (nd)^{-\delta} \cdot \log n > d^{-c\delta}\log n
	\end{align*}
	where the second to last step holds because $\frac{6 \log \log \log nd}{\log d} \ll \left(\frac{1}{\log \log nd}\right)^3 < \delta^3$. The last inequality used $nd \leq d^{c}$. This proves condition (a). Towards proving condition (b), note that $C\delta^{6/\delta}/\log n \geq (nd)^{-\delta}$ for any constant $C$.

	Let's assume without loss of generality that $r$ is even. Thus we have $m(G_r) \geq (nd)^{-\delta} |\bY_r| \Phi_{G_r}(\bY_r) \Gamma_{G_r}(\bY_r)$ by the second bullet point above. Next, since the series terminated at step $r$, we have 
	\[
	m(G_r) \geq (nd)^{-\delta} m(G_{r-1}) \geq (nd)^{-2\delta} |\bX_{r-1}| \Phi_{G_{r-1}}(\bX_{r-1}) \Gamma_{G_{r-1}}(\bX_{r-1}) \geq (nd)^{-2\delta} |\bX_{r}| \Phi_{G_r}(\bX_{r}) \Gamma_{G_r}(\bX_{r})
	\]
	where the second inequality is again by the second bullet point above and the fact that $i-1$ is odd and the third inequality is simply because $G_r$ is a subgraph of $G_{r-1}$. Again using $nd \leq d^{c}$, we have $(nd)^{-\delta} \geq d^{-c\delta}$ and so we get that $G_r$ satisfies conditions (a) and (b) of \Lem{peeling}. \end{proof}


\begin{proof}(Conditions (c) and (d) of \Lem{peeling}.) 
	To obtain condition (c), we simply remove the non-typical points.
		Recall the definition of $c$-typical points (\Def{typical}). By \Clm{typical}, the number of points that are not $98$-typical is at most $(\eps/d)^{93} n^d$. Thus, removing all such vertices can decrease $\dtal{G}$ by at most $(\eps/d)^{93}n^d \cdot \sqrt{d}$ which is negligible compared to the RHS in condition (a). Thus, we remove all such vertices from $G$ and henceforth assume that all of $\bX \cup \bY$ is $98$-typical. 
	
	Condition (d) follows from condition (a). Consider the constant coloring $\chi \equiv 1$ and observe that
	\[
	|\bX|\sqrt{d} \geq \text{Tal}_{\chi\equiv 1}(G) \geq \text{Tal}(G) \geq \eps \cdot d^{-c\delta} \cdot n^d \text{.} 
	\]
	where the first inequality follows from the trivial observation that the maximum $\Phi_{G}(\bx)$ can be is $d$.
	Using the coloring $\chi \equiv 0$ proves the same lower bound for $|\bY|$. \end{proof}

\subsubsection{Choice of the walk length} \label{sec:choice-tau}
We end this section by specifying what the parameters $\sigma_1, \sigma_2$ and $\tau$ are going to be in~\Cref{lem:good_subgraph}.
We now make the assumption $|\bX| \leq |\bY|$. Given~\Cref{rem:wlog}, this is without loss of generality; this fact would be true either in $f$ or in $g$, 
and 
running steps 2, 3, 5 on $f$ is equivalent to running steps 2, 3, 4 on $g$. The violation subgraphs for $f$ and $g$ are isomorphic. 
Then,
\[
	\sigma_1 = \sigma_{\bX} := \frac{|\bX|}{n^d} ~~~\textrm{and}~~~~	\sigma_2 = \sigma_{\bY} := \frac{|\bY|}{n^d}
\]
\noindent
and set $\tau$ to be the unique power of two such that 
	\[
	\frac{1}{2} \lceil \sigma_{\bY} \cdot d^{1/2-7c\delta} \rceil < \tau-1 \leq \lceil \sigma_{\bY} \cdot d^{1/2-7c\delta} \rceil \text{.}
	\]

We conclude the subsection by establishing the following upper bounds on the number of vertices which are not up-persistent. 

\begin{claim} \label{clm:G0_persistence_X} \label{clm:G0_persistence_Y}
	The following are true.
	\begin{itemize}[noitemsep]
		\item  The number of vertices $\bx\in [n]^d$ where $f(\bx) = 1$ that are not $(\tau-1,\log^{-5} d)$-up-persistent is at most $d^{-6c\delta} \cdot |\bX|$.
		\item  The number of vertices $\by\in [n]^d$ where $f(\by) = 0$ that are not $(\tau-1,\log^{-5} d)$-up-persistent is at most $d^{-6c\delta} \cdot |\bY|$. 
	\end{itemize}
 \end{claim}

\begin{proof} Suppose that the first item does not hold. Then, by item (d) of \Lem{seed}, the number of vertices $\bx\in [n]^d$ where $f(\bx) = 1$ that are not $(\tau-1,\log^{-5} d)$-up-persistent is \emph{at least} $d^{-6c\delta} \cdot |\bX| \geq \frac{\eps}{d^{1/2+7c\delta}} \cdot n^d$, violating \Cref{ass:tech-ass} which is a premise of~\Cref{lem:good_subgraph}.
%
%
For the second item, note that by \Clm{persistence}, the total number of $(\tau-1,\log^{-5} d)$-non-persistent vertices is at most $\cper \tau \cdot \log^5 d \cdot \frac{1}{\sqrt{d}} \cdot n^d \leq \sigma_{\bY} \cdot d^{-6c\delta} \cdot n^d$, where we have simply used $\log^5 d \ll d^{c\delta}$ and our definition of $\tau$. \end{proof}


\subsection{Using `Persist-or-Blow-up' Lemma to Obtain Down-Persistence} \label{sec:deriveH}

\Cref{lem:peeling} provides a seed violation subgraph which has a large Talagrand objective and has regularity properties.~\Cref{clm:G0_persistence_X} shows that we may assume these vertices are up-persistent with respect to walk length of $\tau-1$. However, we may not have down persistence. In particular, it could be $|\bX| \ll |\bY|$ and if we try to apply \Clm{persistence} and remove
all nodes from $\bX$ which are not $(\tau-1, 0.6)$-down-persistent, we may end up removing everything. To obtain a subgraph with down-persistence properties, we need to apply a
translation procedure which is encapsulated in the lemma below. The proof of the lemma is deferred to~\Sec{blowup}.
The reader should recall the definitions in~\Cref{def:par} and~\Cref{def:weak-reg}.

\begin{restatable}[Persist-or-Blow-up Lemma] {lemma}{blowuplemma}\label{lem:blowup}~
 Consider a violation subgraph $G = (\bX, \bY, E)$ such that all vertices in $G$ are $c$-typical where $c \leq 99$ and $(\ell,\log^{-5} d)$-up persistent where $1 \leq \ell \leq \sqrt{d}/\log^5 (d/\eps)$. Then, there exists a violation subgraph $G' = (\bX', \bY', E')$ where all vertices are $(c+\frac{\ell}{\sqrt{d}})$-typical and satisfying one of the following conditions.
	\begin{enumerate}
		\item Down-persistent case:
		\begin{asparaenum}
			\item All vertices in $\bX'$ are $(\ell,0.6)$-down persistent.
			\item $m(G') \geq \floor{m(G)/\log^5 d}$.
			\item $D_{G'}(\bX') \leq D_G(\bX)$, and $\forall i \in [d], \Gamma_{G',i}(\bX') \leq \Gamma_{G,i}(\bX)$
			\item $D_{G'}(\bY') \leq D_G(\bY)$, and $\forall i \in [d]$, $\Gamma_{G',i}(\bY') \leq \Gamma_{G,i}(\bY)$.
		\end{asparaenum}
		\item Blow-up case:
		\begin{asparaenum}
			\item $m(G') \geq \floor{2(1-3\log^{-4}d) \cdot m(G)}$.
			\item $D_{G'}(\bX') \leq D_G(\bX)$, and $\forall i \in [d], \Gamma_{G',i}(\bX') \leq \Gamma_{G,i}(\bX)$
			\item $D_{G'}(\bY') \leq 2D_G(\bY)$, and $\forall i \in [d]$,  $\Gamma_{G',i}(\bY') \leq 2\Gamma_{G,i}(\bY)$.
		\end{asparaenum}
	\end{enumerate}
\end{restatable}

That is, the application of the above lemma either gives the violation subgraph we need, or it gives us a violation subgraph with around double the edges.
In the remainder of this section we use \Lem{blowup} and the graph $G(\bX,\bY,E)$ derived in the previous section to prove the following lemma.

\begin{lemma} [Down-Persistent Violation Subgraph] \label{lem:subgraph_after_blowup} 
	Let $G(\bX,\bY,E)$ be the subgraph asserted in~\Cref{lem:peeling}.
	There exists a natural number $s \leq \log^3 d$ and a violation subgraph $H(\bA,\bB,E)$ with the following properties.
	
	\begin{enumerate}
		\item $m(H) \geq 2^s\frac{m(G)}{\log^7 d}$.
		\item $\Gamma_H(\bA) \leq \Gamma_G(\bX)$ and $\Gamma_H(\bB) \leq 2^s \Gamma_G(\bY)$.
		\item $D_H(\bA) \leq D_G(\bX)$ and $D_H(\bB) \leq 2^s D_G(\bY)$.
		\item All vertices in $\bA \cup \bB$ are $(\tau-1,\log^{-5} d)$-up-persistent and $99$-typical.
		\item All vertices in $\bA$ are $(\tau-1,0.6)$-down-persistent.
	\end{enumerate}
\end{lemma}


\begin{proof} We use \Lem{blowup} to define the following process generating a sequence of violation subgraphs. The initial graph is $G_0 = (\bX_0,\bY_0,E_0)$ which is the seed regular violation subgraph obtained from~\Lem{peeling}.
	
	\begin{framed}
		\noindent For each $i \geq 1$:
		\begin{enumerate} 
			\item Obtain $G_{i-1}'$ by removing all vertices from $\bX_{i-1} \cup \bY_{i-1}$ that are not $(\tau-1,\log^{-5} d)$-up-persistent. 
			\item Invoke \Lem{blowup} with walk length $\tau-1$ on $G_{i-1}'$ to obtain $G_{i} = (\bX_i,\bY_i,E_i)$.
			\item If $G_{i}$ satisfies the down persistence condition of \Lem{blowup} then halt and return $G_i$.
			\item If $G_i$ satisfies the blowup condition of \Lem{blowup}, then continue.
		\end{enumerate}
	\end{framed}
	
	By \Lem{blowup}, if the process does not halt on step $i$, then we have the following recurrences. (Recall that by \Cref{obs:oo} we have $m(G) \geq \dtal{G}$ and so using item (a) of \Cref{lem:seed} and our assumption that $\eps \geq d^{-1/2}$ we have $m(G) \geq d^{-1/2-c\delta} n^d$. Thus, $\floor{2(1-3\log^{-4}d) \cdot m(G)} \geq 2(1-3\log^{-3}d) \cdot m(G)$ clearly holds.) 

	\begin{itemize}
		\item $m(G_i) \geq 2(1-3\log^{-3}d) \cdot m(G_{i-1}')$,
		\item $D_{G_i}(\bX_i) \leq D_{G_{i-1}}(\bX_{i-1})$, $\Gamma_{G_i}(\bX_i) \leq \Gamma_{G_{i-1}}(\bX_{i-1})$, and
		\item $D_{G_i}(\bY_i) \leq 2 D_{G_{i-1}}(\bY_{i-1})$, $\Gamma_{G_i}(\bY_i) \leq 2 \Gamma_{G_{i-1}}(\bY_{i-1})$.
	\end{itemize}
	\noindent
	Furthermore, we have the following claim that bounds the number of edges lost in step (1).
	
	\begin{claim} For every $i \geq 1$, we have $m(G_{i-1}') \geq m(G_{i-1}) - d^{-2c\delta} \cdot 2^{i-1} \cdot m(G)$. \end{claim}
	
	\begin{proof} By \Clm{G0_persistence_X}, the number of vertices we remove from $\bX_{i-1}$ in step (1) is at most $d^{-6c\delta} \cdot |\bX|$ and by \Clm{G0_persistence_Y} the number of vertices we remove from $\bY_{i-1}$ in step (1) is at most $d^{-6c\delta} \cdot |\bY|$. The number of edges we remove by deleting these vertices from $\bY_{i-1}$ is at most
		\begin{align} \label{eq:removal}
			d^{-6c\delta}|\bY|D_{G_{i-1}}(\bY_{i-1}) \leq d^{-6c\delta} 2^{i-1}|\bY| D_G(\bY) \leq d^{-3c\delta} 2^{i-1}m(G) 
		\end{align}
		where in the second inequality we used $D_G(\bY) \leq \Phi_G(\bY)\Gamma_G(\bY)$ and the regularity property on $G$ (item (b) of \Lem{seed}).
		
		An analogous argument bounds the number of removed edges when we delete non-persistent vertices from $\bX_{i-1}$. Thus the total number of edges removed is at most $d^{-2c\delta} 2^{i-1} m(G)$. \end{proof}
	
	\begin{claim} If $i \leq \log^3 d$ and the process has not halted by step $i$, then $m(G_{i}) \geq \Omega(2^{i} m(G))$. \end{claim}
	
	\begin{proof} For brevity, let $\alpha = 2(1-3\log^{-3}d)$ and $\beta = d^{-2c\delta}m(G)$. Using the above bounds, we get the recurrence 
		\[
		m(G_{i}) \geq \alpha \cdot m(G_{i-1}') \geq \alpha(m(G_{i-1}) - \beta 2^{i-1}) \text{.}
		\]
		Expanding this recurrence yields $m(G_{i}) \geq \alpha^i m(G) - \beta \sum_{j=1}^i \alpha^j \cdot 2^{i-j}$. Observe that the subtracted term can be bounded as 
		\[
		\beta \sum_{j=1}^i \alpha^j \cdot 2^{i-j} = d^{-2c\delta}2^i m(G) \sum_{j=1}^i (1-3\log^{-3}d)^j \leq d^{-c\delta} 2^i m(G)
		\]
		simply using the fact that $i \leq \log^3 d \ll d^{c\delta}$. The first term is
		\[
		\alpha^i m(G) = 2^i(1-3\log^{-3} d)^i m(G) \geq C \cdot 2^i m(G)
		\]
		for some constant $C$. Combining the above two bounds completes the proof. \end{proof}
	
	\begin{claim} \label{clm:halt} The above process halts in $s \leq \log^{3} d$ iterations. \end{claim}
	
	\begin{proof} Suppose that the above process has not halted by step $i = \log^3 d$. By the previous claim, the number of edges in $G_i$ is at least $C \cdot 2^i m(G) = C \cdot d^{\log^2 d} m(G)$ for some constant $C$. By~\Cref{obs:oo}, note that $m(G) \geq \dtal{G}$ and thus is $\geq \eps \cdot d^{-c\delta} \cdot n^d$ by item (a) of \Lem{seed}. Thus, the number of edges in $G_i$ is at least $C \cdot \eps \cdot d^{\log^2 d - c\delta} n^d$. Note that the total number of edges in the fully augmented hypergrid is at most $nd \cdot n^d$. Moreover, recall that we are assuming $nd \leq d^{c}$ and $\eps \geq d^{-1/2}$. Therefore, $m(G_i) \gg nd \cdot n^d$ and this is a contradiction. \end{proof}
	
	By \Clm{halt} and \Lem{blowup}, the process halts in some $s \leq \log^3 d$ number of steps producing $G_s(\bX_s,\bY_s,E_s)$ with the following properties. (Recall that by \Cref{obs:oo} we have $m(G) \geq \dtal{G}$ and so using item (a) of \Cref{lem:seed} and our assumption that $\eps \geq d^{-1/2}$ we have $m(G) \geq d^{-1/2-c\delta} n^d$. Thus, $\floor{\frac{m(G)}{\log^5 d}} \geq \frac{m(G)}{\log^6 d}$ clearly holds.)
	\begin{itemize}
		\item $m(G_s) \geq 2^s \cdot \frac{m(G)}{\log^6 d}$.
		\item All vertices in $\bX_s$ are $(\tau-1,0.6)$-down-persistent.
		\item $\Gamma_{G_s}(\bX_s) \leq \Gamma_G(\bX)$ and $\Gamma_{G_s}(\bY_s) \leq 2^s \Gamma_G(\bY)$.
		\item $D_{G_s}(\bX_s) \leq D_G(\bX)$ and $D_{G_s}(\bY_s) \leq 2^s D_G(\bY)$.
	\end{itemize}
	Note that by \Lem{blowup} and item (c) of \Lem{seed}, all vertices in $G_1,\ldots,G_s$ are $(98+\frac{s \tau}{\sqrt{d}})$-typical. Moreover, by our choice of $\tau$, we have $s \tau \ll \sqrt{d}$ and so all vertices in $G_1,\ldots,G_s$ are $99$-typical.
	
	One last time, we remove all vertices in $\bX_s \cup \bY_s$ that are not $(\tau-1,\log^{-5} d)$-up-persistent and obtain our final graph $H(\bA,\bB,E)$. Using a similar argument made above in \Eqn{removal}, the number of edges that are removed by deleting the non-persistent vertices from $\bY_s$ is at most
	\begin{align*}
		d^{-6c\delta}|\bY|D_{G_s}(\bY_{s}) \leq 2^{s} d^{-6c\delta} |\bY| D_G(\bY) \leq 2^{s}d^{-3c\delta}m(G) \leq d^{-3c\delta} m(G_s) \log^6 d \leq d^{-2c\delta} m(G_s)
	\end{align*}
	and an analogous argument bounds the number of edges lost when we remove the non-persistent vertices from $\bX_s$. Thus we have $m(H) \geq m(G_s)(1-2d^{-2c\delta}) \geq 2^s \frac{m(G)}{\log^7 d}$ and this completes the proof of \Lem{subgraph_after_blowup}. \end{proof}

\subsection{Using Red/Blue Lemma to Obtain the Final Red or Blue Nice Subgraph} \label{sec:derive_redblue}
In this section, we prove \Lem{good_subgraph} using the violation subgraph $H(\bA,\bB,E)$ obtained in the previous section (\Lem{subgraph_after_blowup}).

We need a key `` red/blue lemma" that shows the existence of a violation subgraph with sufficiently many colored edges.
If we have our hands on a large violation subgraph $G$ with few red edges (but has some other properties), then we can find another comparable sized violation subgraph $H$
all of whose edges are blue, and whose maximum degrees are bounded by those in $G$. The precise statement is given below. We defer the proof of this lemma to~\Sec{redblue}.

\begin{restatable}[Red/Blue Lemma] {lemma}{redbluelemma}\label{lem:redblue}~
Let $G(\bX,\bY,E)$ be a violation subgraph and $1 \leq \ell \leq \sqrt{d}/\log^5 (d/\eps)$ be a walk length such that the following hold.
	\begin{asparaenum}
		\item At most half the edges are red for walk length $\ell$.
		\item All vertices in $\bX \cup \bY$ are $(\ell, \log^{-5} d)$-up-persistent.
		\item All vertices in $\bX \cup \bY$ are $99$-typical.
	\end{asparaenum}
	Then there exists another violation subgraph $H(\bL,\bR,E')$ such that
	\begin{asparaenum}
		\item All edges are blue for walk length $\ell$ and $m(H) \geq \floor{m(G)/6}$.    
		\item $\Gamma_H(\bL) \leq \Gamma_G(\bX)$ and $\Gamma_H(\bR) \leq \Gamma_G(\bY)$.
		\item $D_H(\bL) \leq D_G(\bX)$ and $D_H(\bR) \leq D_G(\bY)$.\footnote{We remark that the inequalities $\Gamma_H(\bL) \leq \Gamma_G(\bX)$ and $D_H(\bL) \leq D_G(\bX)$ are not used later in our proof, but are still stated for the purpose of providing useful context to the reader.}
	\end{asparaenum}
\end{restatable}

	
We apply the above lemma to construct the final nice red or blue subgraph. 
We split into two cases depending on how many edges in $H$ are red.

\subsubsection{Case 1: At least half the edges of $H$ are red}
	
	In this case, we consider the graph $H_1(\bA,\bB,E')$ obtained by simply removing all the non-red edges from $H$. 
	We claim that $H_1$ makes progress towards a $(\sigma_1, \tau)$-nice red subgraph (\Cref{def:red-nice}).
	Condition (a) holds by definition. Condition (b) is satisfied due to~\Cref{lem:subgraph_after_blowup}, condition (5).
	Condition (e) is satisfied because $\ceil{\sigma_{\bY} d^{0.5 - 7c\delta}} \geq \tau - 1\geq 0.5 \ceil{\sigma_{\bY} d^{0.5 - 7c\delta}}$ (recall \Cref{sec:choice-tau}) and $1 \geq \sigma_{\bY} \geq \sigma_{\bX} = \sigma_1$.
	We need to establish condition (c) and (d). That is, we need to establish
	\begin{enumerate}
			\item[(c)] $\sigma_{\bX}\cdot \Phi_{H_1}(\bx) \leq \sqrt{d}$ for all $\bx \in \bA$
			\item[(d)] $\sigma_{\bX}\sum_{\bx \in \bA} \Phi_{H_1}(\bx) \geq \eps^2 \cdot n^d \cdot d^{-6c\delta}$
		\end{enumerate}
	\noindent
	Let $\bA' \subseteq \bA$ be the set of vertices $\bx\in \bA$ for which $\Phi_{H_1}(\bx) > \frac{\sqrt{d}}{\sigma_{\bX}}$. 
	If $|\bA'| \geq d^{-5c\delta}|\bX|$, then consider $H_1'(\bA', \bB, E'')$ obtained by deleting all vertices not in $\bA'$ from $\bA$, then for each $\bx \in \bA'$, removing out-going edges from $\bx$ so that $\Phi_{H_1'}(\bx) = \frac{\sqrt{d}}{\sigma_{\bX}}$.
	Conditions (a) and (b) still hold, and (c) holds by construction of $\bA'$ and $H_1'$. Note that condition (e) clearly still holds since we have not modified $\sigma_1$ or $\tau$. Furthermore, 
	\[
		\sum_{\bx \in \bA'} \Phi_{H_1'}(\bx) \geq d^{-5c\delta}|\bX| \cdot \frac{\sqrt{d}}{\sigma_{\bX}} ~~\Rightarrow~~ \sigma_{\bX}	\sum_{\bx \in \bA'} \Phi_{H_1'}(\bx)  \geq d^{-5c\delta} \cdot \frac{\eps}{d^{1/2+c\delta}} \cdot n^d \cdot \sqrt{d} = \eps \cdot n^d \cdot d^{-6c\delta} 
	\]
	where we used~\Cref{lem:seed}, part (d) for the lower bound on $|\bX|$. Note that this implies something slightly stronger than condition (d) above (the exponent of $\eps$ is $1$).
	
	Therefore, we may assume $|\bA'| \leq d^{-5c\delta}|\bX|$. In this case, let $H_1 = (\bA\setminus \bA', \bB, E')$ where we simply remove the $\bA'$ vertices.
	The number of edges this destroys is at most
	\[
	d^{-5c\delta}D(\bA)|\bX| \leq d^{-5c\delta} D(\bX) |\bX| \leq d^{-2c\delta}  m(G) \leq d^{-c\delta}m(H)
	\]
	where in the second inequality we used $D(\bX) \leq \Phi(\bX)\Gamma(\bX)$ and the regularity property (\Cref{lem:seed}, property (b)) of $G$. 
	Thus, the number of edges we have discarded is negligible, and condition (c) holds. In particular, the number of edges in $H_1$ is at least $m(H)/3$.
	We now prove condition (d) also holds. 

	\begin{claim} $\sigma_{\bX} \sum_{\bx\in \bA\setminus \bA'} \Phi_{H_1}(\bx) \geq \eps^2 \cdot n^d \cdot d^{-6c\delta}$. \end{claim}
	
	\begin{proof} 
		For any $\bx\in \bA\setminus \bA'$, we have $\Phi_{H_1}(\bx) \geq \frac{D(\bx)}{\Gamma(\bx)}$ and thus $\sum_{\bx \in \bA\setminus \bA'} \Phi_{H_1}(\bx) \geq \frac{m(H)/3}{\Gamma(\bA)}$.
		Since $\Gamma(\bA) \leq \Gamma(\bX)$ we have
		\begin{align} \label{eq:PhiAPhiX'}
			&	\sum_{\bx \in \bA\setminus \bA'} \Phi_{H_1}(\bx) \geq \frac{m(H)}{3\Gamma(\bA)} \geq \frac{2^s \cdot m(G)}{3 \Gamma(\bX) \log^7 d} \geq \frac{d^{-3c\delta}|\bX|\Phi(\bX)\Gamma(\bX)}{3 \Gamma(\bX) \log^7 d} \nonumber \\ 
			&\geq d^{-4c\delta}|\bX|\Phi(\bX) \geq d^{-4c\delta} \sum_{\bx \in \bX} \Phi_G(\bx)\text{.} 
		\end{align}
	where in the second inequality we used \Cref{lem:subgraph_after_blowup}, property (1) to lower bound the number of edges in $H$ with that of $G$.
	In the third inequality we used the regularity property (property (b) of~\Cref{lem:peeling}), in the fourth we used $d^{c\delta} \gg 2\log^7 d$ for large enough $d$, and the fifth inequality 
	uses the trivial upper bound $\Phi(\bX) \geq \Phi_G(\bx)$ for all $\bx \in \bX$.
	
	Now we apply the fact (Item (a) of \Cref{lem:peeling}) that $\dtal{G}$ is large. Using the coloring $\chi \equiv 1$ for edges in $G$, we get
	\[
 \sum_{\bx \in \bX} \sqrt{\Phi_G(\bx)} \geq \dtal{G} \geq \eps \cdot d^{-c\delta} \cdot n^d ~~\Rightarrow~~ \Exp_{\bx\in \bX} [\sqrt{\Phi_G(\bx)}] \geq \frac{\eps\cdot d^{-c\delta}}{\sigma_{\bX}}
	\]
	Jensen's inequality gives 
	\[
		\Exp_{\bx\in \bX} [\Phi_G(\bx)] \geq \frac{\eps^2 \cdot d^{-2c\delta}}{\sigma^2_{\bX}} ~~\Rightarrow~~ \frac{\sigma^2_{\bX}}{|\bX|} \sum_{\bx \in \bX} \Phi_G(\bx) \geq \eps^2 d^{-2c\delta}
		~~\Rightarrow \sigma_{\bX} \sum_{\bx \in \bX} \Phi_G(\bx) \geq \eps^2 d^{-2c\delta} n^d
	\]
	Plugging into  \Eqn{PhiAPhiX'} proves the claim. \end{proof}
%
%
%

	\subsubsection{Case 2: At most half the edges of $H$ are red}
	
	In this case we invoke the Red/Blue lemma, \Lem{redblue} to obtain a violation subgraph $H_2 = (\bL,\bR,E')$ with the following key properties. (Recall that by \Cref{obs:oo} we have $m(G) \geq \dtal{G}$ and so using item (a) of \Cref{lem:seed} and our assumption that $\eps \geq d^{-1/2}$ we have $m(G) \geq d^{-1/2-c\delta} n^d$. Thus, $\floor{m(G)/6} \geq m(G)/7$ clearly holds.)
	\begin{itemize}
		\item[(P1)] All edges are blue and $m(H_2) \geq 2^s \frac{m(G)}{7 \log^7 d}$.
		\item[(P2)] $\Gamma(\bR) \leq \Gamma(\bB) \leq 2^s \cdot \Gamma(\bY)$.
		\item[(P3)]  $D(\bR) \leq D(\bB) \leq 2^s \cdot D(\bY)$.
	\end{itemize}
	We claim that $H_2$ makes progress towards a $(\sigma_2, \tau)$-nice blue subgraph (\Cref{def:blue-nice}).
Condition (a) holds by definition. 
Condition (d) is satisfied because $\tau \geq 0.5 \sigma_{\bY} d^{0.5 - 7c\delta}$ and $\sigma_{\bY} =\sigma_2$.
We need to establish condition (b) and (c). That is, we need to establish
\begin{enumerate}
	\item[(b)] $\sigma_{\bY}\cdot \Phi_{H_2}(\by) \leq \sqrt{d}$ for all $\by \in \bR$
	\item[(c)] $\sigma_{\bY}\sum_{\by \in \bR} \Phi_{H_2}(\by) \geq \eps^2 \cdot n^d \cdot d^{-6c\delta}$
\end{enumerate}
As in Case 1, we begin by removing high degree vertices. Let $\bR' \subseteq \bR$ be the vertices $\by\in \bR$ which have  $\Phi_{H_2}(\by) > \frac{\sqrt{d}}{\sigma_{\bY}}$. 
If $|\bR'| \geq d^{-5c\delta}|\bY|$, then we would just focus on $H_2'(\bR', \bL, E'')$ obtained by removing all vertices not in $\bR'$ from $\bR$, then for each $\by \in \bR'$, removing in-coming edges so that $\Phi_{H_2'}(\by) = \frac{\sqrt{d}}{\sigma_{\bY}}$. Then, $H_2'$ satisfies (b) and (c) for a very similar reason as in Case 1, and after these modifications (a) and (d) still hold as well.
Thus, we may assume $|\bR'|$ is smaller than $d^{-5c\delta} |\bY|$ and we define $H_2(\bL, \bR\setminus \bR', E')$, and this leads to a negligible decrease in the number of edges.
Condition (b) holds by design, and the proof that condition (c) holds is similar. We provide it for completeness.
\noindent

	\begin{claim} $\sigma_{\bY} \sum_{\by\in \bR\setminus \bR'} \Phi_{H_2}(\by) \geq \eps^2 \cdot n^d \cdot d^{-6c\delta}$. \end{claim}

\begin{proof} 
	For any $\by\in \bR\setminus \bR'$, we have $\Phi_{H_2}(\by) \geq \frac{D(\by)}{\Gamma(\by)}$ and thus $\sum_{\by \in \bR\setminus \bR'} \Phi_{H_2}(\by) \geq \frac{m(H)/3}{\Gamma(\bR)}$.
	Since $\Gamma(\bR) \leq 2^s \cdot \Gamma(\bY)$ we have
	\begin{align} \label{eq:PhiAPhiX}
		&	\sum_{\by \in \bR\setminus \bR'} \Phi_H(\by) \geq \frac{m(H)}{3\Gamma(\bR)} \geq \frac{2^s \cdot m(G)}{2^s \cdot 21 \Gamma(\bY) \log^7 d} \geq \frac{d^{-3c\delta}|\bY|\Phi(\bY)\Gamma(\bY)}{21 \Gamma(\bY) \log^7 d} \nonumber \\ 
		&\geq d^{-4c\delta}|\bY|\Phi(\bY) \geq d^{-4c\delta} \sum_{\by \in \bY} \Phi_G(\by)\text{.} 
	\end{align}
	where in the second inequality we used~\Cref{lem:subgraph_after_blowup}, part 1, to lower bound the number of edges in $H$ with that of $G$, the original seed graph from~\Cref{lem:peeling}.
	In the third inequality we used the regularity property (property 2 of~\Cref{lem:peeling}), in the fourth we used $d^{c\delta} \gg 21\log^7 d$ for large enough $d$, and the fifth inequality 
	uses the trivial upper bound $\Phi(\bY) \geq \Phi_G(\by)$ for all $\by \in \bY$.

	The rest of the proof is the same as Case 1 except we apply the coloring $\chi \equiv 0$ for edges in $G$. We omit this very similar calculation.
\end{proof}

\noindent
These two cases conclude the proof of~\Cref{lem:good_subgraph}. 
All that remains is to prove the Red/Blue lemma,~\Cref{lem:redblue} and the Persist-or-Blow-up lemma,~\Cref{lem:blowup}.
We prove these in the subsequent two sections, and both of these use the translation of violation subgraphs idea.

\section{Proof of the Red/Blue Lemma, \Lem{redblue}}\label{sec:redblue}
Let us recall the red/blue lemma.

\redbluelemma*
\begin{proof} We first recall the definition of $p_{\bx,\ell}(\bx')$ in \Def{walkpdf}. For a fixed $\bx$, consider the process of sampling a hypercube $\bH \sim \HH(\bx)$ and then sampling $\bz \sim \cU_{\bH,\ell}(\bx)$. Recall from \Fact{dists} that this is one of three equivalent ways of expressing our random walk distribution. Given $\bx,\bx',\ell$, we have
\[
p_{\bx,\ell}(\bx') = \Pr\left[\bx,\bx' \in \bH_{100} \text{ and } \bz = \bx'\right] \text{.}
\]
Since these are probabilities, we trivially obtain	
\begin{equation}\label{eq:prob-triv}
	\text{For any $\bx, \bx'$}, ~~\sum_{\bx' \in \bX} p_{\bx,\ell}(\bx') \leq 1~~\text{and}~~\sum_{\bx \in \bX} p_{\bx',\ell}(\bx) \leq 1
\end{equation}

We use these probabilities  to set up a flow problem as follows.
Recall the definition of red and blue edges (\Cref{def:red} and~\Cref{def:blue}).
Let $B$ denote the set of all edges in the fully augmented hypergrid that are blue for walk length $\ell$. 
For every non-red edge $(\bx,\by)$ of $G$ and every shift $\bs \in \supp(\upshift{\bx}{\ell})$, if the edge $e = (\bx + \bs, \by + \bs)$ is blue, then we put $\flow(e) := p_{\bx,\ell}(\bx+\bs)$ units of flow on $e$. Note that this value is also equal to $p_{\by, \ell}(\by + \bs)$ since both $\bx$ and $\by$ are $99$-typical and any choice that takes $\bx$ to $\bx+\bs$ can be coupled
with one which takes $\by$ to $\by+\bs$.

\begin{claim} \label{clm:flow} Every non-red edge of $G$ inserts at least $0.93$ units of flow on the edges in	 $B$. \end{claim}

\begin{proof} Fix a non-red edge $(\bx,\by)$, and let $i$
denote its dimension. Generate $\bH \sim \HH(\bx)$ and $\bs \sim \cU\cS_{\bH,\ell}(\bx)$. Note that it is equivalent to directly sample $\bs \sim \upshift{\bx}{\ell}$.
We then consider the random edge $e = (\bx + \bs, \by + \bs)$. We set $\bx' = \bx + \bs$ and $\by' = \by + \bs$. Let us define the following
series of events. (i) $\cE_1$: $\bs_i = 0$. (ii) $\cE_2$: $f(\bx') = 1$. (iii) $\cE_3$: $f(\by') = 0$. (iv) $\cE_4$: at least half of $I(\bx',\by')$ is \emph{not} $\ell$-mostly-zero-below, (v) $\cE_5$: $\bx,\bx' \in \bH_{100}$. We will show that whenever $\cE_2$, $\cE_3$, and $\cE_4$ occur, the edge $(\bx',\by')$ is blue by definition. Therefore, recalling the definition of $p_{\bx,\ell}(\bx')$, the edge $(\bx,\by)$ inserts at least $\Pr[\wedge_{j=1}^5 \cE_j]$ units of flow in $B$. Subsequently, we will show that the probability of this event is at least $0.95$ and this will prove the claim.

Since $\norm{\bs}_0 \leq \ell \leq \sqrt{d}$, we have $\Pr[\cE_1] \geq 1-1/\sqrt{d}$. Since $\bx$ is $(\ell,\log^{-5} d)$-up-persistent, $\Pr[\cE_2] \geq 1-\log^{-5} d$. 
Since $\by$ is $(\ell,\log^{-5} d)$-up-persistent, $\Pr[\cE_3] \geq 1-\log^{-5} d$. By a union bound we have
\[
	\Pr[\neg \cE_1 \vee \neg \cE_2 \vee \neg \cE_3] \leq \Pr[\neg \cE_1] + \Pr[\neg \cE_2] + \Pr[\neg \cE_3 ~|~ \cE_1] \leq \frac{1}{\sqrt{d}} + \frac{1}{\log^5 d} + \frac{1}{\log^5 d} \leq 3 \log^{-5} d
\]
and so
\begin{align} \label{eq:prob123}
    \Pr[\cE_1 \wedge \cE_2 \wedge \cE_3] \geq 1-3\log^{-5} d \text{.}
\end{align}
To deal with $\cE_4$, we bring in the non-redness of our edge $(\bx,\by)$. By definition, 
\[
\Pr_{\bz \in I(\bx,\by)} \Pr_{\bz' \sim \updist{\bz}{\ell}}[\bz' \ \textrm{is not $\ell$-$\mzb$}] \geq 0.99
\]
In terms of shifts, we can express this bound as
\[
\Pr_{\bz \in I(\bx,\by)} \Pr_{\bs \sim \upshift{\bz}{\ell}}[\bz+\bs \ \textrm{is not $\ell$-$\mzb$}] \geq 0.99
\]
Since the probability of $\cE_1$ is at least $1-o(1)$, we have
\[
\Pr_{\bz \in I(\bx,\by)} \Pr_{\bs \sim \upshift{\bz}{\ell}}[\bz+\bs \ \textrm{is not $\ell$-$\mzb$} ~|~ \cE_1] \geq 0.98 
\]
Note that conditioned on $\cE_1$, the distributions $\upshift{\bz}{\ell}$ and $\upshift{\bx}{\ell}$ are identical. Hence,
\[
\Pr_{\bs \sim \upshift{\bx}{\ell}} \Pr_{\bz \in I(\bx,\by)} [\bz+\bs \ \textrm{is not $\ell$-$\mzb$} ~|~ \cE_1] \geq 0.98
\]
Let $X_\bs$ be the fraction of points in $I(\bx+\bs,\by+\bs)$ that are not $\ell$-$\mzb$. By linearity
of expectation, $\EX_\bs[X_{\bs} ~|~ \cE_1] \geq 0.98$. Hence $\EX_{\bs}[1-X_{\bs} ~|~ \cE_1] \leq 0.02$
and by Markov's inequality, $\Pr_{\bs}[1-X_{\bs} > 0.5 ~|~ \cE_1] \leq 2/50$. Hence, $\Pr_{\bs}[X_\bs \geq 0.5 ~|~ \cE_1] \geq 48/50 = .96$. Since $\Pr[\cE_1] = 1-o(1)$, we have $\Pr[\cE_4] = \Pr_{\bs}[X_{\bs} \geq 0.5] \geq 0.95$. 

Combining with \Eqn{prob123}, we have $\Pr[\wedge_{j=1}^4 \cE_j] \geq 0.94$. When $\wedge_{j=1}^4 \cE_j$ occurs, the edge $(\bx',\by')$ is a violated edge and at least half of $I(\bx',\by')$ is not $\ell$-$\mzb$. For $\bz' \in I(\bx',\by')$ that is not $\ell$-$\mzb$, by definition $\Pr_{\bw \sim \downdist{\bz'}{\ell}} [f(\bw) = 1] \geq 0.1$.
Hence,
\[
\Pr_{\bz' \in_R I(\bx',\by')} \Pr_{\bw \sim \downdist{\bz'}{\ell}} [f(\bw) = 1] \geq 0.5 \times 0.1 \geq 0.01
\]
We conclude that $(\bx',\by')$ is blue, whenever $\wedge_{j=1}^4 \cE_j$ occurs.

Stepping back, with probability at least $0.94$ over the shift $\bs \sim \upshift{\bx}{\ell}$, the edge $(\bx+\bs,\by+\bs)$ is blue. Finally, since all points in $\bX$ are $99$-typical, we have
$\Pr[\bx \in \bH_{99}] \geq 1 - (\eps/d)^5$, and conditioned on this event we have $\bx' \in \bH_{100}$ since $\ell \ll \sqrt{d}$. 
Together, we get
 $\Pr[\cE_5] \geq 1 - (\eps/d)^5 \geq 0.99$. Thus, by a union bound $\Pr[\wedge_{j=1}^5 \cE_j] \geq 0.93$ and so the amount of flow that $(\bx,\by)$ inserts is at least $0.93$. \end{proof}

Let $E' \subseteq B$ denote the set of blue edges which receive non-zero flow. Let $H(\bL,\bR,E')$ denote the bipartite graph on these edges. 
Since $\ell \leq \sqrt{d}/\log^5 (d/\eps)$, by the reversibility \Lem{prob-rev}, $p_{\bx,\ell}(\bx') \leq 2p_{\bx',\ell}(\bx)$ for any $\bx \in \bX$, $\bx' \in \bL$ and $p_{\by,\ell}(\by') \leq 2 p_{\by',\ell}(\by)$ for any $\by \in \bY$, $\by' \in \bR$. Using this bound we're able to establish the desired capacity constraints on the flow as follows\footnote{We replaced the upper bound $1+\frac{1}{\log^3 d}$ by ``$2$'' as that serves our purpose; hopefully, the reader is not confused.
}
\begin{claim} [Edge Congestion] \label{clm:flow-edge} The total flow on an edge $(\bx',\by') \in B$ is at most $2$. \end{claim}

\begin{proof} By construction, the total flow on an edge $(\bx',\by')$ is at most 
\[
\sum_{\bx \in \bX} p_{\bx,\ell}(\bx') \leq 2 \sum_{\bx \in \bX} p_{\bx',\ell}(\bx) \leq 2
\] 
since by \eqref{eq:prob-triv} we have $\sum_{\bx \in \bX} p_{\bx',\ell}(\bx) \leq 1$. \end{proof}

\begin{claim} [Vertex Congestion] \label{clm:flow-vert} The following hold.
\begin{asparaenum}
    \item The total amount of flow through a vertex $\bx' \in \bL$ is at most $2D_G(\bX)$.
    \item The total amount of flow through a vertex $\by' \in \bR$ is at most $2D_G(\bY)$.
    \item For all $i \in [d]$, the total amount of $i$-flow through a vertex $\bx' \in \bL$ is at most $2\Gamma_{G,i}(\bX)$.
    \item For all $i \in [d]$, the total amount of $i$-flow through a vertex $\by' \in \bR$ is at most $2\Gamma_{G,i}(\bY)$.
\end{asparaenum}
\end{claim}
\begin{proof} 
	Fix a vertex $\bx'\in \bL$ and consider the edges $(\bx',\by'_t)$ of $E' = E(H)$ incident on it with $t = 1, \ldots, d'$.
	For $(\bx',\by'_t)$ to receive flow, there must exist edges of the form $(\bx,\by) \in E(G)$ such that $\bx' = \bx + \bs$ and $\by'_t = \by + \bs$ for some shift $\bs$.
	Call such an $(\bx,\by)$ parallel to $(\bx', \by'_t)$ and denote it as $(\bx,\by)~||~(\bx',\by'_t)$. Note that the same $(\bx,\by)\in E(G)$ can be parallel 
	to at most one edge incident on $\bx'$ since fixing $\bx,\by$ and $\bx'$ fixes the $\by'_t$.
	The total flow through vertex $\bx' \in \bL$ is therefore
	\[
		\flow(\bx') = \sum_{t=1}^{d'} \flow(\bx',\by'_t) = \sum_{t=1}^{d'} \sum_{(\bx,\by) \in E~||~(\bx', \by'_t)} p_{\bx,\ell}(\bx') \leq \sum_{(\bx,\by) \in E(G)} p_{\bx, \ell}(\bx')
	\]
This can be upper-bounded as follows
\begin{align*}
\sum_{(\bx,\by) \in E} p_{\bx,\ell}(\bx') &\leq D_G(\bX) \sum_{\bx \in \bX} p_{\bx,\ell}(\bx')  \ \ \ \textrm{(since degree of $\bx \in \bX$ in $G$ is $\leq D_G(\bX)$)} \\
&\leq 2D_G(\bX) \sum_{\bx \in [n]^d} p_{\bx',\ell}(\bx) \ \ \ \textrm{(since $p_{\bx,\ell}(\bx') \leq 2 p_{\bx',\ell}(\bx)$)} \\
&\leq 2D_G(\bX) \ \ \ \textrm{(since $\sum_{\bx \in [n]^d} p_{\bx',\ell}(\bx) \leq 1$ due to \eqref{eq:prob-triv})}
\end{align*}
Analogous arguments prove all the rest. To see (2), we work from the ``right'' side. For completeness, 
fix $\by'\in \bR$ and consider the edges $(\bx'_t, \by')$ of $E'$ incident on it.
For such an edge to receive flow, there must exist edges of the form $(\bx_t, \by) \in E(G)$ with $\by' = \by + \bs$ and $\bx'_t = \bx_t + \bs$.
The total flow coming into $\by'\in \bR$ is therefore
	\[
\flow(\by') = \sum_{t=1}^{d'} \flow(\bx'_t,\by') = \sum_{t=1}^{d'} \sum_{(\bx_t,\by) \in E~||~(\bx'_t, \by')} p_{\by,\ell}(\by') \leq \sum_{(\bx,\by) \in E(G)} p_{\by, \ell}(\by')
\]
The rest of the argument is same as above. For (3) and (4), we restrict the above argument only to $i$-edges and note that any edge parallel to an $i$-edge is also an $i$-edge.
\end{proof}





By \Clm{flow} and the fact that at least half the edges in $G$ are not red, the total amount of flow is at least $m(G)/3$ and this flow satisfies the constraints listed in \Clm{flow-edge} and \Clm{flow-vert}. Thus, dividing by $2$ yields a flow of value $m(G)/6$ satisfying the following. 
\begin{enumerate}
    \item[C1.] The flow on every edge is at most $1$.
    \item[C2.] The total flow through any vertex in $\bL$ is at most $D_G(\bX)$. The total $i$-flow through any vertex in $\bL$ is at most $\Gamma_{G,i}(\bX)$.
    \item[C3.] The total flow through any vertex in $\bR$ is at most $D_G(\bY)$. The total $i$-flow through any vertex in $\bR$ is at most $\Gamma_{G,i}(\bY)$.
\end{enumerate}
By integrality of flow, there exists an integral flow of at least $\lfloor m(G)/6 \rfloor$ units satisfying the same capacity constraints. By item (C1) above, the integral flow is a subgraph containing at least $\floor{m(G)/6}$ edges and satisfying the desired constraints listed in the lemma statement. \end{proof}


%

\section{Proof of the `Persist-or-Blow-Up' Lemma, \Lem{blowup}} \label{sec:blowup}
Let us recall the `Persist-or-Blow-Up' lemma.

\blowuplemma*

\noindent
The proof strategy of this lemma is similar to the proof of~\Cref{lem:redblue} described in~\Cref{sec:redblue}.
We first recall the definition of $p_{\bx,\ell}(\bx')$ in \Def{walkpdf}. For a fixed $\bx$, consider the process of sampling a hypercube $\bH \sim \HH(\bx)$ and then sampling $\bz \sim \cU_{\bH,\ell}(\bx)$. Recall from \Fact{dists} that this is one of three equivalent ways of expressing our random walk distribution. Given $\bx,\bx',\ell$, we have
\[
p_{\bx,\ell}(\bx') = \Pr\left[\bx,\bx' \in \bH_{100} \text{ and } \bz = \bx'\right] \text{.}
\]
We use these values to set up a flow problem as follows. 
For every edge $(\bx,\by)$ of $G$ and $\bs \in \supp(\upshift{\bx}{\ell})$, if $e = (\bx + \bs, \by + \bs)$ is a violation, then we put $\flow(e) := p_{\bx,\ell}(\bx+\bs)$ units of flow on the edge $e$.
As argued in~\Cref{sec:redblue}, this is the same as $p_{\by, \ell}(\by+\bs)$.
Upon processing every edge of $G$, we get a fractional flow supported on edges of another bipartite violation subgraph
$G' = (\bX', \bY', E')$. Note that by \Clm{translate-typical}, we have that all vertices in $G'$ are $(c+\frac{\ell}{\sqrt{d}})$-typical.

\noindent
The next two claims are analogous to~~\Cref{clm:flow}, \Cref{clm:flow-edge} and~\Cref{clm:flow-vert}, respectively.	

\begin{claim} \label{clm:edge-flow} Every edge $(\bx, \by) \in E(G)$ inserts at least $1-\log^{-4} d$ units of flow. \end{claim}

\begin{proof} The proof of this claim is similar to that of~\Cref{clm:flow}.
	Fix an edge $(\bx,\by) \in E(G)$ and let this be an $i$-edge. 
	Generate $\bH \sim \mathbb{H}(\bx)$ and  a shift $\bs \sim \cU\cS_{\bH,\ell}(\bx)$, and let $\bx' = \bx + \bs$ and $\by' = \by + \bs$. 
	Consider the events: (i) $\calE_1$: $\bs_i = 0$, (ii) $\calE_2$: $f(\bx') = 1$, (iii) $\calE_3$: $f(\by') = 0$, (iv) $\calE_4$: $\bx, \bx' \in \bH_{100}$.
	Note that the total flow inserted by $(\bx, \by)$ is at least $\Pr[\wedge_{i=1}^4 \calE_i]$.
$\Pr[\calE_1] \geq 1-1/\sqrt{d}$, since $\norm{\bs}_0 \leq \ell \le \sqrt{d}$. 
Since $\bx,\by$ are both $(\ell,\log^{-5} d)$-up-persistent and $f(\bx) = 1$, $f(\by) = 0$, we get $\Pr[\calE_2], \Pr[\calE_3~|~\calE_1] \geq 1 - \frac{1}{\log^5 d}$.
Finally, since $\bx$ is $99$-typical, with probability $1-(\eps/d)^5$ we have $\bx \in \bH_{99}$ which implies $\bx'\in \bH_{100}$ since $\ell \ll \sqrt{d}$.
Thus by a union bound, $\Pr[\wedge_{i=1}^4 \calE_i] \geq 1 - 2\log^{-5} d - 1/\sqrt{d} - (\eps/d)^5 \geq 1-\log^{-4} d$.  \end{proof}

\begin{claim} [Edge Congestion] \label{clm:edge-cong} The flow on any edge $(\bx',\by')$ is at most $\sum_{\bx \in \bX} p_{\bx,\ell}(\bx') \leq (1+\log^{-3} d)$. \end{claim}

\begin{proof} Consider an edge $(\bx',\by')\in E(G')$, which receives flow from some $(\bx,\by)$ in $G$. Flow is inserted by translations of edges,
so $\by - \bx = \by' - \bx'$. Hence, for a given $\bx$, there exists a unique $\by$ such that $(\bx,\by)$ inserts flow on $(\bx',\by')$. 
By construction, the flow inserted is $p_{\bx,\ell}(\bx')$. Thus, the total flow that $(\bx',\by')$ receives is at most $\sum_{\bx \in \bX} p_{\bx,\ell}(\bx')$. The RHS bound holds by \Lem{prob-rev} and the observation \eqref{eq:prob-triv} that $\sum_{\bx \in \bX} p_{\bx',\ell}(\bx) \leq 1$  \end{proof}

\begin{claim} [Vertex Congestion] \label{clm:vert-cong} The following hold.
    \begin{asparaenum}
        \item For any $\bx' \in \bX'$, the total flow on edges incident to $\bx'$ is at most \[
        D_G(\bX) \sum_{\bx \in \bX} p_{\bx,\ell}(\bx') \leq D_G(\bX)(1+\log^{-3} d) \text{.}\]
        \item For any $\bx' \in \bX'$, the total $i$-flow on edges incident to $\bx'$ is at most \[\Gamma_{G,i}(\bX) \sum_{\bx \in \bX} p_{\bx,\ell}(\bx') \leq \Gamma_{G,i}(\bX)(1+\log^{-3} d) \text{.}\]
        \item For any $\by' \in \bY'$, the total flow on edges incident to $\by'$ is at most \[D_G(\bY) \sum_{\by \in \bY} p_{\by,\ell}(\by') \leq D_G(\bY)(1+\log^{-3} d) \text{.}\]
        \item For any $\by' \in \bY'$, the total $i$-flow on edges incident to $\by'$ is at most \[\Gamma_{G,i}(\bY) \sum_{\by \in \bY} p_{\by,\ell}(\by') \leq \Gamma_{G,i}(\bY)(1+\log^{-3} d)\text{.}\]
    \end{asparaenum}
\end{claim}

\begin{proof} Consider $\bx' \in \bX'$. All the $i$-flow inserted on edges incident to $\bx'$
comes from $i$-edges $(\bx,\by)$ in $G$. Every $i$-edge in $G$ inserts flow on at most a single edge incident to $\bx'$ and there are at most $\Gamma_{G,i}(\bX)$ $i$-edges incident to any vertex $\bx \in \bX$. Hence, the total $i$-flow inserted by a $\bx \in \bX$ through $\bx'$ is at most $\Gamma_{G,i}(\bX) \cdot p_{\bx,\ell}(\bx')$. Thus, summing over all $\bx \in \bX$ and using the reversibility \Lem{prob-rev} shows that the total $i$-flow on edges incident to $\bx'$ is at most
\[
\Gamma_{G,i}(\bX) \sum_{\bx \in \bX} p_{\bx,\ell}(\bx') \leq (1+\log^{-3} d)\Gamma_{G,i}(\bX) \sum_{\bx \in \bX} p_{\bx',\ell}(\bx) \leq (1+\log^{-3} d)\Gamma_{G,i}(\bX)
\]
and this proves (2). The proof of (1) is identical, with $D_G(\bX)$ replacing $\Gamma_{G,i}(\bX)$, and statements (3) and (4) have analogous proofs where we work
with $p_{\by,\ell}(\by')$'s instead. All this is analogous to the proof of~\Cref{clm:flow-vert}. \end{proof}

\noindent
Next we require a definition of heavy vertices. Recall, we now have a flow $\flow$ on $G' = (\bX', \bY', E')$.

\begin{definition} [Heavy Vertices] \label{def:heavy} A vertex $\bx' \in \bX'$ is called \emph{heavy} if it satisfies any of the following. 
\begin{asparaenum}
    \item[H1.] There is an edge $e' = (\bx',\by') \in E'$ with  $\flow(e') \geq 1/2$.
    \item[H2.] The total flow on edges incident to $\bx'$ is at least $D_G(\bX)/2$.
    \item[H3.] There exists $i \in [d]$ such that the total $i$-flow on edges incident to $\bx'$ is at least $\Gamma_{G,i}(\bX)/2$.
\end{asparaenum}
\end{definition}
\noindent
Note that heavy vertices are defined only in $\bX$ and not in $\bY$.
Note that by~\Cref{clm:edge-flow} and~\Cref{clm:vert-cong}, the upper bounds on (H1), (H2) and (H3) are at most $1+\log^{-3}d$, $D_G(\bX)(1+\log^{-3}d)$ and $\Gamma_{G,i}(\bX)(1+\log^{-3} d)$, respectively, and so a vertex is heavy if any of these upper bounds are reached up to ``factor 2''.

\begin{claim} \label{clm:heavy} All heavy vertices are $(\ell,0.6)$-down persistent. \end{claim}

\begin{proof} Consider a heavy vertex $\bx'$. That is, $\bx'$ satisfies one of the three conditions listed in \Def{heavy}. Suppose it satisfies the first condition: there is some violated edge $(\bx',\by')$ receiving at least $1/2$ units of flow. By \Clm{edge-cong}, the total flow on $(\bx',\by')$
is at most $\sum_{\bx \in \bX} p_{\bx,\ell}(\bx')$. Hence, $\sum_{\bx \in \bX} p_{\bx,\ell}(\bx') \geq 1/2$. In fact, observe that we can prove the exact same inequality if $\bx'$ satisfies the second or third condition of \Def{heavy}, by using the upper bound given by the LHS of items (1) and (2), respectively, of \Clm{vert-cong}. Now, applying the reversibility \Lem{prob-rev}, we have $(1+\log^{-3}d) \sum_{\bx \in \bX} p_{\bx',\ell}(\bx) \geq 1/2$.
Note that $f(\bx) = 1$ for all $\bx \in \bX$. Hence,
\begin{align}
    \Pr_{\bz \sim \cD_{\ell}(\bx')}[f(\bz) = 1] \geq \sum_{\bx \in \bX} p_{\bx',\ell}(\bx) \geq \frac{1}{2(1+\log^{-3} d)} \geq 0.4
\end{align}
and so $\bx'$ is $(\ell,0.6)$-down-persistent. \end{proof}

We are now set up to complete the proof. For convenience, we use $m$ to denote $m(G)$. We refer to the flow on edges incident to heavy vertices as the \emph{heavy flow}. 
We let $\bX'_H \subseteq \bX'$ be the subset of heavy vertices, and let $\bX'_L := \bX' \setminus \bX_H'$ be the subset of non-heavy vertices. 
We let $G_H = (\bX'_H,\bY'_H,E'_H)$ denote the bipartite graph of all edges incident to heavy vertices. 
We refer to the flow on edges incident to non-heavy vertices as the \emph{light flow}. We let $G_L = (\bX'_L,\bY'_L,E'_L)$ denote the bipartite graph of all edges incident to non-heavy vertices. We split into two cases based on the amount of heavy flow.

\subsection{Case 1: The total amount of heavy flow is at least $\frac{m}{\log^4 d}$}

Note that by \Clm{heavy}, all vertices in $\bX'_H$ are $(\ell,0.6)$-down persistent.
By \Clm{edge-cong} and \Clm{vert-cong}, the heavy flow satisfies the following capacity constraints.

\begin{enumerate}
    \item[A1.] The flow on every edge is at most $(1+\log^{-3} d)$.
    \item[A2.] For every $\bx' \in \bX'_H$, the total flow on edges incident to $\bx'$ is at most $D_G(\bX)(1+\log^{-3} d)$ and the total $i$-flow on edges incident to $\bx'$ is at most $\Gamma_{G,i}(\bX)(1+\log^{-3} d)$. 
    \item[A3.] For every $\by' \in \bY'_H$, the total flow on edges incident to $\by'$ is at most $D_G(\bY)(1+\log^{-3} d)$ and the total $i$-flow on edges incident to $\by'$ is at most $\Gamma_{G,i}(\bY)(1+\log^{-3} d)$. 
\end{enumerate}
\noindent
Recall that $\bX$ and $\bY$ were the bipartitions of the {\em original} graph $G$ and not of $G'$.

Let us divide the flow by $(1+\log^{-3}d)$. Thus, we now have at least $\frac{m}{(1+\log^{-3}d)\log^4 d} \geq \frac{m}{\log^5d}$ units
of flow satisfying the following capacity constraints. 
\begin{enumerate}
    \item[A'1.] The flow on every edge is at most one.
    \item[A'2.] For every $\bx' \in \bX'_H$, the total flow on edges incident to $\bx'$ is at most $D_G(\bX)$ and the total $i$-flow on edges incident to $\bx'$ is at most $\Gamma_{G,i}(\bX)$. 
    \item[A'3.] For every $\by' \in \bY'_H$, the total flow on edges incident to $\by'$ is at most $D_G(\bY)$ and the total $i$-flow on edges incident to $\by'$ is at most $\Gamma_{G,i}(\bY)$. 
\end{enumerate}
By integrality of flow, there is an integral flow of at least $\floor{\frac{m}{\log^5 d}}$ units
satisfying the above constraints. By condition (A'1) above, this integral flow is a subgraph of $G_H$ with at least $\floor{\frac{m}{\log^5 d}}$ edges, and satisfying the degree bounds listed in (1c) and (1d) of the lemma statement. Thus, this subgraph satisfies case (1) (``down-persistence case'') of the lemma statement.

\subsection{Case 2: The total amount of heavy flow is at most $\frac{m}{\log^4 d}$}

By \Clm{edge-flow}, the total flow is at least $m(1-\log^{-4}d)$ units. Thus, after removing the heavy flow, the remaining light flow is at least $m(1-2\log^{-4}d)$ units. The light flow satisfies the following capacity constraints. 
\begin{enumerate}
    \item[B1.] Every edge in $G_L = (\bX'_L, \bY'_L, E'_L)$ has at most $1/2$ units of flow.
    \item[B2.] For every $\bx' \in \bX'_L$, the total flow on edges incident to $\bx'$ is at most $D_G(\bX)/2$ and the total $i$-flow on edges incident to $\bx'$ is at most $\Gamma_{G,i}(\bX)/2$.
    \item[B3.] For every $\by' \in \bY'_L$, the total flow on edges incident to $\by'$ is at most $(1+\log^{-3}d)D_G(\bY)$ and the total $i$-flow on edges incident to $\by'$ is at most $(1+\log^{-3}d)\Gamma_{G,i}(\bY)$.
\end{enumerate}
Once again, recall that $\bX$ and $\bY$ were the bipartitions of the {\em original} graph $G$ and not of $G'$, and thus item (B1) above does not imply (B2) or (B3).
Items (B1) and (B2) are simply by \Def{heavy} since all vertices in $\bX_L'$ are not heavy. Item (B3) follows from the RHS bound on the vertex congestion in \Clm{vert-cong}.

We now re-scale the flow by multiplying it by $\frac{2}{1+\log^{-3} d}$. We now have $2m\frac{(1-2\log^{-4}d)}{1+\log^{-3}d} \geq 2m(1-2\log^{-3}d)$ units of flow with the following capacity constraints:
\begin{enumerate}
    \item[B'1.]  Every edge has at most one unit of flow.
    \item[B'2.]  For every $\bx' \in \bX'_L$, the total flow on edges incident to $\bx'$ is at most $D_G(\bX)$ and the total $i$-flow on edges incident to $\bx'$ is at most $\Gamma_{G,i}(\bX)$.
    \item[B'3.]  For every $\by' \in \bY'_L$, the total flow on edges incident to $\by'$ is at most $2D_G(\bY)$ and the total $i$-flow on edges incident to $\by'$ is at most $2\Gamma_{G,i}(\bY)$.
\end{enumerate}

By integrality of flow, we obtain an integral flow of at least $\floor{2m(1-3\log^{-4}d)}$ units satisfying the same constraints listed above. In particular, the flow on any edge is at most one and so the integral flow is a violation subgraph with at least $\floor{2m(1-3\log^{-4}d)}$ edges and satisfying the degree bounds listed in case (2) of the lemma statement.

\section{Proof of \Thm{mono-testing} and \Thm{cont-testing}} \label{sec:proofs_of_main_theorems}

In this section we prove \Thm{mono-testing} and \Thm{cont-testing} using our main result \Cref{thm:main} combined with the domain reduction theorem of \cite{BlackCS20} and the $\widetilde{O}(\eps^{-1}d)$ tester of \cite{DGLRRS99} and \cite{BeRaYa14}. First, observe that by \Cref{thm:main}, repeating \Alg{alg} $\eps^{-2} d^{1/2 + O((\log \log nd)^{-1}) }$ times immediately gives the following corollary. 

\begin{corollary} [Corollary of \Cref{thm:main}] \label{cor:tester_small_n} Let $n \leq \poly(d)$. There is a tester which, given a parameter $\eps \in (0,1)$ where $\eps \geq d^{-1/2}$, and a function $f \colon [n]^d \to \{0,1\}$, makes $\eps^{-2} \cdot d^{1/2 + O((\log \log d)^{-1}) }$ non-adaptive queries to $f$  and (a) accepts when $f$ is monotone, and (b) rejects with probability at least $2/3$ when $f$ is $\eps$-far from monotone. \end{corollary}

We will refer to the tester of \Cref{cor:tester_small_n} as the \pathtester. We will also use the tester of \cite{BeRaYa14} with query-complexity $O(\frac{d}{\eps} \log \frac{d}{\eps})$, which we will refer to as the \linetester. Note that this tester is also non-adaptive and has one-sided error (see \cite{BeRaYa14}, Theorem 1.3). To prove \Thm{mono-testing} using \pathtester~ and \linetester, we use Theorem 1.3 of \cite{BlackCS20} (domain reduction for $[n]^d$), which we state here for ease of reading.


\begin{theorem}[Domain Reduction Theorem 1.3,~\cite{BlackCS20}]\label{thm:bcs20-discrete}
	Let $d$ be at least a sufficiently large constant, and suppose $f \colon [n]^d \to \{0,1\}$ is $\eps$-far from being monotone. If $\bT = T_1 \times \cdots \times T_d$ is a randomly chosen sub-grid, where for each $i \in [d]$, $T_i$ is a (multi)-set formed by taking $\ceil{(\eps^{-1}d)^8}$ independent, uniform samples from $[n]$, then $\EX_{\bT}[\eps_{f_{\bT}}] \geq \eps/2$.\footnote{Note that $f_{\bT}$ denotes the restriction of $f$ to $\bT$, and $\eps_{f_{\bT}}$ denotes its distance to monotonicity.} 
\end{theorem}

Analogously, to prove our \Thm{cont-testing} we use Theorem 1.4 of \cite{BlackCS20} (domain reduction for $\RR^d$).

\begin{theorem} [Domain Reduction Theorem 1.4,~\cite{BlackCS20}] \label{thm:bcs20} Let $d$ be at least a sufficiently large constant. Let $f\colon \mathbb{R}^d \to \{0,1\}$ be any measurable function and let $\cD = \prod_{i=1}^d \cD_i$ be a (Lebesgue integrable) product distribution such that the distance to monotonicity of $f$ w.r.t. $\cD$ is $\eps_f \geq \eps$. If $\bT = T_1 \times \cdots \times T_d$ is a randomly chosen sub-hypergrid, where for each $i \in [d]$, $T_i \subset \mathbb{R}$ is formed by taking $\ceil{(\eps^{-1}d)^8}$ i.i.d. samples from $\cD_i$, then $\EX_{\bT}\left[\eps_{f_{\bT}}\right] \geq \eps/2$.\footnote{Note that $\eps_{f_{\bT}}$ denotes the distance to monotonicity of the restriction $f_{\bT}$ with respect to the uniform distribution over $\bT$.}
\end{theorem}

We can now define our tester \Alg{tester-levin}, using \pathtester~, \linetester~ and the domain reduction theorems. The only difference between our testers for proving \Cref{thm:mono-testing} and \Cref{thm:cont-testing} is that they use \Cref{thm:bcs20-discrete} and \Cref{thm:bcs20}, respectively. Thus, we state them as one tester in \Alg{tester-levin} and prove both theorems together.

\begin{algorithm}
	\caption{Monotonicity tester for $f\colon D^d \to \{0,1\}$ where $D = [n]$ or $D = \RR$. Inputs: $f$ and $\eps \in (0,1)$.} \label{alg:tester-levin}
	\begin{algorithmic}[1]
		\State \textbf{let} $L = \ceil{\log(2/\eps)}$. 
		\State \textbf{for all} $\ell \in [L+1]$:
			\State \indent\label{step:set_eps_ell}\textbf{set} $Q_\ell := \lceil\frac{32\ell^2}{2^\ell \eps}\rceil$ and $\eps_\ell := 1/2^\ell$.
			\State \indent \textbf{repeat} $Q_\ell$ times:
				\State \label{step:domain_reduce_discrete} \indent \indent \textbf{if} $D = [n]$, sample $\bT = T_1 \times \cdots \times T_d$ as in \Thm{bcs20-discrete}.
				\State \label{step:domain_reduce} \indent \indent \textbf{if} $D = \RR$, sample $\bT = T_1 \times \cdots \times T_d$ as in \Thm{bcs20}. 
				\State \label{step:line_tester} \indent \indent \textbf{if} $\eps < d^{-1/2}$, then run \linetester($f_{\bT},\eps_\ell$) and if it rejects, then \Return REJECT.
				\State \label{step:small_grid_tester} \indent \indent \textbf{if} $\eps \geq d^{-1/2}$, then run \pathtester($f_{\bT},\eps_\ell$) and if it rejects, then \Return REJECT.
		\State \Return ACCEPT.
	\end{algorithmic}
\end{algorithm}

\begin{remark} We note that \cite{HY22} obtain a more efficient domain reduction result. However, the domain reduction from \cite{BlackCS20} can be used in a black-box fashion, resulting in a simpler tester.

Our tester (\Alg{tester-levin}) uses Levin's work investment strategy (see \cite{Go-book}, Section 8.2.4) to optimize the dependence on $\eps$. We remark that if one only cares about achieving a dependence of $\poly(1/\eps)$, then the following simpler tester suffices: invoke \Step{domain_reduce_discrete}, \Step{domain_reduce}, \Step{line_tester}, and \Step{small_grid_tester} (with $\eps_\ell$ replaced by $\eps/4$) of \Alg{tester-levin} $16/\eps$ times. By Markov's inequality and the fact that $\EX_{\bT}[\eps_{f_{\bT}}] \geq \eps/2$, with high probability at least one of the calls to \Step{domain_reduce_discrete} or \Step{domain_reduce} will yield a reduced hypergrid $\bT$ satisfying $\eps_{f_{\bT}} \geq \eps/4$. \Step{line_tester} or \Step{small_grid_tester} will then reject the restriction $f_{\bT}$, and thus reject $f$, with high probability. This leads to an $\eps^{-3}$ dependence on $\eps$, as opposed to the $\eps^{-2}$ achieved by \Alg{tester-levin}.

\end{remark}



In \Step{domain_reduce_discrete} or \Step{domain_reduce} of \Alg{tester-levin} we sample a sub-hypergrid $\bT = \prod_{i=1}^d T_i$, where each $T_i$ is of size $\ceil{(\eps^{-1}d)^8}$. By \Thm{bcs20-discrete} or \Thm{bcs20}, $\EX_{\bT}[\eps_{f_{\bT}}] \geq \eps/2$. Now, refer to \Step{set_eps_ell} of \Alg{tester-levin}. We prove \Clm{good-ell} below, which asserts that there exists $\ell^\ast \in [L+1]$ such that $\Pr_{\bT}\left[\eps_{f_{\bT}} \geq \eps_{\ell^\ast}\right] \geq \frac{2^{\ell^{\ast}}\eps}{8 (\ell^{\ast})^2} \geq 4/Q_{\ell^{\ast}}$. Thus when $\ell$ is set to $\ell^{\ast}$ in \Alg{tester-levin} \emph{at least one} of the $Q_{\ell^{\ast}}$ iterations of \Step{domain_reduce_discrete} or \Step{domain_reduce} returns $\bT$ satisfying $\eps_{f_{\bT}} \geq \eps_{\ell^{\ast}}$ with probability $\geq 1 - (1 - 4/Q_{\ell^{\ast}})^{Q_{\ell^{\ast}}} \geq 1 - (1/e)^4 \geq 15/16$. Thus, \Alg{tester-levin} rejects in either \Step{line_tester} or \Step{small_grid_tester} with probability $> \frac{15}{16} \cdot \frac{2}{3} = 5/8$. On the other hand, if $f$ is monotone, then $f_{\bT}$ is always monotone and so \Alg{tester-levin} accepts with probability $1$. (Since the tester has one-sided error, we can boost the rejection probability in the former case to at least $2/3$ by simply repeating the tester twice and rejecting if either iteration rejects. The rejection probability becomes at least $1-(3/8)^2 > 2/3$.)

We now analyze the query complexity. First, suppose $\eps < d^{-1/2}$ and recall that the query complexity of \Step{line_tester} is $O(\frac{d}{\eps_{\ell}} \log \frac{d}{\eps_{\ell}} )$. Thus, the query complexity of \Alg{tester-levin} in this case is at most
\[
\sum_{\ell = 1}^{L+1} Q_{\ell} \cdot O\left(\frac{d}{\eps_{\ell}} \log \frac{d}{\eps_{\ell}} \right) \leq \sum_{\ell = 1}^{L+1} \frac{\ell^2}{2^\ell \eps} \cdot O(2^{\ell} d \cdot \log (2^{\ell} d)) \leq L^3 \cdot O\left(\frac{d}{\eps} \log \frac{d}{\eps}\right)  \leq O\left(\frac{d}{\eps} \log^4 \frac{d}{\eps}\right)
\]
where in the second to last step we used $2^{\ell} \leq 2^{L+1} = O(1/\eps)$ inside the logarithm, and in the last step we simply used $L = O(\log 1/\eps)$. Now, suppose that $d^{-2} < \eps < d^{-1/2}$. Then, using the lower bound on $\eps$, the log-term simplifies to $O(\log^4 d)$. Using the upper bound on $\eps$ yields $d/\eps < \sqrt{d}/\eps^2$. Thus, the query complexity in this case is at most $O(\eps^{-2} \sqrt{d} \log^4 d)$, satisfying the desired bound. On the other hand, if $\eps \leq d^{-2}$, then we have $d \leq \eps^{-1/2}$ and so $d/\eps \leq \eps^{-3/2}$. Moreover, the log-term simplifies to $O(\log^{4} 1/\eps)$ and so the query complexity is bounded by $O(\eps^{-3/2} \log^4 1/\eps) = O(\eps^{-2})$, again satisfying the desired bound. Here we used $\eps^{-1/2} \geq \log^4 1/\eps$ since $\eps^{-1/2} \geq d$, and we assume that $d$ is a sufficiently large constant.

Now, suppose $\eps \geq d^{-1/2}$ and let $q(\eps,n,d)$ denote the query complexity of \pathtester~ with parameters $\eps, n$, and $d$. In particular, 
\[
q(\eps,\ceil{(\eps^{-1}d)^8},d) \leq \eps^{-2} \cdot d^{1/2 + O((\log \log d)^{-1})}
\] 
and so the query complexity of \Alg{tester-levin} in this case is 
\begin{align}
	\sum_{\ell=1}^{L+1} Q_\ell \cdot  q(\eps_{\ell},\ceil{(\eps^{-1}d)^8},d) &= \sum_{\ell=1}^{L+1} \left\lceil\frac{32\ell^2}{2^{\ell}\eps}\right\rceil \cdot 2^{2\ell} \cdot d^{1/2 + O((\log \log d)^{-1})} \nonumber \\
	&\leq \eps^{-1} \cdot d^{1/2 + O((\log \log d)^{-1})} \sum_{\ell=1}^{L+1} \ell^2 \cdot 2^{\ell} \nonumber \\
	&\leq \eps^{-1} \cdot d^{1/2 + O((\log \log d)^{-1})} O(L^3 \cdot 2^L) \leq \eps^{-2} \cdot d^{1/2 + O((\log \log d)^{-1})} \label{eq:query-com-L}
\end{align}


\noindent where in the last step we used the fact that $L = \Theta(\log (1/\eps))$ and $\eps \geq d^{-1/2}$. In particular, these facts imply 
\[
L^3 = O(\log^3 (1/\eps)) = d^{O(\frac{\log \log (1/\eps)}{\log d})} = d^{O(\frac{\log \log d}{\log d})} = d^{o(1/\log\log d)}
\]
and so this factor of $L^3$ is absorbed by the $d^{O((\log \log d)^{-1})}$ term in \Eqn{query-com-L}. \qed



\begin{claim} \label{clm:good-ell} If $\EX_{\bT}[\eps_{f_{\bT}}] \geq \eps / 2$, then there exists $\ell^{\ast} \in [L+1]$ such that $\Pr\left[\eps_{f_{\bT}} \geq 2^{-\ell^{\ast}}\right] \geq \frac{2^{\ell^{\ast}}\eps}{8 (\ell^{\ast})^2}$. \end{claim}

\begin{proof} We have $\int_{0}^1 \Pr\left[\eps_{f_{\bT}} \geq t\right] dt = \EX[\eps_{f_{\bT}}] \geq \eps/2$ and so $\int_{\eps/4}^1 \Pr\left[\eps_{f_{\bT}} \geq t\right] dt \geq \eps/4$. Thus,
\begin{align} \label{eq:levins}
	\frac{\eps}{4} &\leq \int_{\eps/4}^1 \Pr\left[\eps_{f_{\bT}} \geq t\right] dt \leq \sum_{\ell=0}^L \int_{1/2^{\ell+1}}^{1/2^{\ell}} \Pr\left[\eps_{f_{\bT}} \geq t \right] dt \nonumber \\
	&\leq \sum_{\ell=0}^L \frac{1}{2^{\ell+1}} \Pr\left[\eps_{f_{\bT}} \geq 1/2^{\ell+1}\right]
    = \sum_{\ell=1}^{L+1} \frac{1}{2^{\ell}} \Pr\left[\eps_{f_{\bT}} \geq 1/2^{\ell}\right] \text{.}
\end{align}

\noindent For the sake of contradiction, assume $\Pr\left[\eps_{f_{\bT}} \geq 1/2^{\ell}\right] < \frac{2^{\ell}\eps}{8\ell^2}$ for all $\ell \in [L+1]$. Using \Eqn{levins}, we have

\begin{align*}
	\eps \leq 4\sum_{\ell=1}^{L+1} \frac{1}{2^{\ell}} \Pr\left[\eps_{f_{\bT}} \geq 1/2^{\ell}\right]
      < \frac{\eps}{2}\sum_{\ell=1}^{L+1} \frac{1}{\ell^2} < \frac{\eps}{2} \cdot \frac{\pi^2}{6} < \varepsilon \text{.}
\end{align*}

\noindent This is a contradiction. \end{proof}

\section{Acknowledgments}
The authors thank the anonymous reviewers whose detailed comments have greatly improved the presentation of the paper. 

\bibliographystyle{alpha}
\bibliography{../../Bib/monotonicity}

\appendix

\section{Deferred Proofs}

\subsection{Equivalence of the Walk Distributions: Proof of \Fact{dists}} \label{sec:dists}

\begin{proof} 
	Fix a pair $(u,v)$ in $[n]^d$ where $u \preceq v$. We will show that the probability of sampling this pair from each distribution is the same. Let $S = \{i \in [d] \colon v_i > u_i\}$. Note that $u_j = v_j$ for all $j \neq S$. The probability of sampling the pair $(u,v)$ from the distribution described in item (1) of \Fact{dists} is computed as follows.
	
	\begin{align} \label{eq:P}
		\Pr_{\bx \in_R [n]^d \text{, } \by \sim \cU_{\tau}(\bx)} [(\bx, \by) = (u,v)] &&~=~& \frac{1}{n^d} \sum_{R\supseteq S~:~|R|=\tau} {d \choose \tau}^{-1} \prod_{i\in S} \Pr[c_i = v_i ~|~ \bx = u] \prod_{i\in R\setminus S} \Pr[c_i \leq u_i ~|~ \bx = u] \text{.}
	\end{align}
	
	Recall the distribution of $q_i, I_i, c_i$ from \Def{walkdist}. Consider $i \in S$ and let $d_i := \min(v_i-u_i,n-(v_i-u_i))$. Note that conditioned on $q_i$, the total number of intervals $I_i \ni u_i$ is $2^{q_i}$ and the number of such intervals that contain $v_i$ is $\max(0,2^{q_i} - d_i)$. Thus, we have
	
	\begin{align} \label{eq:i-in-S}
		i \in S ~\Longrightarrow~ \Pr[c_i = v_i ~|~ \bx = u] &= \EX_{q_i} \Big[ \Pr_{I_i}[v_i \in I_i] \Pr_{c_i \in I_i}[c_i = v_i ~|~ v_i \in I_i] \Big] \nonumber \\ 
		&= \frac{1}{\log n} \sum_{q \colon 2^{q_i} \geq d_i} \frac{2^{q_i} - d_i}{2^{q_i}} \cdot \frac{1}{2^{q_i}-1} = \frac{1}{2} \cdot \EX_{q_i}\left[\frac{\max(0,2^{q_i}-d_i)}{{2^{q_i} \choose 2}}\right] \text{.}
	\end{align}
	
	For an interval $I_i \ni u_i$, let $I_{i,u_i}$ denote the prefix of $I_i$ preceding (not including) $u_i$. Note that conditioned on an interval $I_i \ni u_i$, the probability of choosing $c_i \leq u_i$ is $|I_{i,u_i}|/(2^{q_i}-1)$. Thus, we have
	
	\begin{align} \label{eq:i-in-RminusS}
		i \in R \setminus S ~\Longrightarrow~ \Pr[c_i \leq u_i ~|~ \bx = u] = \EX_{q_i}\left[\frac{1}{2^{q_i}-1} \cdot \EX_{I_i \ni u_i}[|I_{i,u_i}|]\right]
	\end{align}

	We now compute the probability of sampling $(u,v)$ from the distribution described in item (2) of \Fact{dists}. Recall the distribution of $q_i,I_i,a_i,b_i$ from \Def{hypercube-dist}. For $i \in [d]$, let $\cE_i$ be the event that $a_i = u_i$ or $b_i = u_i$. Note that
	\[
	\Pr[\cE_i] = \EX_{q_i}\left[\Pr_{I_i}[I_i \ni u_i]\Pr_{a_i<b_i \in I_i}[u_i \in \{a_i,b_i\} ~|~ u_i \in I_i]\right] = \EX_{q_i} \left[\frac{2^{q_i}}{n} \cdot \frac{2}{2^{q_i}}\right] = \frac{2}{n}
	\]
	Let $\cE_{u}$ denote the event that $\bx = u$. We have
	\begin{align} \label{eq:E_u}
		\Pr[\cE_{u}] = \prod_{i=1}^d \Pr[\cE_i]\cdot \frac{1}{2^d} = \left(\frac{2}{n}\right)^d \frac{1}{2^d} = \frac{1}{n^d} \text{.}
	\end{align}
	Let $\cE_v$ denote the event that $\by = v$. We have
	\begin{align} \label{eq:E_v|E_u}
		\Pr\left[\cE_{v} ~|~ \cE_{u}\right] = \sum_{R\supseteq S~:~|R|=\tau} {d \choose \tau}^{-1} \prod_{i\in S} \Pr[a_i = u_i ~\textrm{and}~b_i = v_i~|~\cE_u] \cdot \prod_{i\in R\setminus S} \Pr[b_i = u_i~|~\cE_u]
	\end{align}
	Fix an $i\in S$ and recall $d_i := \min(v_i-u_i,n-(v_i-u_i))$. We have
	\[
	\Pr[a_i = u_i ~\textrm{and}~b_i = v_i~|~\cE_u] = \Pr[a_i = u_i ~\textrm{and}~b_i = v_i~|~\cE_i] = \frac{\Pr[a_i = u_i ~\textrm{and}~b_i = v_i]}{\Pr[\cE_i]}
	\]
	where the numerator is
	\[
	\Pr[a_i = u_i ~\textrm{and}~b_i = v_i] = \EX_{q_i}\left[\Pr_{I_i}\big[I_i \supseteq [u_i,v_i]\big] \cdot {2^{q_i} \choose 2}^{-1}\right] = \EX_{q_i}\left[\frac{\max(0,2^{q_i}-d_i)}{n \cdot {2^{q_i} \choose 2}}\right]
	\]
	and so
	\begin{align} \label{eq:2-i-in-S}
		i \in S ~\Longrightarrow~ \Pr[a_i = u_i ~\textrm{and}~b_i = v_i~|~\cE_u] = \frac{1}{2} \cdot \EX_{q_i}\left[\frac{\max(0,2^{q_i}-d_i)}{{2^{q_i} \choose 2}}\right]
	\end{align}
	which is equal to the probability computed in \Eqn{i-in-S}.
	
	Now fix an $i\in R\setminus S$. Recall the definition of $I_{i,u_i}$. We have 
	\[
	\Pr[b_i = u_i~|~\cE_{u}] = \Pr[b_i = u_i~|~\cE_i] = \frac{\Pr[b_i = u_i]}{\Pr[\cE_i]}
	\]
	where
	\begin{align*}
		\Pr[b_i = u_i] = \EX_{q_i} \EX_{I_i} \left[\mathbf{1}(u_i \in I_i) \frac{|I_{u_i}|}{{2^{q_i} \choose 2}}\right] = \EX_{q_i} \left[ \frac{1}{n} \sum_{I_i \ni u_i} |I_{i,u_i}| {2^{q_i} \choose 2}^{-1} \right] = \frac{2}{n} \EX_{q_i} \left[\frac{1}{2^{{q_i}-1}} \cdot \EX_{I_i \ni u_i} [|I_{i,u_i}|] \right]
	\end{align*}
	and so recalling that $\Pr[\cE_i] = 2/n$ we have
	\begin{align} \label{eq:2-i-in-RminusS}
		i \notin R \setminus S ~\Longrightarrow~ \Pr[b_i = u_i~|~\cE_{u}] = \EX_{q_i} \left[\frac{1}{2^{{q_i}-1}} \cdot \EX_{I_i \ni u_i} [|I_{i,u_i}|] \right] 
	\end{align}
	which is equal to the probability computed in \Eqn{i-in-RminusS}. Combining \Eqn{P}, \Eqn{i-in-S}, \Eqn{i-in-RminusS}, \Eqn{E_u}, \Eqn{E_v|E_u}, \Eqn{2-i-in-S}, \Eqn{2-i-in-RminusS}, we have
	\[
	\Pr_{\bH \sim \HH}\Pr_{\bx \in_R \bH \text{, } \by \sim \cU_{\bH,\tau}(\bx)}[(\bx,\by) = (u,v)] = \Pr[\cE_u] \cdot \Pr[\cE_v ~|~ \cE_u] = \Pr_{\bx \in_R [n]^d \text{, } \by \sim \cU_{\tau}(\bx)}[(\bx,\by) = (u,v)]
	\]
	and this proves that (1) and (2) of \Fact{dists} are equivalent.
	
	To show equivalence of (1) and (3), note that we only need to show that
	\begin{align} \label{eq:1-3-equiv}
		\Pr_{\bH \sim \HH(u) \text{, } \by \sim \cU_{\bH,\tau}(u)}[\by = v] = \Pr_{\by \sim \cU_{\tau}(u)}[\by = v]
	\end{align}
	This is proven by an analogous calculation. The expression for $\Pr_{\by \sim \cU_{\tau}(u)}[\by = v]$ is given by dropping the $\frac{1}{n^d}$ factor from \Eqn{P} and then plugging in the expressions obtained in \Eqn{i-in-S} and \Eqn{i-in-RminusS}. The quantity $\Pr_{\bH \sim \HH(u) \text{, } \by \sim \cU_{\bH,\tau}(u)}[\by = v]$ is precisely $\Pr[\cE_v ~|~ \cE_u]$, and an expression for this is obtained by \Eqn{E_v|E_u} and plugging in the expressions obtained in \Eqn{2-i-in-S} and \Eqn{2-i-in-RminusS}. Thus, (1) and (3) are equivalent and this completes the proof. \end{proof}

\subsection{Influence and Persistence Proofs}\label{sec:app:inf-per}

\totaltonegativeinf*


\begin{proof} Theorem 9.1 of \cite{KMS15} asserts that for any $\bH$, if $I_{\bH} > 6\sqrt{d}$, then $I^-_{\bH} > I_{\bH}/3$. (This holds for any Boolean hypercube function.) If $\widetilde{I}_f > 9 \sqrt{d}$, then
	by \Clm{inf-hypercube}, $\EX_{\bH}[I_{\bH}] > 9\sqrt{d}$. Hence,
	\begin{eqnarray*}
		9\sqrt{d} < \EX_{\bH}[I_{\bH}] & = & \Pr[I_{\bH} \leq 6\sqrt{d}] \ \EX_{\bH}[I_{\bH} | I_{\bH} \leq 6\sqrt{d}] \ + \ \Pr[I_{\bH} > 6\sqrt{d}] \ \EX_{\bH}[I_{\bH} | I_{\bH} > 6\sqrt{d}] \nonumber \\
		& < & 6\sqrt{d} + \Pr[I_{\bH} > 6\sqrt{d}] \EX_H[3 I^-_{\bH} | I_{\bH} > 6\sqrt{d}] \leq 6\sqrt{d} + 3\EX_{\bH}[I^-_{\bH}]
	\end{eqnarray*}
	Hence, $\EX_{\bH}[I^-_{\bH}] > \sqrt{d}$. By \Clm{inf-hypercube}, $\widetilde{I}^-_f > \sqrt{d}$. \end{proof}


\persistence*

\begin{proof} We will analyze the random walk using the distributions described in the first and second bullet point of \Fact{dists} and leverage the analysis that \cite{KMS15} use to prove their Lemma 9.3. Let $\alpha_{up}$ denote the fraction of vertices in the fully augmented hypergrid that are not $(\tau,\beta)$-up-persistent. 
	Using the definition of persistence and \Fact{dists}, we have
	\begin{align} \label{eq:persistence1}
		\alpha_{up} \cdot \beta < \Pr_{\bx \in_R [n]^d\text{, } \by \sim \cU_{\tau}(\bx)}\left[f(\bx) \neq f(\bz)\right] = \EX_{\bH \sim \HH}\left[\Pr_{\bx \in_R \bH \text{, } \by \sim \cU_{\bH,\tau}(\bx)}\left[f(\bx) \neq f(\bz)\right]\right] \text{.}
	\end{align}
	Let $\widehat{\cU}_{\bH,\tau}(\bx)$ denote the same distribution as $\cU_{\bH,\tau}(\bx)$ except with the set $R$ being a uar subset of the $0$-coordinates of $\bx$. I.e. $\widehat{\cU}_{\bH,\tau}(\bx)$ is the \emph{non-lazy} walk distribution on $\bH$. Let $\bx = \bx^0, \bx^1, \ldots, \bx^{\tau} = \bz$ be the $\tau$ steps taken on the walk sampled by $\cU_{\bH,\tau}(\bx)$ and let $\bx = \widehat{\bx}^0, \widehat{\bx}^1, \ldots, \widehat{\bx}^{\tau} = \bz$ be the $\tau$ steps taken on the walk sampled by $\widehat{\cU}_{\bH,\tau}(\bx)$. For a fixed $\bH$ we have
	\begin{align} \label{eq:persistence2}
		\Pr_{\bx \in_R \bH \text{, } \by \sim \cU_{\bH,\tau}(\bx)}\left[f(\bx) \neq f(\bz)\right] \leq \sum_{\ell=0}^{\tau - 1} \Pr\left[f(\bx^{\ell}) \neq f(\bx^{\ell+1})\right] \leq \sum_{\ell=0}^{\tau - 1} \Pr\left[f(\widehat{\bx}^{\ell}) \neq f(\widehat{\bx}^{\ell+1})\right] \text{.}
	\end{align}
	The first inequality is by a union bound and the second inequality holds because the first walk is lazy and the second is not. More precisely, we can couple the $\tau' \leq \tau$ steps 
	of the lazy-random walk where the point actually moves to the first $\tau'$ steps of the second non-lazy walk, and the remaining $\tau-\tau'$ terms of the non-lazy walk can only increase the RHS.
	
	By Lemma 9.4 of \cite{KMS15}, the edge $(\widehat{\bx}^{\ell},\widehat{\bx}^{\ell+1})$ is distributed approximately as a uniform random edge in $\bH$. In particular, this implies $\Pr\left[f(\widehat{\bx}^{\ell}) \neq f(\widehat{\bx}^{\ell+1})\right] \leq C \cdot 2I_{\bH}/d$ for an absolute constant $C$. (Note $2I_{\bH}/d$ is the probability of a uniform random edge in $\bH$ being influential.) Putting \Eqn{persistence1} and \Eqn{persistence2} together yields $\alpha_{up} \leq \frac{4 C \tau }{\beta d}\EX_{\bH}[I_{\bH}]$ and an analogous argument gives the same bound for $\alpha_{down}$. Thus, by \Clm{inf-hypercube} we have $\EX_{\bH}[I_{\bH}] \leq 9\sqrt{d}$ and the fraction of $(\tau,\beta)$-non-persistent vertices is at most $\frac{72C\tau}{\beta\sqrt{d}}$. Therefore, setting $\cper := 72C$ completes the proof. \end{proof}


\subsection{Typical Points and Reversibility Proofs} \label{sec:app:reversible}


\hypmiddle*

\begin{proof} Consider a uniform random point $\bx$ in the hypercube. The Hamming weight $\|\bx\|_1$ is $\sum_{i=1}^d \bx_i$, where each $\bx_i$ is an iid unbiased Bernoulli. By Hoeffding's theorem, 
$\Pr[\Big|\|\bx\|_1 - d/2\Big| \geq t] \leq 2\exp(-2t^2/d)$. We set
$t = \sqrt{4cd\log(d/\eps)}$. The probability of not being 
in  the $c$-middle layers 
is at most
$$2\exp(-2t^2/d) = 2\exp(-8c\log(d/\eps)) = 2(\eps/d)^{8c} \leq (\eps/d)^c \text{.}$$
Hence, the probability of being in the $c$-middle layers is at least $(1-(\eps/d)^c)$.
\end{proof}


\typical*

\begin{proof} Given $\bx \in [n]^d$ and a hypercube $\bH \ni \bx$, let $\chi(\bx,\bH) = \mathbf{1}(\bx \in \bH \setminus \bH_c)$. By \Fact{dists} and \Clm{hyp-middle}, we have 
	\[
	\EX_{\bx \in_R [n]^d}\EX_{\bH \sim \HH(\bx)} \left[\chi(\bx,\bH)\right] = \EX_{\bH \sim \HH} \EX_{\bx \in_R \bH} \left[\chi(\bx,\bH)\right] \leq (\eps/d)^c
	\]
	Let us set $q_\bx \eqdef \EX_{\bH \sim \HH(\bx)} [\chi(\bx,\bH)]$, so $\EX_{\bx} [q_\bx] \leq (\eps/d)^{c}$. By Markov's inequality, $\Pr_\bx[q_\bx \geq (\eps/d)^{5}] \leq (\eps/d)^{c-5}$. Note that when $q_{\bx} < (\eps/d)^5$, $\bx$ is $c$-typical. Hence, at least a $(1-(\eps/d)^{c-5})$-fraction of points are $c$-typical. \end{proof}



\translatetypical*

\begin{proof} We prove the claim for $\bx' \in \text{supp}(\cU_{\tau}(\bx))$. The argument for points in $\text{supp}(\cD_{\tau}(\bx))$ is analogous. Let $\bH$ be any hypercube containing $\bx$ and $\bx'$ and let $\norm{\bx}_{\bH}$, $\norm{\bx'}_{\bH}$ denote the Hamming weight of these points in $\bH$. Observe that $\norm{\bx'}_{\bH} \leq \norm{\bx}_{\bH} + \tau$ and so if $\bx \in \bH_c$, then $\norm{\bx'}_{\bH} \leq d/2 + \sqrt{4cd \log (d/\eps)} + \tau$ and since $\tau \leq \sqrt{d}$, we have
\[
\sqrt{4cd \log (d/\eps)} + \tau \leq  \sqrt{4cd\log (d/\eps) + \tau\sqrt{d}\log d} = \sqrt{\left(c+\frac{\tau}{\sqrt{d}}\right)d\log d}\text{.}
\]
To see that the first inequality, observe that
\[
\tau^2 + 2\tau\sqrt{4cd\log (d/\eps)} \leq \tau\sqrt{d}\log d ~\Longleftrightarrow~ \tau \leq \sqrt{d}(\log d - 4\sqrt{c \log (d/\eps)})
\]
which clearly holds by our upper bound on $\tau$. Thus, if $\bx \in \bH_c$, then $\bx' \in \bH_{c+\frac{\tau}{\sqrt{d}}}$. Therefore, the number of hypercubes $\bH$ for which $\bx' \in \bH_{c+\frac{\tau}{\sqrt{d}}}$ is at least the number of hypercubes $\bH$ for which $\bx \in \bH_c$. 
Therefore $\bx'$ is $(c+\frac{\tau}{\sqrt{d}})$-typical. \end{proof}


\probrev*

\begin{proof} If $t \eqdef \|\bx - \bx'\|_0 > \ell$, then $p_{\bx,\ell}(\bx') = p_{\bx',\ell}(\bx) = 0$. So assume $t \leq \ell$. Fix any $\bH$ containing $\bx$ and $\bx'$ such that $\bx,\bx' \in \bH_{100}$ and let $x$ and $x'$ denote the corresponding hypercube (bit) representations of $\bx,\bx'$ in $\bH$. Let $p_{x,\ell}(x') = \Pr_{z \sim \cU_{\bH,\ell}(x)}[z=x']$ and $p_{x',\ell}(x) = \Pr_{z \sim \cD_{\bH,\ell}(x')}[z=x]$. By definition of $p_{\bx,\ell}(\bx')$ (recall \Cref{def:walkpdf}) it suffices to show that $p_{x,\ell}(x') = (1\pm \log^{-3}d)p_{x',\ell}(x)$.
	
	Let $S$ be the set of $t$ coordinates
	where $x$ and $x'$ differ. Let $Z(x)$ be the set of zero coordinates of the point $x$;
	analogously, define $Z(x')$.
	Recall that the directed upward walk making $\ell$ steps might not flip $\ell$ coordinates.
	The process (recall \Def{hypercube-walk}) picks a uar set $R$ of $\ell$ coordinates, and only flips the zero bits in $x$ among $R$.
	Hence, an $\ell$-length walk leads from $x$ to $x'$ iff $R \cap Z(x) = S$.
	
	Let the Hamming weight of $x$ be represented as $d/2 + e_x$, where $e_x$ denotes the ``excess".
	Since $x$ is in the $100$-middle layers, $|e_x| \leq \sqrt{400d\log(d/\eps)}$.
	The sets $R$ that lead from $x$ to $x'$ can be constructed by picking any $\ell-t$ coordinates
	in $\overline{Z(x)}$ and choosing all remaining coordinates to be $S$. Hence,
	$$ p_{x,\ell}(x') = \frac{{d/2 + e_x \choose \ell-t}}{{d \choose \ell}}$$
	
	\noindent
	Analogously, consider the downward $\ell$ step walks from $x'$.
	This walk leads to $x$ iff $R \cap \overline{Z(x')} = S$. The sets $R$ that lead from $x'$ to $x$
	can be constructed by picking any $\ell-t$ coordinates in $Z(x')$ and choosing
	all remaining coordinates to be $S$. The size of $Z(x')$ is precisely $|Z(x)| - t = d/2 - e_x - t$.
	Hence,
	$$ p_{x',\ell}(x) = \frac{{d/2 - e_x - t \choose \ell - t}}{{d \choose \ell}} $$
	
	Taking the ratio,
	\begin{eqnarray*}
		\frac{p_{x,\ell}(x')}{p_{x',\ell}(x)} & = & \frac{{d/2 + e_x \choose \ell-t}}{{d/2 - e_x - t \choose \ell - t}}
		= \frac{\prod_{i=0}^{\ell-t-1} (d/2 + e_x - i)}{\prod_{i=0}^{\ell-t-1}(d/2 - e_x - t-i)}
		= \prod_{i=0}^{\ell-t-1} \frac{d/2 + e_x - i}{d/2-e_x-t-i}  \\
		& = & \prod_{i=0}^{\ell-t-1} \Big( 1 + \frac{2e_x + t}{d/2 - e_x-t-i}\Big) 
	\end{eqnarray*}
	Recall that $|e_x| \leq \sqrt{400d\log(d/\eps)}$, $t \leq \ell < \sqrt{d}/\log^5(d/\eps)$.
	For convenience, let $b \eqdef \sqrt{400d\log(d/\eps)}$. So $2e_x + t \leq 3b$.
	Also, $d/2 - e_x -t-i \geq d/3$ for all $i < \ell$. Applying these bounds,
	\begin{eqnarray*}
		\frac{p_{x,\ell}(x')}{p_{x',\ell}(x)}  
		& \leq & \prod_{i=0}^{\ell-1} \Big( 1 + \frac{3b}{d/3}\Big) \leq \exp\Big(\frac{9\ell b}{d}\Big)
		= \exp\Big(\frac{\sqrt{d} \cdot \sqrt{400d\log(d/\eps)}}{d\log^5(d/\eps)}\Big) \leq 1 + \log^{-3}d
	\end{eqnarray*}
	An analogous calculation proves that $ \frac{p_{x,\ell}(x')}{p_{x',\ell}(x)} \geq 1 - \log^{-3} d$. \end{proof}

\end{document}